\documentclass[submission, Phys]{SciPost}

\usepackage[english]{babel}

\usepackage{amsmath,amssymb,amsthm,bm}
\usepackage{graphicx}
\usepackage[colorinlistoftodos]{todonotes}
\usepackage{hyperref}
\usepackage{diagbox}
\usepackage{epigraph}
\usepackage{float}
\usepackage{multirow}
\usepackage{tikz}
\usepackage{braket}
\usepackage[all]{xy}
\usepackage{bbm}
\usepackage{makecell}
\usepackage{dsfont}
\usepackage{ulem}
\usepackage{comment}
\usepackage{rotating}
\usepackage{mathtools}

\newcommand{\xc}[1]{{\color{brown}{[(XC) #1]}}}

\newcommand{\Rep}{{\rm Rep}}
\newcommand{\VEC}{{\rm Vec}}

\newcommand\IC{\mathbb{C}}

\newcommand\IZ{\mathbb{Z}}

\newcommand\CA{\mathcal{A}}

\newcommand\CC{\mathcal{C}}
\newcommand\CD{\mathcal{D}}

\newcommand\CM{\mathcal{M}}
\newcommand\CN{\mathcal{N}}

\newcommand\CP{\mathcal{P}}

\newcommand\CZ{\mathcal{Z}}

\newcommand{\tabref}[1]{Tab.\,\ref{#1}}

\newcommand{\appref}[1]{Appendix\,\ref{#1}}
\newcommand{\grp}[1]{\langle  #1 \rangle}
\newcommand{\TY}{\mathrm{TY}}
\newcommand{\od}{\text{od}}
\newcommand{\dd}{\text{d}}
\newcommand{\kk}{\mathds{C}}
\newcommand{\ii}{\mathsf{i}}

\newcommand{\Noe}{\mathcal{N}_{oe}}
\newcommand{\hcenter}{{\CZ(\Rep(H_8))}}

\newtheorem{physicstheorem}{Physics Theorem}

\newcommand{\Rom}[1]{\uppercase\expandafter{\romannumeral#1}}

\newcommand{\ex}[1]{\left\langle #1 \right\rangle}

\newcommand{\mc}{\mathcal}

\newtheorem{lemma}{Lemma}
\newtheorem{corollary}{Corollary}

\usepackage[comma,numbers,sort&compress]{natbib}
\newcommand{\onlinecite}[1]{\nocite{#1}\citenum{#1}}

\begin{document}

\begin{center}{\Large \textbf{
Spontaneous breaking of non-invertible symmetries and
duality to beyond-Landau transitions
}}\end{center}

\begin{center}
Xie Chen\textsuperscript{1,2},
Shang Liu\textsuperscript{3,2},
Da-Chuan Lu\textsuperscript{4,5*},
Nathanan Tantivasadakarn\textsuperscript{6}
\end{center}

\begin{center}
{\bf 1.} Institute for Quantum Information and Matter, \\
California Institute of Technology, Pasadena, CA, 91125, USA
\\
{\bf 2.} Department of Physics,
California Institute of Technology, Pasadena, CA, 91125, USA
\\
{\bf 3.} Institute of Physics, Chinese Academy of Sciences, Beijing, 100190, China
\\
{\bf 4.} Department of Physics, Harvard University, Cambridge, MA, 02138, USA
\\
{\bf 5.} Department of Physics and Center for Theory of Quantum Matter,\\
University of Colorado, Boulder, CO, 80309, USA
\\
{\bf 6.} C.N. Yang Institute for Theoretical Physics,\\
Stony Brook University, Stony Brook, New York, 11794, USA
\\
* dclu137@gmail.com
\end{center}

\begin{center}
\today
\end{center}


\section*{Abstract}
{Spontaneous symmetry breaking is a well-understood mechanism for generating distinct phases of matter. Recently, the notion of symmetry has been broadened to include operations without inverses, leading to the concept of non-invertible symmetries. How do symmetry-breaking phases associated with non-invertible symmetries differ from those arising from conventional invertible symmetries? We address this question using concrete lattice models of the gapped phases with non-invertible Rep($H_8$) symmetry as an example. We find that, despite the symmetry being non-invertible, the symmetry-breaking phases can still be characterized by the long-range correlation of local order parameters, which obey a more general algebraic structure than in the invertible setting. Furthermore, via generalized gauging, certain non-invertible symmetry-breaking transitions can be mapped to deconfined quantum critical points of invertible symmetries, and vice versa. We establish precise conditions under which this duality holds and illustrate them with several families of examples, providing a systematic route to studying beyond-Landau phase transitions.
}

\tableofcontents
\section{Introduction}

The study of phase and phase transitions is the central topic in condensed matter physics. Landau’s paradigm of symmetry breaking provides a systematic approach to describe transitions when the symmetry of the system is spontaneously broken\cite{Landau1937}. It is now well understood that the applicability of this paradigm is limited, and a variety of ‘beyond-Landau’ transitions have been discovered and analyzed in recent decades. Is there a more general framework that encompasses both within-Landau and beyond-Landau types of transitions? 

Recently, a generalized Landau paradigm has emerged\cite{Gaiotto2015, Hofman2019, Delacretaz2020, Iqbal2020, Moradi2022topological,Bhardwaj2024,McGreevy2023,Chen2025}, which maps beyond-Landau transitions to symmetry-breaking transitions by making use of the notions of generalized symmetry\cite{Gaiotto2015,McGreevy2023,Shao2024TASI,SchaferNameki2024} and generalized gauging\cite{Fuchs2002,Carqueville2016,aasen2016topological,yuji2017gauging,Bhardwaj2018,Thorngren2019Fusion,aasen2020topological,Lootens2023,Lootens2024dualities,Haegeman2015,Kong2020,Gaiotto2021,Chatterjee2023holographic,Huang2023topological,kong2025higher,Bhardwaj2025,seifnashri2025gauging}. In particular, it has been shown that in $1+1$-dimensions, various beyond-Landau transitions, such as the deconfined critical point between incompatible symmetry-breaking phases and the transition between different symmetry-protected topological phases, can all be mapped to symmetry-breaking transitions, but with potentially non-invertible symmetries. Such mappings generalize the well-known examples of the Kramers-Wannier duality\cite{Kramers1941}, the Jordan-Wigner transformation\cite{Jordan1928} and the Kennedy-Tasaki transformation\cite{Kennedy1992a,Kennedy1992b} and gives rise to the notion of generalized gauging. Therefore, if we can understand symmetry-breaking phases and symmetry-breaking phase transitions of non-invertible symmetries, we can achieve a more systematic understanding of transitions between gapped phases in $1+1$D. 

Noninvertible symmetries in $1+1$D contain many interesting examples: 1. The Kramers-Wannier duality\cite{Kramers1941} and the Kennedy-Tasaki transformation\cite{Kennedy1992a,Kennedy1992b} generate the simplest representatives of the Tambara-Yamagami fusion category symmetries\cite{Tambara1998}, which are non-invertible $\IZ_2$-extensions of abelian invertible symmetries. 2. The $\Rep(G)$ symmetry, where $G$ is a non-abelian finite group, can be obtained by gauging the non-abelian symmetry group in a spin chain \cite{Bhardwaj2018,Thorngren2019Fusion,Fechisin:2023odt,Chatterjee:2024ych,bhardwaj2025lattice,Chung:2025ulc}. 3. The Fibonacci fusion category and the Ising fusion category play an important role in the minimal models of CFT \cite{Oshikawa1997,Petkova2001,Frohlich2004,Chang2019} and find explicit realization in anyon chains\cite{Feiguin2007,Buican:2017rxc,Ning:2023qbv}. 4. More generally, the non-invertible symmetry generated by the topological defect lines in rational CFTs \cite{Fuchs2002,Frohlich2004,frohlichDefectLinesDualities2010,Chang2019}. Ref. \cite{Thorngren2019Fusion, Thorngren:2021yso, Fechisin:2023odt,Bhardwaj:2024qrf,Bhardwaj2025, Inamura:2021szw} studied the gapped phases under non-invertible symmetries. 

In this paper, we study in more detail how non-invertible symmetries are spontaneously broken in $1+1$d systems. Symmetry-breaking of noninvertible symmetries is similar to that of invertible symmetries in many ways, but also differs in many others. Invertible symmetries form groups, which are broken down to subgroups when an order parameter, labeled by an irreducible representation of the group, fluctuates and attains a nonzero expectation value. The symmetry broken phase is characterized by the long-range correlation of the order parameter and a degenerate ground space, where the short range correlated symmetry broken states correspond to the cosets of the subgroup and are mapped into each other by the broken symmetries. Noninvertible symmetries, on the other hand, form a fusion category. When a non-invertible symmetry is broken, to what extent does our previous understanding of the invertible case generalize? Do non-invertible symmetries break into sub-categories? How to label the symmetry broken states and how do they transform under the broken symmetries? Can we still identify the symmetry breaking pattern using the correlation function of local order parameters? It turns out the notion of `target manifold' or the number of order parameter configurations can be straight-forwardly generalized while the notion of remaining symmetry is more tricky to define. 

Using the Symmetry Topological Field Theory (SymTFT) formalism\cite{Kong2015,Ji2020,Kong2020,Bhardwaj2020,Pulmann2021,Gaiotto2021,Lichtman2021,kong2020mathematical,Kong2020classification,Apruzzi2023,Chatterjee2023symmetry,Moradi2022topological,Freed2023topological,Lin2023,Kong2018,Kong2021,Kong2022one,Kong2022categories,Kong2024categories,Xu2024,Bhardwaj2025}, the gapped phases with non-invertible symmetry can be understood as a quasi-$1+1$D sandwich structure with a $2+1$D topological order in the bulk and gapped boundaries on the two sides. The algebra of the ground states in the sandwich structure and the operators acting on them is given in Ref.~\cite{Cong2017}. In this paper, we illustrate the structure of this algebra using the example of symmetry breaking phases with Rep($H_8$) symmetry, where $H_8$ is the smallest non-commutative and non-cocommutative self-dual Hopf algebra. Through explicit lattice model, we demonstrate that
\begin{itemize}
\item The long-range correlation of local order parameters can still be used to identify the symmetry breaking order even though the order parameters can transform in a non-local way under the symmetry.
\item Local order parameters of a symmetry breaking phase correspond to possible `tunneling channels' between the two boundaries in the sandwich structure and satisfy the algebra of the `tunneling' operators.
\item When the symmetry is fully broken, the short-range correlated symmetry broken states can be labeled by the simple objects in the fusion category of the symmetry. Their transformations under the symmetry satisfy the same fusion rule as the symmetry operators.
\item When the symmetry is partially broken, there is no obvious labeling for the short-range correlated symmetry broken states. But it is always possible to find a `cat-state' basis as a superposition of the symmetry-broken states, such that the `cat-states' are labeled by the tunneling channels and transform accordingly.
\item It is not always obvious what the remaining symmetry is. In partial symmetry broken phases, it is possible that certain linear combination of the symmetries keeps all the symmetry-broken states invariant, while individual symmetries in the combination do not.
\end{itemize}

Note that in this paper, when we discuss order parameters, we consider only \textit{local} order parameters. This corresponds to the `untwisted' order parameters discussed in, for example, Ref.~\cite{Bhardwaj2024}.

$\Rep(H_8)$ is an interesting non-invertible symmetry: it is equivalent to a particular $\IZ_2 \times \IZ_2$ Kramers-Wannier symmetry, but unlike the ordinary $\IZ_2$ Kramers-Wannier symmetry, it is anomaly-free and hence admits a symmetric gapped phase. There are six gapped phases with Rep($H_8$) symmetry, and we can study the transitions between them. Not all transitions between the pairs of the gapped phases can be mapped to `within Landau'
symmetry-breaking transitions of regular invertible symmetries through generalized gauging. This feature makes it more interesting than other simple non-invertible symmetries like Rep($S_3$) and Rep($D_8$).
Instead, the Rep($H_8$) symmetry can be partially gauged to yield a dual anomalous $D_8$ symmetry and the transitions between some of the Rep($H_8$) gapped phases can be mapped to Deconfined-Quantum-Critical-Point (DQCP) transitions between symmetry-breaking phases of the $D_8$ symmetry. Gapped phases and critical theories with the Rep($H_8$) symmetry have also been studied in Ref.~\cite{Thorngren2019Fusion,Thorngren:2021yso,Choi2024,wang2024gauging,Perez-Lona:2023djo,Meng:2024nxx,Chen:2025ivo,Li:2025wes,lu2026generalized}. Besides the $\Rep(H_8)$ example, we generalize this correspondence and show that for all group-theoretical Hopf algebras, and spontaneous symmetry breaking transitions of such anomaly-free non-invertible symmetries can be gauged into transitions of ordinary group symmetries \footnote{In particular all the anomaly-free non-invertible symmetries with Frobenius-Perron dimension less than $36$ are group-theoretical.}, which become DQCPs when the associated $3$-cocycle is non-trivial. 

This paper is organized as follows. In Section~\ref{sec:SymTFT}, we discuss the SymTFT realization of the non-invertible symmetry $\Rep(H_8)$ and describe the bulk topological order of the SymTFT sandwich $\CZ(\Rep(H_8))$ using two copies of the doubled Ising after the anyon condensation $1 \oplus \psi_1\overline\psi_1 \psi_2\overline\psi_2$. The bulk topological order admits several equivalent descriptions including the twisted quantum double $D(D_8)^\gamma$. In Section~\ref{sec:Lattice}, we provide explicit lattice realizations of all the gapped phases with $\Rep(H_8)$ symmetry and analyze their spontaneous symmetry breaking patterns, ground-state degeneracies, and their characterization using local order parameters. In Section~\ref{sec:Dual}, we gauge an anomaly-free subgroup of $\Rep(H_8)$ and obtain the dual lattice model with $D_8$ symmetry carrying a ``mixed'' ’t~Hooft anomaly $\gamma$, thereby establishing the duality between the order-to-disorder transition of $\Rep(H_8)$ and a deconfined quantum critical point of $(D_8,\gamma)$. In Section~\ref{sec:General}, we use the duality framework to discuss the general relation between 1+1D DQCPs with anomalous group symmetry $(G,\omega)$ and order-to-disorder transitions of anomaly-free (possibly non-invertible) group-theoretical symmetries, and give various examples. The appendices collect the mathematical details of the twisted bicrossed product construction of Hopf algebra (Appendix~\ref{app:twbiprod}), the full categorical data for Rep($H_8$) (Appendix~\ref{app:grpthryH8}), the Drinfeld center related to the $\Rep(H_8)$ symmetry (Appendix~\ref{app:Drinfeld_center}), the Lagrangian Algebra describing the gapped boundaries of the bulk topological order $\CZ(\text{Rep}(H_8))$ (Appendix~\ref{app:LagrangianAlgebra}), and the rigorous derivation of the dual lattice model symmetry (Appendix~\ref{app:LatticeDualSymmetry}). 

\section{Symmetry TFT Description of Rep($H_8$) symmetry}
\label{sec:SymTFT}

In this section, we describe the realization of the Rep($H_8$) symmetry in the Symmetry TFT formalism. In the symmetry TFT formalism, a $1+1$D system is realized as a sandwich structure with a $2+1$D topological bulk, as shown in Fig.~\ref{fig:symTFT}. The top boundary is set to be in a gapped state through the condensation of certain bulk anyons. The bottom boundary is left open to host the dynamics of the system. With a finite distance between the top and bottom boundary, the sandwich reduces to a $1+1$D system. The advantage of having the sandwich structure is that the dynamics at the bottom boundary is spatially separated from the action of the symmetry. The string operators that run parallel to the boundary ($W_{\alpha}$ associated with the anyon $\alpha$) become the ($0$-form) symmetry of the $1+1$D system. The ones that correspond to anyons condensed on the top boundary take fixed values while the others become true symmetries of the $1+1$D systems that act non-trivially. Tunneling of condensed anyons out of the top boundary toggles the symmetry sectors. Therefore, the vertical string operators ($V_{\beta}$ associated with the anyon $\beta$) corresponds to local charged operators under the symmetry. 

\begin{figure}[th]
\begin{center}
\includegraphics[width= 0.6\textwidth]{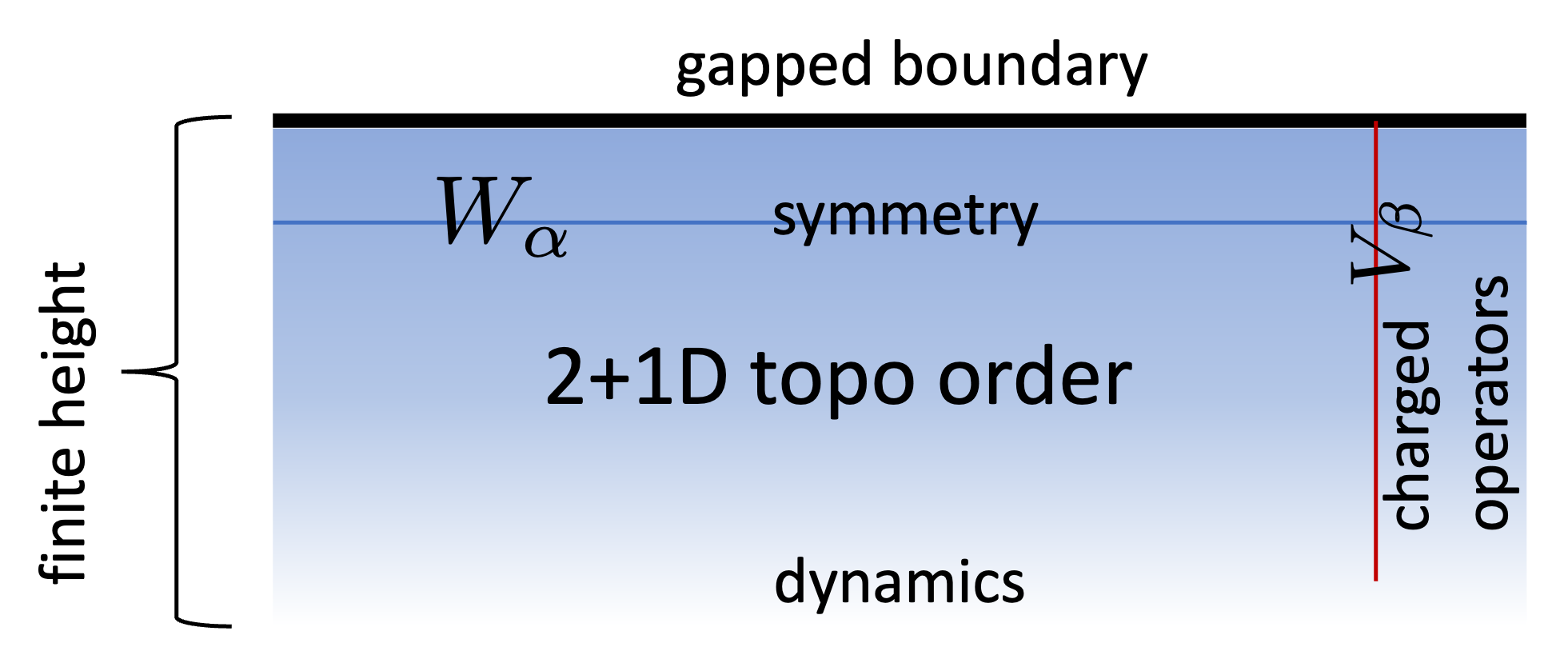}
\caption{In the Symmetry Topological Field Theory, a $1+1$D system is realized in a sandwich structure with a $2+1$D topological order in the bulk. The top boundary is gapped, fixed, and determines the symmetry of the system. The bottom boundary contains all the dynamics of the system and is at a finite distance from the top boundary. The horizontal string operators $W_{\alpha}$ become the symmetry operators of the $1+1$D system. The vertical string operators $V_{\beta}$ become local charged operators if $\beta$ is condensed on the top boundary.} 
\label{fig:symTFT}
\end{center}
\end{figure}

Using the Symmetry TFT formalism, we can specify the Rep($H_8$) symmetry without writing down all the categorical data, but instead as a descendant of a $2+1$D topological order. This realization also naturally leads to the 1D lattice formulation used in section~\ref{sec:Lattice}.

\subsection{$2+1$D Topological Bulk}

The $2+1$D topological bulk can be obtained from two copies of the doubled Ising topological order through a simple condensation. A single copy of the chiral Ising topological order contains two nontrivial anyons $\psi$ and $\sigma$ satisfying the fusion rule 
$$
\psi\otimes \psi = 1, \ \psi\otimes\sigma = \sigma, \ \sigma\otimes \sigma = 1 \oplus \psi
$$ 
The anyon content of the doubled Ising topological order has a tensor product structure $\{1,\psi,\sigma\}\otimes \{1,\bar{\psi},\bar{\sigma}\}$, where $\bar{\psi}$ and $\bar{\sigma}$ are the time reversal copies of $\psi$ and $\sigma$. The anyons in two copies of the doubled Ising topological order come from the tensor product of four parts
$$
\{1,\psi_1,\sigma_1\}\otimes \{1,\bar{\psi}_1,\bar{\sigma}_1\}\otimes \{1,\psi_2,\sigma_2\}\otimes \{1,\bar{\psi}_2,\bar{\sigma}_2\}
$$
The topological order that leads to the Rep($H_8$) symmetry can be obtained from two copies of the doubled Ising topological order by condensing $\psi_1\bar{\psi}_1\psi_2\bar{\psi}_2$. Since all the $\psi$ anyons are fermions, the composite $\psi_1\bar{\psi}_1\psi_2\bar{\psi}_2$ is a boson and can be condensed. After the condensation, the topological state contains 22 types of deconfined anyons. 8 of them are abelian and come from the composition of the $\psi$ anyons
$$
1, \psi_1, \bar{\psi}_1, \psi_2, \bar{\psi}_2, \psi_1\bar{\psi_1}, \psi_1\psi_2, \psi_1\bar{\psi}_2
$$
Single $\sigma$ anyons are confined due to the condensation of $\psi_1\bar{\psi}_1\psi_2\bar{\psi}_2$, while the composite of two $\sigma$ anyons survive as dimension 2 nonabelian excitations in the new topological state. There are 12 of them.
$$
\sigma_1\sigma_2, \sigma_1\sigma_2\bar{\psi}_1, \bar{\sigma}_1\bar{\sigma}_2, 
\bar{\sigma}_1\bar{\sigma}_2\psi_1,
\sigma_i\bar{\sigma}_i,
\sigma_i\bar{\sigma}_i\psi_j,
\sigma_i\bar{\sigma}_j, \sigma_i\bar{\sigma}_j\bar{\psi}_i\  (i,j=1,2, i\neq j) 
$$

Finally, the composite of four $\sigma$ anyons splits into two nonabelian anyons $A$ and $B$, each with dimension 2 and satisfying the fusion rule
$$
A \otimes A = B \otimes B = 1 \oplus \psi_1\bar{\psi}_1 \oplus \psi_1\psi_2 \oplus \psi_1\bar{\psi}_2, \ A \otimes B = \psi_1 \oplus \bar{\psi}_1 \oplus \psi_2 \oplus \bar{\psi}_2
$$
It is helpful to notice several automorphisms of the anyon set. The fusion and braiding data of the anyons remain invariant if $\{1, \psi_1, \sigma_1\}$ is exchanged with $\{1, \psi_2, \sigma_2\}$, if $\{1, \bar{\psi}_1, \bar{\sigma}_1\}$ is exchanged with $\{1, \bar{\psi}_2, \bar{\sigma}_2\}$, and if $A$ is exchanged with $B$. This forms the automorphism group $\mathbb Z_2^3$ of the bulk topological order. In the category language, the bulk is the `center' of the boundary which gives rise to the Rep($H_8$) symmetry. Therefore, the bulk topological order can be labeled as $\CZ(\Rep(H_8))$.

\subsection{Gapped boundaries and associated symmetries}
\label{sec:boundary_sym}

Among all the anyons in the topological state, the following ones are bosonic
$$
1, \psi_1\bar{\psi_1}, \psi_1\psi_2, \psi_1\bar{\psi}_2, \sigma_1\bar{\sigma}_1, \sigma_1\bar{\sigma}_2, \sigma_2\bar{\sigma}_1, \sigma_2\bar{\sigma}_2, A, B
$$
and can be potentially condensed and result in a gapped boundary. One gapped boundary can be obtained by condensing the Lagrangian algebra of
$$
\mathcal{A}_1: 1 \oplus \psi_1\bar{\psi}_1 \oplus \sigma_1\bar{\sigma}_1 \oplus \sigma_2\bar{\sigma}_2 \oplus A
$$
Due to the condensation, the horizontal string operator of $\psi_1\bar{\psi}_1$, $ \sigma_1\bar{\sigma}_1$, $\sigma_2\bar{\sigma}_2$, $A$ are fixed to specific values. The remaining (independent) string operators of $\psi_1$, $\psi_2$, $\psi_1\psi_2$, $\sigma_1\sigma_2$ (together with that of $1$) give rise to the Rep($H_8$) symmetry. The horizontal string operators $\{W_{1}, W_{\psi_1}, W_{\psi_2}, W_{\psi_1\psi_2}\}$ are invertible and form a $\mathbb Z_2\times \mathbb Z_2$ group. $W_{\sigma_1\sigma_2}$ is a non-invertible symmetry operator with the nonabelian fusion rule
$$
W_{\psi_i}\times W_{\sigma_1\sigma_2} = W_{\sigma_1\sigma_2}, \ \ W_{\sigma_1\sigma_2}\times W_{\sigma_1\sigma_2} = W_1+W_{\psi_1}+W_{\psi_2}+W_{\psi_1\psi_2}
$$

Using the automorphism among the bulk anyons, we can find three more Lagrangian algebras with the same structure:
$$
\begin{array}{c}
\mathcal{A}_2: 1 \oplus \psi_1\bar{\psi}_1 \oplus \sigma_1\bar{\sigma}_1 \oplus \sigma_2\bar{\sigma}_2 \oplus B; \\
\mathcal{A}_3: 1 \oplus \psi_1\bar{\psi}_2 \oplus \sigma_1\bar{\sigma}_2 \oplus \sigma_2\bar{\sigma}_1 \oplus A; \\
\mathcal{A}_4: 1 \oplus \psi_1\bar{\psi}_2 \oplus \sigma_1\bar{\sigma}_2 \oplus \sigma_2\bar{\sigma}_1 \oplus B
\end{array}
$$
Boundaries described by $\mathcal{A}_1$ to $\mathcal{A}_4$
all result in a non-invertible Rep($H_8$) symmetry. 

The bulk topological order has a different type of gapped boundary described by Lagrangian algebra
$$
\begin{array}{c}
\mathcal{A}_5: 1\oplus\psi_1\bar{\psi}_1 \oplus \psi_1\bar{\psi}_2 \oplus \psi_1\psi_2 \oplus 2A \\ 
\mathcal{A}_6: 1\oplus\psi_1\bar{\psi}_1 \oplus \psi_1\bar{\psi}_2 \oplus \psi_1\psi_2 \oplus 2B
\end{array}
$$
These boundaries correspond to an invertible $D_8$ symmetry but with an anomaly. To see this, we analyze the independent string operators that correspond to nontrivial symmetries for boundary $\mathcal{A}_5$ are $W_{\psi_1}$, $W_{\sigma_1\bar{\sigma}_1}$, $W_{\sigma_1\sigma_2}$, and $W_{\bar{\sigma}_1\sigma_2}$ (the analysis for $\mathcal{A}_6$ is analogous.). Since all fermion pairs are condensed on the boundary, there is a single representative $W_{\psi_1} \sim W_{\psi_2}\sim W_{\bar \psi_1}\sim W_{\bar\psi_2} \equiv W_\psi$, which is an order $2$ invertible symmetry. The others turn out to be non-simple and hence split near the boundary. For example,
$$
W_{\sigma_1\bar{\sigma}_1} \times W_{\sigma_1\bar{\sigma}_1} = W_1 + W_{\psi_1} + W_{\bar{\psi}_1} + W_{\psi_1\bar{\psi}_1} 
$$
Due to the condensation of $\psi_1\bar{\psi}_1$, there are two copies of identity in the fusion of $W_{\sigma_1\bar{\sigma}_1}$ with itself. Therefore, $W_{\sigma_1\bar{\sigma}_1}$ splits into two parts $W_{e}$ and $W_{m}$. Analogous to anyons of the toric code. Each is invertible with order two and the pair fuse into $W_{\psi_1}$.
\begin{equation}
W_{e} \times W_{e} = W_{m} \times W_{m} = 1, \ W_{e} \times W_{m} = W_{\psi}
\label{eq:W11bar}
\end{equation}
Similarly, $W_{\bar \sigma_1\sigma_2}$ splits into $W_{e'}$ and $W_{m'}$. Now, let us consider the splitting of  $\sigma_1\sigma_2$, which has topological spin $e^{i\pi/4}$. Let us call the splitting $W_d$ and $W_{d'}$. These operators are not order two because they square to the fermion $\psi$:

\begin{equation}
W_{d}\times W_{d} = W_{d'}\times W_{d'} = W_{\psi}, W_{d} \times W_{d'} = 1
\label{eq:W12}
\end{equation}
Thus, we see that $W_{d}$ generates a $\mathbb Z_4$ subgroup. These anyons behave like the anyons of $U(1)_4$. However, note that the toric code and $U(1)_4$ share the same fermion.

We have identified eight invertible lines, so our symmetry must be a group of order eight. To see the group structure, we notice that the product of $e$ and $d$ must come from the splitting from $\bar{\sigma}_1\sigma_2$, which we choose to be $e'$. However, $e'$ also has order two. Thus, we conclude that the symmetry group must be $D_8$.

To show that $D_8$ is anomalous, we analyze the gauging of these symmetries. Here, we focus on the analysis of the self-anomaly of Abelian subgroups. The non-trivial spin does not immediately signify an anomaly. This is because in the sandwich, we are only gauging the symmetry along a codimension-1 manifold rather than the full spacetime. Such higher gauging of a $\mathbb Z_N$ 1-form symmetry, is only possible if $\theta^N=1$ where $\theta$ is its corresponding topological spin~\cite{Roumpedakis:2022aik}.  The $\mathbb Z_4$ subgroup generated by $d$ is not gaugeable because $\theta_d^4 = -1$. Thus, this $\mathbb Z_4$ subgroup is anomalous. However, its $\mathbb Z_2$ subgroup generated by $\psi$ is gaugeable. In fact, all the $\mathbb Z_2$ subgroups (generated by $\psi,e,m,e',m'$) are gaugeable.

\begin{table}[]
    \centering
    \begin{tabular}{|c|c|c|c|}
    \hline
    Before splitting & after splitting & $D_8$ & $\theta$\\
    \hline
    1 & 1 & 1 & 1\\
    \hline
    $\psi_1\sim \psi_2$ & $\psi$ & $r^2$ & $-1$\\
        \hline
       \multirow{2}{*}{$\sigma_1\bar \sigma_{1}$}& $e$ & $s$ & 1 \\
       \cline{2-4}
       & $m$ & $sr^2$ & 1\\
       \hline
    \multirow{2}{*}{$\bar{\sigma}_1\sigma_2$}& $e'$ & $sr$ & 1\\
    \cline{2-4}
       & $m'$ & $sr^3$ & 1 \\
       \hline
        \multirow{2}{*}{$\sigma_1\sigma_2$}& $d$ & $r$ & $e^{\pi i/4}$\\
        \cline{2-4}
       & $d'$ & $r^3$ & $e^{\pi i/4}$\\
         \hline
    \end{tabular}
    \caption{The anomalous $D_8$ symmetry can be thought of as the fusion category  $\frac{\text{TC}\boxtimes U(1)_4}{ 1 \oplus\psi\psi'}$, where the fermion $\psi$ of the toric code and the fermion $\psi'$ in $U(1)_4$ is identified.}
    \label{tab:placeholder}
\end{table}

\subsection{Equivalent twisted quantum double description}
Since the boundaries $\CA_5$ and $\CA_6$ realize an anomalous $D_8$ symmetry, the bulk topological order $\hcenter$ must admit a description as a twisted quantum double. Indeed, $\hcenter \cong \mathcal{D}(D_8)^\gamma$, where $\gamma \in H^3(D_8, U(1))$. We parameterize $D_8$ using its standard generators and relations:
\begin{equation}
    D_8=\langle r,s| r^4=s^2=1,srs=r^{-1} \rangle
\end{equation}
The specific 3-cocycle twist is explicitly given as \cite{ostrik2002module,natale2017equivalence,Lu:2025gpt}:
\begin{equation}
    \gamma(r^{i_1}s^{j_1},r^{i_2}s^{j_2},r^{i_3}s^{j_3}) = \exp\left(\frac{4\pi i}{4^2}(-1)^{j_1}i_1(i_2+(-1)^{j_2}i_3-[i_2+(-1)^{j_2}i_3]_4)\right)
\end{equation}
where $[\cdots]_4$ denotes the expression modulo 4. Indeed, setting $j_1=j_2=j_3=0$, the cocycle reduces to the standard 3-cocycle representative for the $\mathbb Z_4$ subgroup $\nu=2$.

The anyons in the twisted quantum double are labeled by $[a, \chi^{C_G(a)}]$, where $a$ is a representative of a conjugacy class, and $\chi^{C_G(a)}$ is an irreducible projective representation of the centralizer $C_G(a)$ corresponding to the following 2-cocycle called the slant product of $\gamma$:
$$\beta_a(h, g) = \gamma(a, h, g) \gamma(h, h^{-1}ah, g)^{-1} \gamma(h, g, (hg)^{-1}ahg)$$
However, in this case, the cocycle turns out to be a coboundary \cite{Coste:2000tq}. This means that for each $a \in D_8$, there exists a 1-cochain $\epsilon_a$ on $C_G(a)$ such that $\beta_a(h, g)$ can be expressed as:
$$\beta_a(h, g) = (\delta\epsilon_a)(h, g) = \epsilon_a(h) \epsilon_a(g) \epsilon_a(hg)^{-1}$$
Thus, all the resulting representations are actually linear representations. Consequently, while the 3-cocycle twist $\gamma$ leaves the fusion rules of the anyons invariant, it modifies their spins and braiding statistics, manifesting as changes in the $S$ and $T$ matrices \cite{Coste:2000tq}. It is straightforward to calculate the anyon data. Up to anyon permutation symmetry, the identifications are listed in Table \ref{tab:anyon_id}. In particular, the Lagrangian algebras in the last section can be expressed as
\begin{align*}
\CA_1 &= [1,\,\chi_{1}^{D_8}]
\oplus [1,\,\chi_{3}^{D_8}]
\oplus [s,\,\chi_{0,0}^{\mathbb{Z}_2\times\mathbb{Z}_2}]
\oplus [s,\,\chi_{0,1}^{\mathbb{Z}_2\times\mathbb{Z}_2}]
\oplus [r^2,\,\chi_{5}^{D_8}] \\[6pt]
\CA_2 &= [1,\,\chi_{1}^{D_8}]
\oplus [1,\,\chi_{3}^{D_8}]
\oplus [s,\,\chi_{0,0}^{\mathbb{Z}_2\times\mathbb{Z}_2}]
\oplus [s,\,\chi_{0,1}^{\mathbb{Z}_2\times\mathbb{Z}_2}]
\oplus [1,\,\chi_{5}^{D_8}] \\[6pt]
\CA_3 &= [1,\,\chi_{1}^{D_8}]
\oplus [1,\,\chi_{4}^{D_8}]
\oplus [rs,\,\chi_{1,0}^{\mathbb{Z}_2\times\mathbb{Z}_2}]
\oplus [rs,\,\chi_{1,1}^{\mathbb{Z}_2\times\mathbb{Z}_2}]
\oplus [r^2,\,\chi_{5}^{D_8}] \\[6pt]
\CA_4 &= [1,\,\chi_{1}^{D_8}]
\oplus [1,\,\chi_{4}^{D_8}]
\oplus [rs,\,\chi_{1,0}^{\mathbb{Z}_2\times\mathbb{Z}_2}]
\oplus [rs,\,\chi_{1,1}^{\mathbb{Z}_2\times\mathbb{Z}_2}]
\oplus [1,\,\chi_{5}^{D_8}] \\[6pt]
\CA_5 &= [1,\,\chi_{1}^{D_8}]
\oplus [1,\,\chi_{2}^{D_8}]
\oplus [1,\,\chi_{3}^{D_8}]
\oplus [1,\,\chi_{4}^{D_8}]
\oplus 2\,[r^2,\,\chi_{5}^{D_8}] \\[6pt]
\CA_6 &= [1,\,\chi_{1}^{D_8}]
\oplus [1,\,\chi_{2}^{D_8}]
\oplus [1,\,\chi_{3}^{D_8}]
\oplus [1,\,\chi_{4}^{D_8}]
\oplus 2\,[1,\,\chi_{5}^{D_8}]
\end{align*}
where $\chi^G_i$ is the $i$-th character of the group $G$. For $D_8$, it contains four 1-dimensional irreducible representations labeled by $\chi^{D_8}_i, i=1\sim4$, and $\chi^{D_8}_5$ is the character of the 2-dimensional irrep. For $\IZ_2\times \IZ_2$, the character is defined as $\chi_{r,s}^{\mathbb{Z}_2\times\mathbb{Z}_2}(a^i b^j)=(-1)^{i r+j s}$ where $a,b$ are the generators for each $\IZ_2$ and $r,s\in \{0,1\}$ label the different $\IZ_2\times \IZ_2$ characters. It is also clear that $\CA_5,\CA_6$ form $\Rep(D_8)$ subcategory, and condensing which on the symmetry boundary will lead to $(D_8,\gamma)$ symmetry.

We conclude this section by noting that the TFT bulk of $\Rep(H_8)$ can also be equivalently described by the center of $\IZ_2$-graded fusion category \cite{Gelaki:2009blp,Lu:2025gpt} as $\Rep(H_8)$ is equivalent to the Tambara-Yamagami fusion category which has a natural $\IZ_2$-grading. Using this description, it is straightforward to view $\hcenter$ as the result of simultaneously gauging the $e$-$m$ permutation symmetry in two copies of the $\IZ_2$ toric code \cite{Zhang:2023wlu}. To sum up, the $\hcenter$ is described as,
\begin{enumerate}
    \item Anyon condensation of $(\mathrm{Ising}\boxtimes \overline{\mathrm{Ising}})^{\boxtimes2}$ by the condensible algebra $1\oplus \psi_1 \overline\psi_1\psi_2\overline\psi_2$.
    \item Twisted quantum double of $D_8$ with 3-cocycle $\gamma$, $D(D_8)^\gamma$. 
    \item Center of $\IZ_2$-graded fusion category, $\CZ(\TY(\IZ_2\times \IZ_2,\chi_\text{diag},+1))$ \cite{Gelaki:2009blp,Lu:2025gpt}.
    \item Or directly from the Hopf algebra $H_8$, the category of Yetter–Drinfeld modules over $H_8$ \cite{hu2007beta,shi2016finite}.
\end{enumerate}
The non-invertible symmetry $\mathrm{Rep}(H_8)$ can be realized by placing the bulk in any aforementioned descriptions of $\hcenter$ and choosing one of the Lagrangian algebras $\CA_{1}$–$\CA_{4}$ as the symmetry boundary.

\section{\texorpdfstring{$\Rep(H_8)$}{RepH8} Symmetry in a Qubit Lattice Model}
\label{sec:Lattice}

Since we want to study systems with Rep($H_8$) symmetry, we will fix the top boundary to be in the $\mathcal{A}_1$ state. Different gapped boundaries at the bottom corresponds to different gapped phases under Rep($H_8$) symmetry. We see that there are six gapped phases in total. The overlap between the Lagrangian algebra $\mathcal{A}_1$ at the top and $\mathcal{A}_i$ at the bottom gives the ground state degeneracy of the gapped phase on a 1D ring. We will study these gapped phases in detail in this section.

The Symmetry TFT construction naturally leads to the 1D lattice realization discussed in this section. It is well known that the Kramers-Wannier symmetry $\{1,\eta, \CN\}$ with fusion rule
$$
\eta \times \eta = 1, \eta \times \CN = \CN, \CN \times \CN = 1 + \eta
\label{eq:Ising}
$$
can be realized in a 1D qubit Ising chain where the invertible $\IZ_2$ symmetry $\eta$ is the product of spin flip $\eta = \prod_i X_i$ and the non-invertible Kramers-Wannier symmetry $\CN$ can be realized as a sequential circuit (together with a projection) that maps between the $\IZ_2$ symmetric state $|++...+\rangle$ and the $\IZ_2$ symmetry breaking state $|00...0\rangle, |11...1\rangle$.
\begin{align}
\CN|00...0\rangle = \CN|11...1\rangle = |++...+\rangle, \ \CN|++...+\rangle = |00...0\rangle + |11...1\rangle
\end{align}
On the other hand, the Kramers-Wannier symmetry can be realized in a sandwich structure with a single copy of the doubled Ising topological order in the bulk. Therefore, it is natural to expect that Rep($H_8$), whose symmetry TFT bulk comes from two copies of the doubled Ising topological order, can be realized in 1D with two qubit chains. We will give explicit lattice formulation of the Hamiltonian of each gapped phase in the double qubit chain and label them $H_{1i}$, for $i = 1,...,6$. The Hamiltonian $H_{1i}$ are constructed to be frustration free. That is, their ground states can be found exactly. 

\begin{figure}[th]
\begin{center}
\includegraphics[width= 0.8\textwidth]{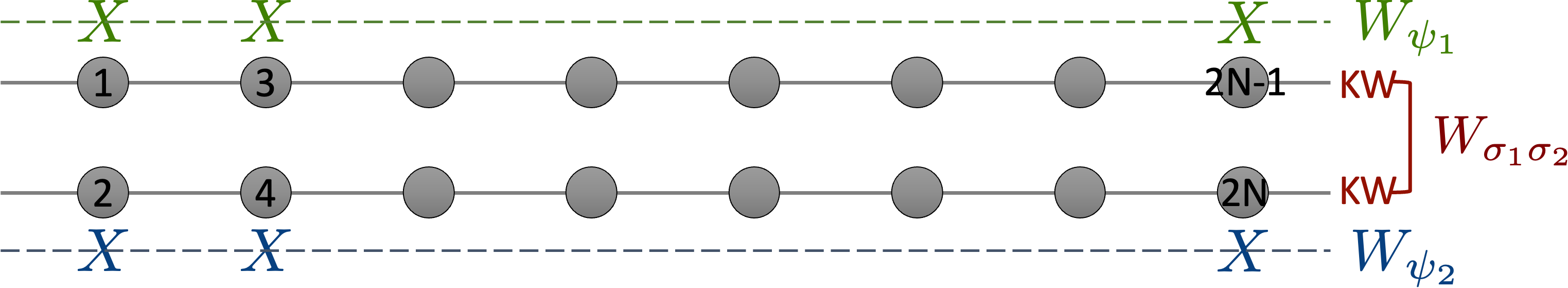}
\caption{Rep($H_8$) symmetry realized in two qubits chains. $W_{\psi_1}$, the spin flip operation in the first chain, and $W_{\psi_2}$, the spin flip operation in the qubit chain, generate the invertible part of the symmetry. $W_{\sigma_1\sigma_2}$, the joint Kramers-Wannier transformation in both chains, gives the non-invertible symmetry operation.} 
\label{fig:lattice}
\end{center}
\end{figure}

Consider two qubit chains each with $N$ qubits, as shown in Fig.~\ref{fig:lattice}. We assign odd number labels to qubits in the first chain and even number labels to the ones in the second chain. 
A lattice version of $\Rep(H_8)$ symmetry can be generated by $W_{\psi_1} = \eta_o = \prod_k X_{2k-1}$, $W_{\psi_2} = \eta_e = \prod_k X_{2k}$, and $W_{\sigma_1\sigma_2}=\CN_o\CN_e$, where $\CN_o$ ($\CN_e$) is the Kramers-Wannier duality on odd (even) sites. The overall action of $\CN_o\CN_e$ on local operators symmetric under $W_{\psi_1}$ and $W_{\psi_2}$ is then
\begin{align}
    \CN_o\CN_e (X_n,Z_nZ_{n+2})=(Z_n Z_{n+2},X_{n+2})\CN_o\CN_e. 
\end{align}
If we assume periodic boundary condition, $\CN_o^2\CN_e^2=T^2(1+\eta_o+\eta_e+\eta_o\eta_e)$ where the translation symmetry generator $T$ acts as $T(X_n,Z_n)T^{-1}=(X_{n+1},Z_{n+1})$. 
The Kramers-Wannier transformation on one of the qubit chains is represented by for example the $W_{\sigma_1}$ string operator in the topological bulk of the sandwich structure. It is confined due to the $\psi_1\bar{\psi}_1\psi_2\bar{\psi}_2$ condensation, but if we were able to apply $W_{\sigma_1}$, we would be able to exchange $A$ and $B$. Therefore, $W_{\sigma_1}$ represents the $\IZ_2$ automorphism that exchanges $A$ and $B$ and as a result also exchanges boundaries $\mathcal{A}_{2i-1}$ and $\mathcal{A}_{2i}$, for $i = 1,2,3$. In the lattice model, we can apply the Kramers-Wannier transformation on one of the qubit chains and map between  
\[
\xymatrix{(H_{11},H_{13},H_{15})\ar@/^/[r]^{\CN_o} & (H_{12},H_{14},H_{16})\ar@/^/[l]^{\CN_e}}. 
\]

We use this property to construct some of the Hamiltonians below. 
The Hamiltonian are all constructed to be frustration free. Some of the $\Rep(H_8)$ symmetric gapped phases are constructed on the qubit lattice in \cite{Chen:2025ivo}.

One main feature we want to illustrate using the lattice model is the algebra of local order parameters. In the sandwich structure, local order parameters are string operators that tunnel between the top and bottom boundaries. Their fusion algebra is given in Ref.~\cite{Cong2017} and illustrated in Fig.~\ref{fig:OrderPara}. A gapped boundary with Lagrangian Algebra $\mathcal{A}$ is characterized by a set of $M$-$3j$ symbols. When two gapped boundaries are put together to form a sandwich, the fusion algebra of two order parameter $V$'s is given by two set of $M$-$3j$ symbols, one from each gapped boundary $\mathcal{A}_1$ and $\mathcal{A}_2$, as shown in Fig.~\ref{fig:OrderPara}. Moreover, acting the order parameters on the `no-tunneling' state gives a basis for the degenerate ground space of the symmetry breaking phase. This corresponds to what is usually called the `cat' state basis of global superpositions of short range correlated symmetry broken states. Using this basis, Fig.~\ref{fig:OrderPara} can also be interpreted as the action of the order parameters on each basis `cat' state.

\begin{figure}[th]
\begin{center}
\includegraphics[width= 0.6\textwidth]{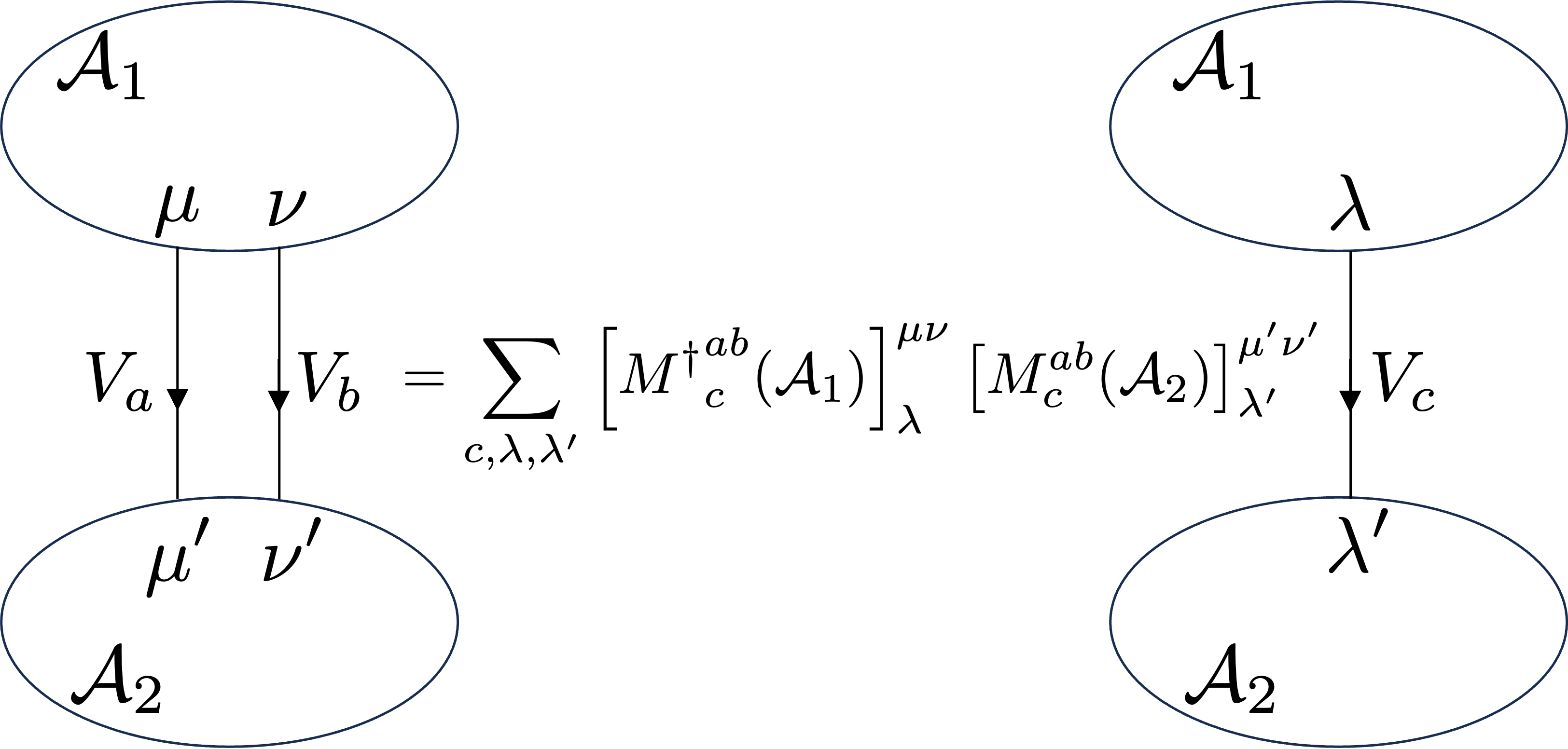}
\caption{Fusion algebra of two local order parameters $V_a$ and $V_b$ into $V_c$ in the SymTFT formalism. An equivalent interpretation is the action of a local order parameter $V_a$ on a `cat' state labeled by $b$. The $M$-$3j$ symbols come from the Lagrangian algebra $\mathcal{A}$'s describing the top and bottom gapped boundaries in the sandwich structure.} 
\label{fig:OrderPara}
\end{center}
\end{figure}

We provide a `physical' derivation of the $M$-$3j$ symbols of the Lagrangian algebra $\mathcal{A}_1$ to $\mathcal{A}_6$ in appendix~\ref{app:LagrangianAlgebra}, making use of the notion of `condensable dimensions' of the anyons condensed on the boundaries.

\subsection{\texorpdfstring{$\Rep(H_8)$}{RepH8} fully symmetry breaking phase $H_{11}$}

When the top and bottom boundary in the sandwich construction are both described by the Lagrangian algebra $\mathcal{A}_1 = 1 \oplus \psi_1\bar{\psi}_1 \oplus \sigma_1\bar{\sigma}_1 \oplus \sigma_2\bar{\sigma}_2 \oplus A$, all the anyons condensed on the top boundary can tunnel to the bottom boundary. Therefore, the Rep($H_8$) symmetry is fully broken, resulting in a five-fold ground state degeneracy. A Hamiltonian realization of this phase in the double qubit chain is given by 
\begin{align}
    H_{11}&=\sum_k\bigg[ \left(\frac{1-Z_{2k-1}Z_{2k+1}}{2}\right)\left( 
\frac{1-X_{2k}}{2}+\frac{1-X_{2k+2}}{2} \right)\nonumber\\
&\quad\quad  + \left(\frac{1-Z_{2k}Z_{2k+2}}{2}\right)\left( 
\frac{1-X_{2k-1}}{2}+\frac{1-X_{2k+1}}{2} \right)\bigg]. 
\end{align}

The Hamiltonian is gapped with a five-fold degenerate ground state. The five short-range correlated symmetry broken ground states take a very simple form
\begin{align}
\ket{++++\cdots},\quad\ket{xyxy\cdots}
\end{align}
with $x,y\in\{0,1\}$. Note that these states are orthogonal in the thermodynamic limit.

It is straight-forward to verify that these five states are degenerate ground states of $H_{11}$. $H_{11}$ is a sum over projectors, so it's lowest possible energy is zero. The five states listed above are all zero energy eigenstates, hence the ground states of $H_{11}$. We checked numerically that $H_{11}$ does not have other ground states. Moreover, as the system size $2N$ increases, the energy gap saturates to a nonzero value as shown in Fig.\,\ref{fig:H11Gap}.
\begin{figure}
    \centering
    \includegraphics[width=0.5\textwidth]{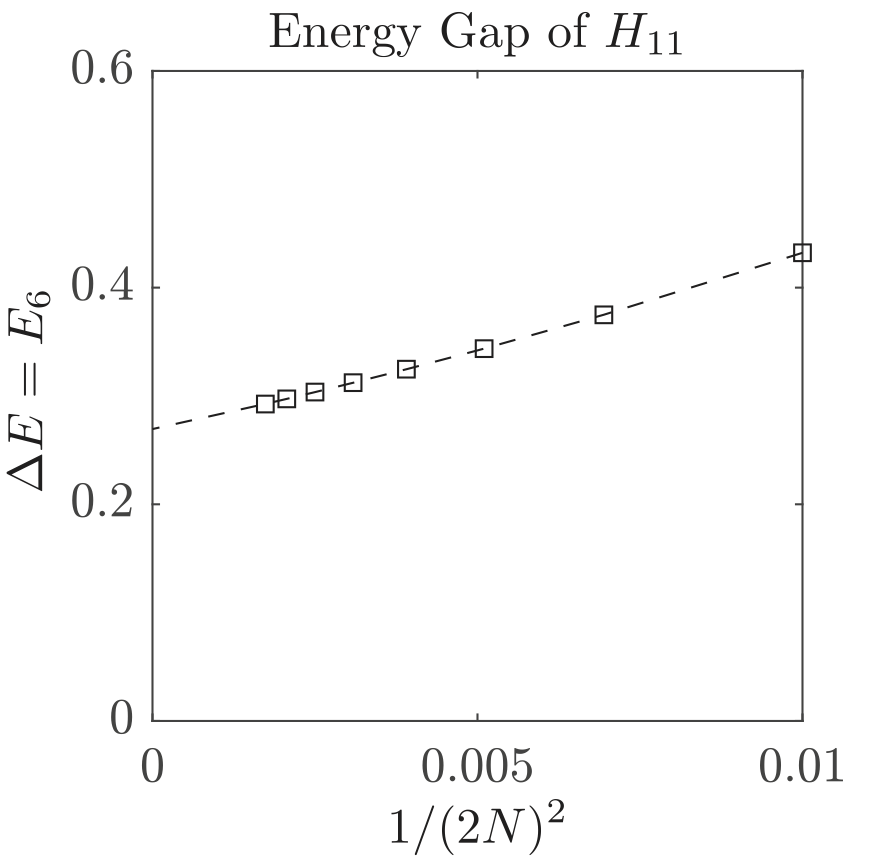}
    \caption{The energy gap $\Delta E$ of $H_{11}$ as a function of $1/(2N)^2$. $\Delta E$ is the same as the 6-th lowest energy eigenvalue. The squares are numerical data from exact diagonalization, with $2N$ goes from $10$ to $24$. The dashed line is a quadratic fitting, i.e. $\Delta E_{\rm fit}=a + b/(2N) + c/(2N)^2$. The intercept $a=\Delta E_{\rm fit}(N\rightarrow \infty)$ is clearly nonzero. }
    \label{fig:H11Gap}
\end{figure}

The five short-range correlated symmetry-broken ground states can be labeled by the different symmetry operators in Rep($H_8$). We assign the labeling as
\begin{equation}
\begin{array}{l}
|\phi_1\rangle =  |0000...00\rangle, \ |\phi_{\eta_o}\rangle =  |1010...10\rangle, \ |\phi_{\eta_e}\rangle =  |0101...01\rangle, \\ 
|\phi_{\eta_o\eta_e}\rangle = |1111...11\rangle, \ |\phi_{\CN}\rangle = |++++...++\rangle 
\end{array}
\end{equation}

The labeling is assigned such that when the symmetry operators act on these states, the result is the same as the symmetry fusion rule. For example,
\begin{equation}
\begin{array}{l}
\eta_o |\phi_1\rangle = \prod_{k} X_{2k-1} |0000...00\rangle = |1010...10\rangle = |\phi_{\eta_o}\rangle \\
\CN_o\CN_e |\phi_{\eta_o}\rangle = \CN_o\CN_e |1010...10\rangle = |++++...++\rangle = |\phi_{\CN}\rangle \\
\CN_o\CN_e |\phi_{\CN}\rangle = \CN_o\CN_e |++++...++\rangle = \sum_{x,y=0,1}|xyxy...xy\rangle = |\phi_1\rangle + |\phi_{\eta_o}\rangle + |\phi_{\eta_e}\rangle + |\phi_{\eta_o\eta_e}\rangle
\end{array}
\end{equation}
Note that the same pattern holds for the symmetry breaking ground states of the Kramers-Wannier symmetry in Eq.~\ref{eq:Ising} if we assign labels to the symmetry broken ground states as 
$$
|\phi_1\rangle = |00...0\rangle, \ |\phi_{\eta}\rangle = |11...1\rangle, \ |\phi_{\CN}\rangle = |++...+\rangle
$$
This is a generic feature of the short range correlated ground states when a categorical symmetry is completely broken \cite{kitaev2012models,Thorngren2019Fusion}. When conventional invertible symmetries are completely broken, the symmetry broken states are labeled by different group elements. Applying a nontrivial symmetry transformation to a symmetry broken state maps it to a different symmetry broken state. When a categorical symmetry is completely broken, the symmetry broken states are labeled by simple objects in the category, but applying a symmetry does not necessarily map one symmetry broken state to a different symmetry broken state. This relation still holds for at least one symmetry broken state -- the one labeled by the identity operator in the category. That is, applying a symmetry operation to the short range correlated symmetry broken state labeled by the identity symmetry operator always maps it to another symmetry broken state (labeled by the applied symmetry operator). In this sense, we can still say without ambiguity that the symmetry is completely broken. For the other symmetry broken states, the action of a symmetry could take them to a superposition of several symmetry broken states.

What are the local order parameters of this symmetry breaking phase? Since the symmetry broken states are product states supported on (approximately) orthogonal local dimensions, local operators that takes different values on these dimensions can be used as local order parameters. One may notice that while the local dimensions $|xy\rangle$, $x,y=0,1$ are strictly orthogonal to each other, they are not orthogonal to $|++\rangle$. This problem can be mitigated by taking local dimensions on a longer segment. For example, on a segment of length four, the local dimensions become $\ket{xyxy}$, $x,y = 0,1$ and $\ket{++++}$. The orthogonality gets exponentially better with the length of the segment. Let's denote these local projectors by $P_{xy}$, $x,y=0,1$ and $P_{++}$ with the implicit assumption that the projectors acts on a segment long enough that the overlap between the projectors is effectively zero. 

Taking linear combinations of these projectors, we can define local order parameters labeled by the tunneling channels
\begin{equation}
\begin{aligned}
V_1 &= P_{00} + P_{01} + P_{10} + P_{11} + P_{++} \\
V_{\psi_1\bar{\psi}_1} &= P_{00} + P_{01} + P_{10} + P_{11} - P_{++} \\
V_{\sigma_1\bar{\sigma}_1} &= P_{00} + P_{01} - P_{10} - P_{11}\\
V_{\sigma_2\bar{\sigma}_2} &= P_{00} - P_{01} + P_{10} - P_{11} \\
V_{A} &= P_{00} - P_{01} - P_{10} + P_{11}
\end{aligned}
\end{equation}
It can be verified through explicit calculation that the correlation length measured by the (nontrivial) local order parameters on the degenerate ground space is infinity. 

The labeling by the tunneling channels indicates how the order parameters transform under the Rep($H_8$) symmetry. For example, $\sigma_1\bar{\sigma}_1$ braids with a $-1$ phase factor with $\psi_1$ and a $+1$ phase factor with $\psi_2$. Therefore, $V_{\sigma_1\bar{\sigma}_1}$ carries a charge under the $\eta_o$ symmetry but no charge under the $\eta_e$ symmetry, which can be explicitly verified since $P_{00}$, $P_{01}$ map to $P_{10}$, $P_{11}$ under the $\eta_o$ symmetry while $P_{00}$, $P_{10}$ map to $P_{01}$, $P_{11}$ under the $\eta_e$ symmetry. Similarly, we can see that $V_{\sigma_2\bar{\sigma}_2}$ is charged under $\eta_e$ but not $\eta_o$ and $V_A$ is charged under both $\eta_o$ and $\eta_e$. Under the non-invertible $\mathcal{N}_o\mathcal{N}_e$ symmetry, all three operators become non-local operators. There is no way to construct two point operators that are symmetric under the full Rep($H_8$) using $V_{\sigma_1\bar{\sigma}_1}$,$V_{\sigma_2\bar{\sigma}_2}$ and $V_{A}$, which is usually what we expect to have so that we can measure the long-range correlation in the symmetry broken state. However, the lack of symmetric two-point operator based on certain order parameters does not affect our ability to detect long-range correlation, as can be explicitly checked using the $V$ operators given above. $V_{\psi_1\bar{\psi}_1}$, on the other hand, remains invariant under $\eta_o$ and $\eta_e$ and gets a minus sign under $\mathcal{N}_o\mathcal{N}_e$, as expected from the braiding statistics between $\psi_1\bar{\psi}_1$ and $\psi_1$, $\psi_2$, $\sigma_1\sigma_2$. 

While the order parameters carry anyon-like labels, their fusion rules are different from the fusion rules of the anyons in the bulk. The order parameters are labeled by objects in the Lagrangian Algebra describing the gapped boundaries and their fusion follows the multiplication of objects in the Lagrangian Algebra, the data of which is presented in appendix~\ref{app:LagrangianAlgebra}. In particular, the order parameters fuse in the following way, which is consistent with the $M$-$3j$ symbol derived for the $\mathcal{A}_1$ boundary in appendix~\ref{app:LagrangianAlgebra}.
\begin{equation}
\begin{array}{l}
V_1 \times V_{x} = V_{x}, \ \ 
V_{\psi_1\bar{\psi}_1} \times V_{\psi_1\bar{\psi}_1} = V_1, \ \ V_{\psi_1\bar{\psi}_1} \times V_{\sigma_i\bar{\sigma}_i} = V_{\sigma_i\bar{\sigma}_i}, \ \  V_{\psi_1\bar{\psi}_1} \times V_A = V_A\\
V_{\sigma_i\bar{\sigma}_i} \times V_{\sigma_i\bar{\sigma}_i} = V_A \times V_A = \frac{1}{2}\left(V_1 + V_{\psi_1\bar{\psi}_1}\right), \ \ 
V_{\sigma_1\bar{\sigma}_1} \times V_{\sigma_2\bar{\sigma}_2} = V_A
\end{array}
\label{eq:Vfuse11}
\end{equation}

We can also define `cat' states labeled by the tunneling channels from the superposition of the short range correlated symmetry broken states.
\begin{equation}
\begin{aligned}
\ket{\varphi_{1}} =&  \ket{\phi_1} + \ket{\phi_{\eta_o}} + \ket{\phi_{\eta_e}} + \ket{\phi_{\eta_o\eta_e}} + 2\ket{\phi_{\mathcal{N}_o\mathcal{N}_e}} \\
\ket{\varphi_{\psi_1\bar{\psi}_1}} =& \ket{\phi_1} + \ket{\phi_{\eta_o}} + \ket{\phi_{\eta_e}} + \ket{\phi_{\eta_o\eta_e}} - 2\ket{\phi_{\mathcal{N}_o\mathcal{N}_e}}\\
\ket{\varphi_{\sigma_1\bar{\sigma}_1}} =& \ket{\phi_1} - \ket{\phi_{\eta_o}} + \ket{\phi_{\eta_e}} - \ket{\phi_{\eta_o\eta_e}}\\
\ket{\varphi_{\sigma_2\bar{\sigma}_2}} =& \ket{\phi_1} + \ket{\phi_{\eta_o}} - \ket{\phi_{\eta_e}} - \ket{\phi_{\eta_o\eta_e}}\\
\ket{\varphi_{A}} =&  \ket{\phi_1} - \ket{\phi_{\eta_o}} - \ket{\phi_{\eta_e}} + \ket{\phi_{\eta_o\eta_e}} 
\end{aligned}
\end{equation}
The superposition is chosen such that the `cat' states transform under the Rep($H_8$) symmetry as the labeling channel. That is, $\ket{\varphi_{\psi_1\bar{\psi}_1}}$ is invariant under the full symmetry (up to a normalization). $\ket{\varphi_{\sigma_1\bar{\sigma}_1}}$ is odd under $\eta_o$, $\ket{\varphi_{\sigma_2\bar{\sigma}_2}}$ is odd under $\eta_e$ and $\ket{\varphi_{A}}$ is odd under both. All three vanish under $\mathcal{N}_o\mathcal{N}_e$ while $\ket{\varphi_{\psi_1\bar{\psi}_1}}$ gets a minus sign under $\mathcal{N}_o\mathcal{N}_e$ while remaining invariant under both $\eta_o$ and $\eta_e$. 

The action of the $V$ order parameters on the cat states follows the same fusion rule as that between pairs of $V$ operators. For example,
\begin{equation}
\begin{array}{l}
V_1 |\varphi_x\rangle = |\varphi_x\rangle, \ \ 
V_{\psi_1\bar{\psi}_1}|\varphi_{\psi_1\bar{\psi}_1}\rangle = |\varphi_1\rangle, \ \ 
V_{\psi_1\bar{\psi}_1}|\varphi_{\sigma_i\bar{\sigma}_i}\rangle = |\varphi_{\sigma_i\bar{\sigma}_i}\rangle, \ \ V_{\psi_1\bar{\psi}_1}|\varphi_{A}\rangle = |\varphi_{A}\rangle\\
V_{\sigma_i\bar{\sigma}_i}|\varphi_{\sigma_i\bar{\sigma}_i}\rangle = V_{A}|\varphi_{A}\rangle = \frac{1}{2}\left(|\varphi_1\rangle + |\varphi_{\psi_1\bar{\psi}_1}\rangle\right), \ \ V_{\sigma_1\bar{\sigma}_1}|\varphi_{\sigma_2\bar{\sigma}_2}\rangle = |\varphi_A\rangle
\end{array}
\end{equation}
corresponding to the fusion rule listed in Eq.~\ref{eq:Vfuse11}.

\subsection{Rep($H_8$) to $\IZ_2$ symmetry-breaking phase $H_{12}$}

When the top boundary in the sandwich construction is described by the Lagrangian Algebra $\mathcal{A}_1 = 1 \oplus\psi_1\bar{\psi}_1 \oplus \sigma_1\bar{\sigma}_1 \oplus \sigma_2\bar{\sigma}_2 \oplus A$ and the bottom boundary described by $\mathcal{A}_2 = 1 \oplus \psi_1\bar{\psi}_1 \oplus \sigma_1\bar{\sigma}_1 \oplus \sigma_2\bar{\sigma}_2 \oplus B$, the Rep($H_8$) symmetry of system is again broken, but only partially. 

A representative Hamiltonian of this phase is given by
$$
H_{12} =\sum_n\left[ \left( \frac{1-X_n}{2} \right)\left( \frac{1-X_{n+1}}{2} \right) +\left( \frac{1-Z_nZ_{n+2}}{2} \right)\left( \frac{1-Z_{n+1}Z_{n+3}}{2} \right)\right]
$$
$H_{12}$ can be derived from $H_{11}$ by applying the even-site Kramers-Wannier transformation $\mathcal{N}_e$ (or the odd-site transformation $\mathcal{N}_o$). It has four degenerate ground states
\begin{align}
\ket{\mathord{+}x\mathord{+}x\cdots},\quad\ket{x\mathord{+}x\mathord{+}\cdots}
\end{align}
with $x\in\{0,1\}$.

$\mathcal{A}_1$ overlaps with $\mathcal{A}_2$ in four of the five condensed anyons $1$, $\psi_1\bar{\psi}_1$, $\sigma_1\bar{\sigma}_1$, and $\sigma_2\bar{\sigma}_2$, resulting in four tunneling channels between the two boundaries. The four-fold degenerate ground states correspond directly to the four tunneling channels. In particular, we can define cat-states labeled by the tunneling channels as
\begin{equation}
\begin{aligned} |\varphi_1\rangle =& \ket{+0\mathord{+}0\cdots} + \ket{+1\mathord{+}1\cdots} + \ket{0\mathord{+}0\mathord{+}\cdots} + \ket{1\mathord{+}1\mathord{+}\cdots}\\
|\varphi_{\psi_1\bar{\psi}_1}\rangle =& \ket{+0\mathord{+}0\cdots} + \ket{+1\mathord{+}1\cdots} - \ket{0\mathord{+}0\mathord{+}\cdots} - \ket{1\mathord{+}1\mathord{+}\cdots}\\
|\varphi_{\sigma_1\bar{\sigma}_1}\rangle =& \ket{0\mathord{+}0\mathord{+}\cdots} - \ket{1\mathord{+}1\mathord{+}\cdots}\\
|\varphi_{\sigma_2\bar{\sigma}_2}\rangle =& \ket{+0\mathord{+}0\cdots} - \ket{+1\mathord{+}1\cdots}\\
\end{aligned}
\end{equation}
and local order parameters labeled by the tunneling channels as
$$
\begin{aligned}
V_1 =& P_{+0} + P_{+1} + P_{0+} + P_{1+} \\
V_{\psi_1\bar{\psi}_1} =& P_{+0} + P_{+1} - P_{0+} - P_{1+} \\
V_{\sigma_1\bar{\sigma}_1} =& P_{0+} - P_{1+} \\
V_{\sigma_2\bar{\sigma}_2} =& P_{+0} - P_{+1}
\end{aligned}
$$
Again, the projectors $P_{+i}$ and $P_{i+}$ are onto a finite segment of the double qubit chains whose length is long enough that we can ignore the overlap between the projectors.

The cat-states and the order parameters transform as their labeling tunneling channels. $\ket{\varphi_{1}}$ and $V_1$ are invariant under the full Rep($H_8$) symmetry. $\ket{\varphi_{\psi_1\bar{\psi}_1}}$ and $V_{\psi_1\bar{\psi}_1}$ are invariant under both $\eta_o$ and $\eta_e$ and gets a minus sign under $\mathcal{N}_o\mathcal{N}_e$. $\ket{\varphi_{\sigma_1\bar{\sigma}_1}}$ ($\ket{\varphi_{\sigma_2\bar{\sigma}_2}}$) and $V_{\sigma_1\bar{\sigma}_1}$ ($V_{\sigma_2\bar{\sigma}_2}$) is odd under $\eta_o$ ($\eta_e$), even under $\eta_e$ ($\eta_o$) and vanishes (or becomes nonlocal) under $\mathcal{N}_o\mathcal{N}_e$.

The fusion of the order parameters and their action on the cat-states satisfy the Lagrangian Algebra given in Ref.~\onlinecite{Cong2017} and appendix~\ref{app:LagrangianAlgebra}. It is similar to the algebra of order parameters in $H_{11}$ with one major difference being $V_{\sigma_1\bar{\sigma}_1} \times V_{\sigma_2\bar{\sigma}_2} = 0$ since the $A$ tunneling channel does not exist anymore. Other than that, $V_{\sigma_1\bar{\sigma}_1} \times V_{\sigma_1\bar{\sigma}_1} = \frac{1}{2}\left(V_1 - V_{\psi_1\bar{\psi}_1}\right)$ differs from the corresponding fusion rule in $H_{11}$ by a minus sign in front of $V_{\psi_1\bar{\psi}_1}$, which follows from the change in the $M$-$3j$ symbol from $\CA_1$ and $\CA_2$ as explained in Appendix~\ref{app:LagrangianAlgebra}. Otherwise, the fusion rule of the order parameters is similar to that in $H_{11}$.

What is the remaining symmetry of this phase? Each symmetry broken state is invariant under a $\IZ_2$ subgroup of the symmetry. $\ket{x\mathord{+}x\mathord{+}\cdots}$ are invariant under $\eta_e$ and $\ket{+x\mathord{+}x\cdots}$ are invariant under $\eta_o$. $\eta_o\eta_e$ and $\mathcal{N}_o\mathcal{N}_e$ on the other hand, do not keep any symmetry-broken state invariant.   $\ket{+0\mathord{+}0\cdots}$ and $\ket{+1\mathord{+}1\cdots}$ are mapped to each other under $\eta_o\eta_e$. So are $\ket{0+0\mathord{+}\cdots}$ and $\ket{1\mathord{+}1\mathord{+}\cdots}$. Under $\mathcal{N}_o\mathcal{N}_e$, the states are mapped as
\begin{equation}
\begin{array}{l}
\mathcal{N}_o\mathcal{N}_e\ket{+x\mathord{+}x\cdots} = \ket{0\mathord{+}0\mathord{+}\cdots} + \ket{1\mathord{+}1\mathord{+}\cdots} \\
\mathcal{N}_o\mathcal{N}_e\ket{x\mathord{+}x\mathord{+}\cdots} = \ket{+0\mathord{+}0\cdots} + \ket{+1\mathord{+}1\cdots} 
\end{array}
\end{equation}
That is, the symmetry broken states are mapped into their linear combinations under $\mathcal{N}_o\mathcal{N}_e$. This is a phenomenon that is unique to non-invertible symmetries. With invertible symmetries, a symmetry broken state either remains invariant or is mapped to a different symmetry broken state. 

Therefore, in the phase represented by $H_{12}$, a $\IZ_2$ invertible symmetry remains. Note that the dependence of the remaining symmetry on the symmetry broken state is not a surprising feature. It shows up already in simple ferromagnet where the $SO(3)$ rotation symmetry breaks down to a $U(1)$ subgroup. Which $U(1)$ subgroup remains depends on the direction of the magnetic domain. 

\subsection{Rep($H_8$) to diagonal $\IZ_2$ symmetry-breaking phase $H_{13}$}

When the top boundary in the sandwich construction is described by the Lagrangian Algebra $\mathcal{A}_1 = 1 \oplus \psi_1\bar{\psi}_1 \oplus \sigma_1\bar{\sigma}_1 \oplus \sigma_2\bar{\sigma}_2 \oplus A$ and the bottom boundary described by $\mathcal{A}_3 = 1 \oplus \psi_1\bar{\psi}_2 \oplus \sigma_1\bar{\sigma}_2 \oplus \sigma_2\bar{\sigma}_1 \oplus A$, the Rep($H_8$) symmetry is broken down to the diagonal $\IZ_2$ subgroup. A representative Hamiltonian of this phase is given by
$$
H_{13}= -\sum_k (X_{2k-1}X_{2k}+Z_{2k-1}Z_{2k}Z_{2k+1}Z_{2k+2})
$$
$H_{13}$ is exactly solvable with two degenerate ground states
\begin{align}
    \ket{\phi_{00+11}} = \bigotimes_k \frac{1}{\sqrt{2}}(\ket{00}+\ket{11})_{2k-1,2k},\quad \ket{\phi_{01+10}} = \bigotimes_k \frac{1}{\sqrt{2}}(\ket{01}+\ket{10})_{2k-1,2k}. 
\end{align}
Both are eigenvalue $+1$ eigenstates of $X_{2k-1}X_{2k}$  while $\ket{\phi_{00+11}}$ has eigenvalue $+1$ under $Z_{2k}Z_{2k+1}$ and $\ket{\phi_{01+10}}$ has eigenvalue $-1$.

What symmetries remain after the symmetry breaking? Both symmetry-broken states remain invariant under the diagonal $\IZ_2$ symmetry $\eta_o\eta_e$, so there is definitely a $\IZ_2$ symmetry that remains. 
$$
\eta_o\eta_e \ket{\phi_{00+11}} = \ket{\phi_{00+11}}, \quad \eta_o\eta_e \ket{\phi_{01+10}} = \ket{\phi_{01+10}}
$$
On the other hand, no state remains invariant under the action of $\eta_o$, $\eta_e$ or $\mathcal{N}_o\mathcal{N}_e$.

Compared to the case of $H_{12}$, it seems that $H_{13}$ should have more remaining symmetry because the ground state degeneracy is smaller. In fact, if we look closer at the action of $\eta_o$, $\eta_e$ or $\mathcal{N}_o\mathcal{N}_e$, we see that certain linear combinations of these operators actually do keep both symmetry-broken states invariant. In particular, 
$$
\begin{array}{l}
\eta_o \ket{\phi_{00+11}} = \ket{\phi_{01+10}}, \quad \eta_o \ket{\phi_{01+10}} = \ket{\phi_{00+11}} \\
\eta_e \ket{\phi_{00+11}} = \ket{\phi_{01+10}}, \quad \eta_e \ket{\phi_{01+10}} = \ket{\phi_{00+11}} \\
\mathcal{N}_o\mathcal{N}_e \ket{\phi_{00+11}} = \mathcal{N}_o\mathcal{N}_e \ket{\phi_{01+10}} = \ket{\phi_{00+11}} + \ket{\phi_{01+10}}
\end{array}
$$
The relation in the last row can be derived by noticing that the action of $\mathcal{N}_o\mathcal{N}_e$ involves a projection onto the $\IZ_2\times \IZ_2$ invariant subspace and $\ket{\phi_{00+11}} + \ket{\phi_{01+10}}$ is the only state in the ground space invariant under the full $\IZ_2 \times \IZ_2$ symmetry. It is interesting to notice that $\mathcal{N}-\eta_o$ and $\mathcal{N}-\eta_e$ keep both symmetry-broken states invariant. However, due to the minus sign in the linear combination, $\mathcal{N}-\eta_o$ and $\mathcal{N}-\eta_e$ are not generalized symmetry operators and do not have a corresponding topological defect.

These two ground states correspond to the two tunneling channels $1$ and $A$ between the top and bottom boundary. In particular, we can define cat-states labeled by $1$ and $A$ as
$$
\ket{\varphi_1} = \ket{\phi_{00+11}} + \ket{\phi_{01+10}}, \quad
\ket{\varphi_A} = \ket{\phi_{00+11}} - \ket{\phi_{01+10}}
$$
and local order parameters labeled by $1$ and $A$
$$
V_1 = P_{00+11} + P_{01+10}, \ \ V_A = P_{00+11} - P_{01+10}
$$
$P_{00+11}$ is the projector onto the $\frac{1}{\sqrt{2}}(\ket{00}+\ket{11})$ state on a pair of qubits $2k-1$ and $2k$, while $P_{01+10}$ is the projector onto the $\frac{1}{\sqrt{2}}(\ket{01}+\ket{10})$ state. $V_A$ is an $\mathbb Z_2$ order parameter in this case. Note that while in both $\mathcal{A}_1$ and $\mathcal{A}_3$, the fusion of two copies of $A$ results in more than just the identity, the extra fusion channel does not match between $\mathcal{A}_1$ and $\mathcal{A}_3$. Therefore, as an order parameter, the fusion of $V_A$ is simply $\mathbb Z_2$.
$$
V_A \times V_A = V_1
$$
The action of $V_A$ on the $\ket{\varphi_1}$ and $\ket{\varphi_A}$ state also reflects this $\mathbb Z_2$ structure.
$$
V_A \ket{\varphi_1} = \ket{\varphi_A}, \ \ V_A \ket{\varphi_A} = \ket{\varphi_1}
$$

It can be checked that $V_1$ and $\ket{\varphi_1}$ remain invariant under the full Rep($H_8$) symmetry, while $V_A$ and $\ket{\varphi_A}$ are charged under either $\eta_o$ or $\eta_e$ and becomes nonlocal (vanishes) under $\mathcal{N}_o\mathcal{N}_e$. Therefore, the order parameter and the cat state transforms as their labeling tunneling channel.

\subsection{Rep($H_8$) symmetric phase $H_{14}$}

When the top boundary in the sandwich construction is described by the Lagrangian Algebra $\mathcal{A}_1 = 1 \oplus \psi_1\bar{\psi}_1 \oplus \sigma_1\bar{\sigma}_1 \oplus \sigma_2\bar{\sigma}_2 \oplus A$ and the bottom boundary described by $\mathcal{A}_4 = 1 \oplus \psi_1\bar{\psi}_2 \oplus \sigma_1\bar{\sigma}_2 \oplus \sigma_2\bar{\sigma}_1 \oplus B$, their only overlap is the trivial anyon. Therefore, the sandwich represents a symmetric phase with Rep($H_8$) symmetry. 

A representative Hamiltonian in this phase is given by the cluster state Hamiltonian
$$
H_{14} = -\sum_n Z_{n-1}X_nZ_{n+1}
$$
which has a unique gapped ground state. This is the only symmetric phase under Rep($H_8$). $H_{14}$ can be obtained from $H_{13}$ by applying the even-site (or odd-site) Kramers-Wannier transformation $\mathcal{N}_e$ (or $\mathcal{N}_o$). The fact that the cluster state is invariant by this realization of the Rep($H_8$) symmetry has been pointed out in~\cite{Seifnashri24}

\subsection{Rep($H_8$) to diagonal $\IZ_2$ symmetry-breaking phase $H_{15}$}

When the top boundary in the sandwich construction is described by the Lagrangian Algebra $\mathcal{A}_1 = 1 \oplus \psi_1\bar{\psi}_1 \oplus \sigma_1\bar{\sigma}_1 \oplus \sigma_2\bar{\sigma}_2 \oplus A$ and the bottom boundary described by $\mathcal{A}_5 = 1 \oplus \psi_1\psi_2 \oplus \psi_1\bar{\psi}_2 \oplus \psi_1\bar{\psi}_1 \oplus 2A$, there are four tunneling channels: one for $1$, one for $\psi_1\bar{\psi}_1$ and two for $A$. There is a four-fold ground state degeneracy.

A representative Hamiltonian for this phase is
$$
H_{15} = -\sum_k(Y_{2k-1}Y_{2k}Z_{2k+1}Z_{2k+2}+Z_{2k-1}Z_{2k}Y_{2k+1}Y_{2k+2})
$$
WLOG, we assume that each qubit chain has an even number of qubits and the total number of qubits in the double chain is a multiple of four. For convenience, denote the four Bell states by $\ket{\Psi_{ab}}$ with $a,b\in\{0,1\}$, such that the eigenvalues of $XX$ and $ZZ$ are $(-1)^a$ and $(-1)^b$, respectively. More explicitly, $\ket{\Psi_{ab}}=(1/\sqrt{2})\sum_{x\in\{0,1\}}(-1)^{ax}\ket{x,x\oplus b}$ where $\oplus$ denotes addition modulo 2. Hamiltonian $H_{15}$ is exactly solvable and the 4 degenerate ground states are given by
\begin{align}
  \ket{\phi_{ab}} =   \bigotimes_k \ket{\Psi_{ab}}_{4k-3,4k-2}\ket{\Psi_{a(1+a+b)}}_{4k-1,4k}, \ a,b=0,1
\end{align}
Each wavefunction is a product of Bell states. The Bell states in the same wavefunction have the same eigenvalue under $XX$ ($\pm 1$).  When $XX=1$, the Bell states have alternating eigenvalues under $ZZ$. When $XX=-1$, the Bell states have the same eigenvalue under $ZZ$.

Under the action of $\eta_o$ or $\eta_e$, the eigenvalue of $X_{2k-1}X_{2k}$ remains the same while the eigenvalue of $Z_{2k-1}Z_{2k}$ changes. Therefore, under the $\IZ_2\times \IZ_2$ symmetry, the $\ket{\phi_{ab}}$ states are mapped as
$$
\eta_o\ket{\phi_{ab}} = \eta_e\ket{\phi_{ab}} = \ket{\phi_{a\bar{b}}}, \ \ \eta_o\eta_e\ket{\phi_{ab}} = \ket{\phi_{ab}}
$$
Under the action of $\CN_o\CN_e$, the eigenvalue of $X_{2k-1}X_{2k}$ is mapped to the eigenvalue of $Z_{2k-1}Z_{2k}Z_{2k+1}Z_{2k+2}$ and vice versa. Therefore, $\ket{\phi_{ab}}$ is mapped into the subspace spanned by $\ket{\phi_{\bar{a}b}}$ and $\ket{\phi_{\bar{a}\bar{b}}}$. Moreover, since $\CN_o\CN_e$ involves a projection onto the $\IZ_2\times \IZ_2$ invariant subspace, its action on $\ket{\phi_{ab}}$ is given by
$$
\CN_o\CN_e\ket{\phi_{ab}} = \ket{\phi_{\bar{a}b}} + \ket{\phi_{\bar{a}\bar{b}}}
$$

The $\ket{\phi_{ab}}$ states can be put into linear superpositions that correspond to the four tunneling channels. Define the cat states labeled by $1$, $\psi_1\bar{\psi}_1$, $A_1$, $A_2$ as
\begin{equation}
    \begin{aligned}
        \ket{\varphi_1} &= \ket{\phi_{00}} + \ket{\phi_{01}} + \ket{\phi_{10}} + \ket{\phi_{11}} \\
\ket{\varphi_{\psi_1\bar{\psi}_1}} &= \ket{\phi_{00}} + \ket{\phi_{01}} - \ket{\phi_{10}} - \ket{\phi_{11}} \\
\ket{\varphi_{A_1}} &= \ket{\phi_{00}} - \ket{\phi_{01}} \\
\ket{\varphi_{A_2}} &= \ket{\phi_{10}} - \ket{\phi_{11}}
    \end{aligned}
\end{equation}
and the corresponding local order parameters as
\begin{equation}
    \begin{aligned}
V_1 &= P_{00}+P_{01}+P_{10}+P_{11} \\
V_{\psi_1\bar{\psi}_1} &= P_{00}+P_{01}-P_{10}-P_{11} \\
V_{A_1} &= P_{00}-P_{01}\\
V_{A_2} &= P_{10}-P_{11} \\
    \end{aligned}
\end{equation}
We can check that the cat-states and the local order parameters transform under the Rep($H_8$) symmetry as indicated by their labeling tunneling channel. $V_1$ and $\ket{\varphi_1}$ are invariant under the full symmetry. $V_{\psi_1\bar{\psi}_1}$ and $\ket{\varphi_{\psi_1\bar{\psi}_1}}$ are invariant under the $\IZ_2\times \IZ_2$ symmetry but changes sign under $\CN_o\CN_e$. $V_{A_1}$, $V_{A_2}$, $\ket{\varphi_{A_1}}$ and $\ket{\varphi_{A_2}}$ carry a charge under $\eta_o$ and $\eta_e$ but are not charged under $\eta_o\eta_e$. Under $\CN_o\CN_e$, they vanish (or become nonlocal). 

The fusion of the order parameters again follow from the Lagrangian Algebra of the top and bottom boundaries. 
\begin{equation}
\begin{array}{l}
V_1 \times V_{x} = V_{x}, \ \ 
V_{\psi_1\bar{\psi}_1} \times V_{\psi_1\bar{\psi}_1} = V_1, \ \ V_{\psi_1\bar{\psi}_1} \times V_{A_1} = V_{A_1}, \ \  V_{\psi_1\bar{\psi}_1} \times V_{A_2} = -V_{A_2}, \\ 
V_{A_1} \times V_{A_1} = \frac{1}{2}\left(V_1 + V_{\psi_1\bar{\psi}_1}\right), \ \ 
V_{A_2} \times V_{A_2} = \frac{1}{2}\left(V_1 - V_{\psi_1\bar{\psi}_1}\right), \ \ 
V_{A_1} \times V_{A_2} = 0
\end{array}
\label{eq:Vfuse15}
\end{equation}
Their corresponding actions on the cat-states are given by
\begin{equation}
\begin{array}{l}
V_1\ket{\varphi_{x}} = \ket{\varphi_{x}}, \ \ 
V_{\psi_1\bar{\psi}_1} \ket{\varphi_{\psi_1\bar{\psi}_1}} = \ket{\varphi_1}, \ \ V_{\psi_1\bar{\psi}_1} \ket{\varphi_{A_1}} = \ket{\varphi_{A_1}}, \ \  V_{\psi_1\bar{\psi}_1} \ket{\varphi_{A_2}} = -\ket{\varphi_{A_2}}, \\ 
V_{A_1}\ket{\varphi_{A_1}} = \frac{1}{2}\left(\ket{\varphi_1} + \ket{\varphi_{\psi_1\bar{\psi}_1}}\right), \ \ 
V_{A_2} \ket{\varphi_{A_2}} = \frac{1}{2}\left(\ket{\varphi_1} - \ket{\varphi_{\psi_1\bar{\psi}_1}}\right), \ \ 
V_{A_1} \ket{\varphi_{A_2}} = 0
\end{array}
\label{eq:Vphi15}
\end{equation}

What symmetry remains? This case is similar to the case of $H_{13}$ because all ground states remain invariant under the action of $\eta_o\eta_e$, so at least the diagonal $\IZ_2$ symmetry remains. On the other hand, no state remains invariant under $\eta_o$, $\eta_e$ or $\CN_o\CN_e$. Moreover, no linear combination of $\eta_o$, $\eta_e$ and $\CN_o\CN_e$ keeps the $\ket{\phi_{ab}}$ states invariant. In this sense, the remaining symmetry is `smaller' than that of $H_{13}$, which is consistent with the larger ground state degeneracy than $H_{13}$.

\subsection{Rep($H_8$) to $\IZ_2\times \IZ_2$ symmetry-breaking phase $H_{16}$}

When the top boundary in the sandwich construction is described by the Lagrangian Algebra $\mathcal{A}_1 = 1 \oplus \psi_1\bar{\psi}_1 \oplus \sigma_1\bar{\sigma}_1 \oplus \sigma_2\bar{\sigma}_2 \oplus A$ and the bottom boundary described by $\mathcal{A}_6 = 1 \oplus \psi_1\psi_2 \oplus \psi_1\bar{\psi}_2 \oplus \psi_1\bar{\psi}_1 \oplus 2B$, there are two tunneling channels: one for $1$ and one for $\psi_1\bar{\psi}_1$. There is a two-fold ground state degeneracy. 

A representative Hamiltonian for this phase is
$$
H_{16}=\sum_n Z_nY_{n+1}Y_{n+2}Z_{n+3}. 
$$
Each term in $H_{16}$ is the product of two neighboring cluster state Hamiltonian terms $Z_nX_{n+1}Z_{n+2}$ and $Z_{n+1}X_{n+2}Z_{n+3}$. With a positive sign in front, the ground states spontaneously breaks the translation symmetry and form the `anti-ferromagnetic' order of cluster state. That is, one of the ground states is the $+1$ ($-1$) eigenstate of the cluster state Hamiltonian terms centered on odd sites (even sites). The other one can be obtained by translating the first state by one lattice site. The two ground states are
    \begin{align}
        \ket{\phi_{+-}} = \left(\prod_n {\rm CZ}_{n,n+1}\right)\ket{\mathord{+}\mathord{-}\mathord{+}\mathord{-}\cdots},\quad \ket{\phi_{-+}} = \left(\prod_n {\rm CZ}_{n,n+1}\right)\ket{\mathord{-}\mathord{+}\mathord{-}\mathord{+}\cdots}. 
    \end{align}
The two states are mapped to each other by $\CN_o\CN_e$. Therefore, the two cat-states corresponding to the two tunneling channels $1$ and $\psi_1\bar{\psi}_1$ are
$$
\ket{\varphi_1} = \ket{\phi_{+-}} + \ket{\phi_{-+}}, \quad \ket{\varphi_{\psi_1\bar{\psi}_1}} = \ket{\phi_{+-}} - \ket{\phi_{-+}}
$$
and the two corresponding local order parameters are
$$
V_1 = P_{+-} + P_{-+}, \quad V_{\psi_1\bar{\psi}_1} = P_{+-} - P_{-+}
$$
where $P_{+-}$ ($P_{-+}$) is a projector onto the subspace where $Z_nX_{n+1}Z_{n+2} = 1$ when $n$ is even (odd) and $Z_nX_{n+1}Z_{n+2} = -1$ when $n$ is odd (even).  $V_{\psi_1\bar{\psi}_1}$ is a $Z_2$ order parameter that distinguishes between the two symmetry-broken states. This $Z_2$ structure follows naturally from the $\psi_1\bar{\psi}_1 \times \psi_1\bar{\psi}_1 = 1$ fusion rule in both $\mathcal{A}_1$ and $\mathcal{A}_6$.

The two symmetry broken states $\ket{\phi_{+-}}$ and $\ket{\phi_{-+}}$ are both invariant under the invertible $\IZ_2\times \IZ_2$ symmetry generated by $\eta_o$ and $\eta_e$. They are mapped into each other under $\CN_o\CN_e$. Therefore, this phase can be straight-forwardly described as symmetry broken from Rep($H_8$) to $\IZ_2\times \IZ_2$. 

\subsection{Comparison with Hopf algebra construction}
The gapped Rep($H_8$) phases also admits a $1+1$d lattice realization formulated directly in terms of a Hopf algebra and its comodule algebra. In the construction presented above, the non-invertible symmetry line $\CN_o\CN_e$ is implemented via the Kramers--Wannier duality on both qubit chains, whose lattice realization necessarily mixes with lattice translation. However, in the infrared this symmetry flows to a genuine $\Rep(H_8)$ symmetry under which translation acts trivially. Since $\Rep(H_8)$ is anomaly-free, it admits an alternative lattice realization in which the non-invertible line $\CN_o\CN_e$ is implemented in an ``on-site'' manner and does not mix with translation. Such a formulation is described in \cite{Inamura:2021szw,lu2026generalized}, where the construction of $\Rep(H_8)$-symmetric gapped phases is given in terms of the Hopf algebra $H_8$ and an $H_8$-comodule algebra $V$.

To be specific, for $\Rep(H_8)$ symmetry, there are 6 symmetric gapped phases given by the $H_8$-comodule algebra. We refer the readers to \cite{lu2026generalized} for details. All the symmetric gapped phases can be realized on tensor product Hilbert space with two-body commuting projector Hamiltonian. We list the $H_8$-comodule algebra $V$, the ground state degeneracy (GSD) and the symmetry action on the short-range entangled ground states as follows.
\begin{itemize}
    \item $V=\IC$ with GSD=1, identified with $H_{14}$,
    \begin{equation}
	\begin{array}{c|c|c|c|c}
		1 & \eta_o & \eta_e &\eta_o \eta_e& \CN_o\CN_e \\ \hline
		1 & 1& 1& 1& 2
	\end{array}
\end{equation}
 \item $V=\IC[\grp{1,x}] \cong \IC[\grp{1,y}]$ with GSD=2, identified with $H_{13}$,
 \begin{equation}
	\begin{array}{c|c|c|c|c}
		1 & \eta_o & \eta_e &\eta_o \eta_e& \CN_o\CN_e \\ \hline
		\left(
		\begin{array}{cc}
			1 & 0 \\
			0 & 1 \\
		\end{array}
		\right)& \left(
		\begin{array}{cc}
			0 & 1 \\
			1 & 0 \\
		\end{array}
		\right)&\left(
		\begin{array}{cc}
			0 & 1 \\
			1 & 0 \\
		\end{array}
		\right)&\left(
		\begin{array}{cc}
			1 & 0 \\
			0 & 1 \\
		\end{array}
		\right)&\left(
		\begin{array}{cc}
			1 & 1 \\
			1 & 1 \\
		\end{array}
		\right)
	\end{array}
\end{equation}
\item $V=\IC[\grp{1,xy}]$ with GSD=2, identified with $H_{16}$,
\begin{equation}\label{eq:H13act}
	\begin{array}{c|c|c|c|c}
		1 & \eta_o & \eta_e &\eta_o \eta_e& \CN_o\CN_e \\ \hline
		\left(
		\begin{array}{cc}
			1 & 0 \\
			0 & 1 \\
		\end{array}
		\right)& \left(
		\begin{array}{cc}
			1 & 0 \\
			0 & 1 \\
		\end{array}
		\right)&\left(
		\begin{array}{cc}
			1 & 0 \\
			0 & 1 \\
		\end{array}
		\right)&\left(
		\begin{array}{cc}
			1 & 0 \\
			0 & 1 \\
		\end{array}
		\right)&\left(
		\begin{array}{cc}
			0 & 2 \\
			2 & 0 \\
		\end{array}
		\right)
	\end{array}
\end{equation}
\item $V=\IC[\grp{1,x,y,xy}]$ with GSD=4, identified with $H_{15}$,
\begin{equation}
	\begin{array}{c|c|c|c|c}
		1 & \eta_o & \eta_e &\eta_o \eta_e& \CN_o\CN_e \\\hline
		\left(
		\begin{smallmatrix}
			1 & 0 & 0 & 0 \\
			0 & 1 & 0 & 0 \\
			0 & 0 & 1 & 0 \\
			0 & 0 & 0 & 1
		\end{smallmatrix}
		\right)&
		\left(
		\begin{smallmatrix}
			0 & 1 & 0 & 0 \\
			1 & 0 & 0 & 0 \\
			0 & 0 & 0 & 1 \\
			0 & 0 & 1 & 0 \\
		\end{smallmatrix}
		\right)&
		\left(
		\begin{smallmatrix}
			0 & 1 & 0 & 0 \\
			1 & 0 & 0 & 0 \\
			0 & 0 & 0 & 1 \\
			0 & 0 & 1 & 0 \\
		\end{smallmatrix}
		\right)&
		\left(
		\begin{smallmatrix}
			1 & 0 & 0 & 0 \\
			0 & 1 & 0 & 0 \\
			0 & 0 & 1 & 0 \\
			0 & 0 & 0 & 1
		\end{smallmatrix}
		\right)&
		\left(
		\begin{smallmatrix}
 0 & 0 & 1 & 1 \\
0 & 0 & 1 & 1 \\
1 & 1 & 0 & 0 \\
1 & 1 & 0 & 0 
		\end{smallmatrix}
		\right)
	\end{array}
\end{equation}
\item $V=A_{xy}^q$ with GSD=4, identified with $H_{12}$,
\begin{equation}
	\begin{array}{c|c|c|c|c}
		1 & \eta_o & \eta_e &\eta_o \eta_e& \CN_o\CN_e \\ \hline
		\left(
		\begin{smallmatrix}
			1 & 0 & 0 & 0 \\
			0 & 1 & 0 & 0 \\
			0 & 0 & 1 & 0 \\
			0 & 0 & 0 & 1
		\end{smallmatrix}
		\right)&\left(
		\begin{smallmatrix}
				1 & 0 & 0 & 0 \\
				0 & 1 & 0 & 0 \\
				0 & 0 & 0 & 1 \\
				0 & 0 & 1 & 0 
		\end{smallmatrix}
		\right)&\left(
		\begin{smallmatrix}
				0 & 1 & 0 & 0 \\
				1 & 0 & 0 & 0 \\
				0 & 0 & 1 & 0 \\
				0 & 0 & 0 & 1 
		\end{smallmatrix}
		\right)&\left(
		\begin{smallmatrix}
	0 & 1 & 0 & 0 \\
	1 & 0 & 0 & 0 \\
	0 & 0 & 0 & 1 \\
	0 & 0 & 1 & 0 
		\end{smallmatrix}
		\right)&\left(
		\begin{smallmatrix}
	0 & 0 & 1 & 1 \\
	0 & 0 & 1 & 1 \\
	1 & 1 & 0 & 0 \\
	1 & 1 & 0 & 0 
		\end{smallmatrix}
		\right)
	\end{array}
\end{equation}
where $A_{xy}^q$ is a 4-dimension $H_8$-comodule algebra, but is neither a subalgebra of $H_8$ nor a (twisted) group algebra. The full data is presented in \cite{van2025h,lu2026generalized}.
\item $V=\IC[H_8]$ with GSD=5, identified with $H_{11}$,
\begin{align}
	\begin{array}{c|c|c|c|c}
		1 & \eta_o & \eta_e &\eta_o \eta_e& \CN_o\CN_e \\ \hline
		\left(
		\begin{smallmatrix}
			1 & 0 & 0 & 0 & 0 \\
			0 & 1 & 0 & 0 & 0 \\
			0 & 0 & 1 & 0 & 0 \\
			0 & 0 & 0 & 1 & 0 \\
			0 & 0 & 0 & 0 & 1
		\end{smallmatrix}
		\right)
		&
		\left(
		\begin{smallmatrix}
			0 & 0 & 1 & 0 & 0 \\
			0 & 0 & 0 & 1 & 0 \\
			1 & 0 & 0 & 0 & 0 \\
			0 & 1 & 0 & 0 & 0 \\
			0 & 0 & 0 & 0 & 1
		\end{smallmatrix}
		\right)
		&
		\left(
		\begin{smallmatrix}
			0 & 1 & 0 & 0 & 0 \\
			1 & 0 & 0 & 0 & 0 \\
			0 & 0 & 0 & 1 & 0 \\
			0 & 0 & 1 & 0 & 0 \\
			0 & 0 & 0 & 0 & 1
		\end{smallmatrix}
		\right)
		&
		\left(
		\begin{smallmatrix}
			0 & 0 & 0 & 1 & 0 \\
			0 & 0 & 1 & 0 & 0 \\
			0 & 1 & 0 & 0 & 0 \\
			1 & 0 & 0 & 0 & 0 \\
			0 & 0 & 0 & 0 & 1
		\end{smallmatrix}
		\right)
		&
		\left(
		\begin{smallmatrix}
			0 & 0 & 0 & 0 & 1 \\
			0 & 0 & 0 & 0 & 1 \\
			0 & 0 & 0 & 0 & 1 \\
			0 & 0 & 0 & 0 & 1 \\
			1 & 1 & 1 & 1 & 0
		\end{smallmatrix}
		\right)
	\end{array}
\end{align}
\end{itemize}
One might expect the twisted group algebra $\IC[\IZ_2 \times \IZ_2]^\psi$ to yield a distinct $H_8$-comodule algebra, but whose gapped phase is in fact equivalent to the $\Rep(H_8)$ symmetric phase $H_{14}$ constructed from the trivial algebra $\IC$. The phase $H_{12}$ is instead realized by the comodule algebra $A_{xy}^q$, which is neither a subalgebra of $H_8$ nor a (twisted) group algebra. Such a partially symmetry-breaking phase has no analogue with the symmetries of the form $G$ or $\Rep(G)$.

\section{Dual $(D_8,\gamma)$ Lattice 
Model}
\label{sec:Dual}

In this section, we will gauge the diagonal $\IZ_2^{oe}$ symmetry on the lattice to get the dual $(D_8,\gamma)$ model. Upon gauging the $\IZ_2^{oe}$ in $\Rep(H_8)$, the flat connection will lead to dual quantum symmetry which is given by invertible symmetry $D_8$ with a 't Hooft anomaly $\gamma\in H^3(D_8,U(1))$. Some subtlety related to the lattice translation symmetry occurs in this process as we will see. 

\subsection{The Symmetry after Gauging}
In this subsection, we will identify the symmetry after gauging. Gauging the diagonal $\mathbb{Z}_2^{oe}$ symmetry generated by $\eta_o\eta_e=\prod_n X_n$ (and then integrating out the matter fields using the gauge constraints) is equivalent to implementing the Kramers-Wannier transformation: $(X_n,Z_nZ_{n+1})\mapsto (Z_nZ_{n+1},X_{n+1})$. As a result, the theory after gauging will automatically have a dual symmetry $\mathbb{Z}_2^X$ generated by $\eta_X:=\prod_n X_n$. We claim that the Hamiltonian after gauging commutes with two more symmetry generators: $\eta_Z:=\prod_n Z_n$ and $V:=T\prod_n \mc H_n$ where $T$ as defined previously is the translation generator and $\mc H_n$ is the Hadamard gate acting on site $n$. If we restrict to the subspace where $T^2=1$, these operators effectively generate a symmetry group $D_8=\langle r,s| r^4=s^2=1,srs=r^{-1} \rangle$. 
More precisely, we may identify the dual $\mathbb{Z}_2^X$ symmetry generator $\eta_X$ as $s$. The two order-4 operators $\eta_X V$ and $\eta_Z V$ may be identified with $r$ and $r^3$, respectively. We note that there is no fundamental distinction between $\eta_X V$ and $\eta_Z V$, as they are related by an onsite unitary transformation $\eta_X$ that preserves the action of the full symmetry group. 
It is not hard to see that this dual symmetry has a Lieb-Schultz-Mattis type anomaly: Although $\eta_X$ and $\eta_Z$ commute given the even number of total sites, their actions on each single site anti-commutes. Moreover, the single-site translation in $V$ can detect this projective representation. 

We will first derive the above claim by a non-rigorous but relatively simple calculation, ignoring all subtleties associated with boundary conditions or symmetry charge sectors. We simply push the Rep($H_8$) symmetry generators through the $\mathbb{Z}_2^{oe}$ gauging operator which we denote by $\Noe$ \footnote{Note that it implements gauging $\mathbb{Z}_2^{oe}$, which is different from $\CN = \CN_o\CN_e$.}, and see what we get. Using $\Noe(X_n,Z_nZ_{n+1})= (Z_nZ_{n+1},X_{n+1})\Noe$, it is straightforward to check that $\Noe \eta_o=\Noe \eta_e=\eta_Z\Noe$, hence $\eta_Z$ should be a symmetry generator after gauging. Furthermore, by presuming that an unknown operator $V$ satisfies $\Noe\mc N=V\Noe$, we can compute
\begin{align}
    VX_n\Noe&=V\Noe Z_{n-1}Z_n=\Noe \mc NZ_{n-1}Z_n=\Noe (\cdots X_{n-2}X_{n-1}X_n)\mc N\nonumber\\
    &=Z_{n+1}\Noe \mc N=Z_{n+1}V\Noe , 
\end{align}
implying $(VX_n-Z_{n+1}V)\Noe =0$. We also have
\begin{align}
    VZ_n\Noe &=V\Noe (\cdots X_{n-3}X_{n-2}X_{n-1})=\Noe \mc N(\cdots X_{n-3}X_{n-2}X_{n-1})\nonumber\\
    &=\Noe  Z_nZ_{n+1}\mc N=X_{n+1}\Noe \mc N=X_{n+1}V\Noe , 
\end{align}
implying $(VZ_n-X_{n+1}V)\Noe =0$. We can then take $V$ to be $T\prod_n\mc H_n$ up to a normalization factor. This derivation is intuitive but not totally rigorous. For example, the relation $\Noe(\cdots X_{n-2}X_{n-1}X_n)=Z_{n+1}\Noe $ is not strictly correct on a finite closed chain. Moreover, since the duality operator $\Noe$ is not invertible, we can not directly conclude $VX_n=Z_{n+1}V$ from $(VX_n-Z_{n+1}V)\Noe =0$. Nonetheless, the result we have derived is correct and we refer interested readers to Appendix \ref{app:LatticeDualSymmetry} for a more rigorous proof. 

\subsection{Dual Hamiltonians}
The dual Hamiltonians $\tilde H_{1j}$ are listed as follows. 
\begin{align}
    \tilde H_{11}&=\sum_k\left[ \left( \frac{1-X_{2k}X_{2k+1}}{2} \right)\left( \frac{1-Z_{2k}Z_{2k+1}}{2}+\frac{1-Z_{2k+2}Z_{2k+3}}{2} \right)\right.
    \nonumber\\
    &~~~+\left.\left( \frac{1-X_{2k+1}X_{2k+2}}{2} \right)\left( \frac{1-Z_{2k-1}Z_{2k}}{2}+\frac{1-Z_{2k+1}Z_{2k+2}}{2} \right)
    \right]. \\
    \tilde H_{12}&=\sum_n\left[ 
    \left(\frac{1-Z_nZ_{n+1}}{2}\right)\left(\frac{1-Z_{n+1}Z_{n+2}}{2}\right)
    +\left(\frac{1-X_{n+1}X_{n+2}}{2}\right)\left(\frac{1-X_{n+2}X_{n+3}}{2}\right)
    \right]. \\
    \tilde H_{13}&= -\sum_k(Z_{2k-1}Z_{2k+1}+X_{2k}X_{2k+2}). \\
    \tilde H_{14}&=\sum_n Y_nY_{n+1}. \\
    \tilde H_{15}&= \sum_k(Z_{2k-1}X_{2k}Z_{2k+1}X_{2k+2}+X_{2k}Z_{2k+1}X_{2k+2}Z_{2k+3}). \\
    \tilde H_{16}&=\sum_n Y_nY_{n+2}. 
\end{align}
Ground states of those Hamiltonians are given below. 
\begin{enumerate}
    \item $\tilde H_{11}$ has $4$ degenerate ground states: 
    \begin{align}
        \ket{++\cdots +},\quad
        \ket{--\cdots -},\quad
        \ket{00\cdots 0},\quad \ket{11\cdots 1}. 
    \end{align}
    If we take the ground state $\ket{++\cdots +}$, the remaining symmetry is a $\mathbb{Z}_2$ subgroup generated by $\eta_X$, or $\ex{s}$ in terms of the $D_8$ notation. The remaining symmetry subgroups of the other three ground states are all conjugate to $\ex{s}$. We emphasize that in the case of nonabelian symmetry breaking, different vacuum states\footnote{A ``vacuum state'' is a ground state that satisfies the cluster decomposition property, i.e. a non-cat ground state. } generically preserve distinct but conjugate symmetry subgroups. 
    
    \item $\tilde H_{12}$ has $2$ ground states: 
    \begin{align}
        \bigotimes_k(\ket{00}+\ket{11})_{2k-1,2k},\quad \bigotimes_k(\ket{00}+\ket{11})_{2k,2k+1}. 
    \end{align}
    Both ground states preserve $\ex{s,r^2}$, which is a central $\mathbb{Z}_2\times \mathbb{Z}_2$ subgroup. 
    
    \item $\tilde H_{13}$ has $4$ ground states: $\prod_k\ket{x\alpha}_{2k-1,2k}$ where $x\in\{0,1\}$ and $\alpha\in\{+,-\}$. 
    Considering the ground state $\prod_k\ket{0+}_{2k-1,2k}$, the remaining symmetry is a $\mathbb{Z}_2$ subgroup $\ex{sr}$. 
    
    \item $\tilde H_{14}$ has $2$ ground states. If we denote by $\ket{\tilde x}$ with $x\in\{0,1\}$ the eigenstates of Pauli-$Y$ operator with eigenvalues $(-1)^x$, then these ground states are given by
    \begin{align}
        \ket{\tilde 0\tilde 1 \tilde0 \tilde1\cdots},\quad \ket{\tilde1 \tilde0 \tilde1 \tilde0 \cdots}. 
    \end{align}
    Both ground states preserve the central $\mathbb{Z}_2\times \mathbb{Z}_2$ subgroup $\ex{sr,r^2}$. 
    
    \item If we assume $2N\in 4\mathbb{Z}$, then $\tilde H_{15}$ has $8$ degenerate ground states:  
    \begin{align}
        \bigotimes_k\ket{x\alpha y\beta}_{4k-3,4k-2,4k-1,4k},\quad{\rm s.t.}\quad (-1)^{x+y}\alpha\beta=-1,  
    \end{align}
    where $x,y\in\{0,1\}$ and $\alpha,\beta\in\{+,-\}$. Note that $T^2$, the two-site translation symmetry, is spontaneously broken, hence the effective symmetry group in the ground state subspace is no longer $D_8$. Considering the ground state $\prod_k\ket{0+0-}_{4k-3,4k-2,4k-1,4k}$, the remaining symmetry subgroup is generated by $\eta_Z T^2$. 
    
    \item If we assume $2N\in4\mathbb{Z}$, then $\tilde H_{16}$ has $4$ degenerate ground states: 
    \begin{align}
        \ket{\tilde0\tilde0\tilde1\tilde1}^{\otimes N/2},\quad \ket{\tilde0\tilde1\tilde1\tilde0}^{\otimes N/2},\quad \ket{\tilde1\tilde0\tilde0\tilde1}^{\otimes N/2},\quad \ket{\tilde1\tilde1\tilde0\tilde0}^{\otimes N/2}. 
    \end{align}
    Note that $T^2$ is also spontaneously broken and the  symmetry group can not be regarded as $D_8$. Considering the first ground state $\ket{\tilde0\tilde0\tilde1\tilde1}^{\otimes N/2}$, the remaining symmetry is $\ex{\eta_X\eta_Z,\eta_X T^2}$. 
\end{enumerate}

In particular, the transition between $H_{11}$ and $H_{14}$, which realizes the order-to-disorder transition of $\Rep(H_8)$, is dual to the transition between $\tilde{H}_{11}$ and $\tilde{H}_{14}$, which realizes a $(D_8,\gamma)$ DQCP between two partially SSB phases with incompatible unbroken symmetries $\langle s \rangle$ and $\langle sr, r^2 \rangle$.

\section{DQCPs are dual to generalized LG transitions}
\label{sec:General}

In the previous sections, we showed explicitly that the order-to-disorder transition of the non-invertible symmetry $\mathrm{Rep}(H_8)$ is dual to a transition between symmetry-breaking phases of a $D_8$-symmetric model carrying a nontrivial ’t Hooft anomaly $\gamma \in H^3(D_8, U(1))$. In the $\mathrm{Rep}(H_8)$ frame, the transition is an ordinary symmetry-breaking transition. In the dual $D_8$ frame, however, no trivially symmetric gapped phase exists due to the anomaly, and the direct transition between two incompatible symmetry-breaking phases is naturally interpreted as a deconfined quantum critical point (DQCP).

In this section, we describe the general kinematic structure underlying such DQCPs in $1+1$ dimensions and the generalization of the duality mapping.

\subsection{Kinematic structure of deconfined quantum critical points}
A deconfined quantum critical point (DQCP) is a direct continuous transition between two symmetry-breaking phases with incompatible unbroken symmetries. A necessary condition is a mixed anomaly between the two broken symmetry groups, which forces the symmetry defects of one broken symmetry to carry fractionalized charges of the other \cite{levin2004dqcp,Gaiotto2015,Senthil:2023vqd}. In the following, we consider a DQCP as a gapless theory with symmetry $G$ and 't Hooft anomaly $\omega$, such that $G$ contains at least two incompatible subgroups on each of which $\omega$ trivializes. This definition provides the desired features of the existing DQCP examples, namely 
\begin{enumerate}
    \item there exist partially SSB phases with incompatible unbroken symmetries;
    \item due to the anomaly, the symmetry defects carry fractionalized charges of the other symmetry and no trivially symmetric gapped phase exists;
    \item a gapless theory can serve as the direct transition between the partially SSB phases. 
\end{enumerate}
To be specific, let's focus on 1+1d and consider a symmetry group $G$ with 't Hooft anomaly given by the non-trivial element $\omega \in H^{3}(G,U(1))$, and we label the symmetry by $(G,\omega)$. $G$ contains incompatible subgroups $K_i$ with $K_i \nleq K_j \ (i \neq j)$, and the union of $K_i$s generates $G$. The restriction $\omega|_{K_i}$ is cohomologically trivial. Each $K_i$ may have its own twist $\psi_i\in H^2(K_i,U(1))$. Therefore, we denote the symmetry-breaking phases by $\CM_{K_i,\psi_i}$, where $K_i$ is the unbroken symmetry and $\psi_i \in H^2(K_i, U(1))$ specifies a possible SPT order. When more than two incompatible subgroups $K_i$ are present, the gapless theory describes a deconfined quantum multicritical point (DQMCP). The DQCP reduces to the case with $K_1,K_2$ subgroups in $(G,\omega)$ and $G=\grp{K_1,K_2}$: the low-energy gapped phases are primarily (a) SSB with $K_1$ unbroken and some possible $\psi_1$ SPT order, $\CM_{K_1,\psi_1}$ (b) SSB with $(K_2,\psi_2)$ unbroken, $\CM_{K_2,\psi_2}$ and (c) gapless theory with $G$ and anomaly $\omega$. (c) can be viewed as the DQCP between (a) and (b). In the SSB phase $\CM_{K_1,\psi_1}$, the symmetry defects of $K_2$ carry the fractionalized charges of $K_1$ due to the mixed anomaly, similarly for the SSB with $K_2$ unbroken. 

Gauging an anomaly-free subgroup $L \subset G$ with a twist $\psi \in H^2(L, U(1))$ yields a dual theory with group-theoretical fusion category symmetry $\CD = \CC(G,\omega;L,\psi)$ \cite{etingof2005fusion,Bhardwaj2018},
\begin{equation}
   (G,\omega) 
\overset{\xleftarrow{\text{gauge dual symmetry}}}
{\xrightarrow[\text{gauge }L\text{ with twist } \psi]{}}
\CD=\CC(G,\omega;L,\psi)
\end{equation}
This generalized gauging maps gapped phases of $(G,\omega)$ to gapped phases of the generally non-invertible dual symmetry $\CD$ \cite{yuji2017gauging} and see \appref{app:gappedphase} for the concrete example of $\Rep(H_8)$; the fusion rules of $\CD$ are given in \cite{natale2013fusion} and reviewed in \appref{app:gaugesub}. Physically, the phase $\CM_{K_i,\psi_i}$ is first stacked with the SPT $\psi$ and then gauged with respect to $L$. In particular, choosing $L = K_i$ and $\psi = \psi_i^{-1}$ maps $\CM_{K_i,\psi_i}$ to the fully symmetry-broken phase of $\CD$. Alternatively, if $L$ is the broken part of the symmetry in $\CM_{K_i,\psi_i}$ and $\omega$ trivializes on $L$, then there is no obstruction to gauging $L$ and it maps $\CM_{K_i,\psi_i}$ to a trivially symmetric gapped phase of $\CD$.

More generally, a DQCP of $(G,\omega)$ maps to an order-to-disorder transition of the dual symmetry $\CD$ when $\omega$ trivializes on the broken part of the symmetry in the gapped phase $\CM_{K_i,\psi_i}$. The key questions are then: what is the broken part $L$, and when does the anomaly trivialize on $L$? A natural sufficient condition is that $G$ admits an \textit{exact} factorization $G = K_1 K_2$, meaning every $g \in G$ decomposes \textit{uniquely} as $g = k_1 k_2$ with $k_i \in K_i$ (the direct product $G = K_1 \times K_2$ being a special case), such that $K_1\cap K_2 = \{e\}$. In this case, the broken symmetry in $\CM_{K_1,\psi_1}$ is $K_2$, and the trivialization of $\omega$ on $K_2$ follows from the definition of DQCP. Gauging $K_2$ then maps $\CM_{K_1,\psi_1}$ to a trivially symmetric gapped phase of $\CD$, and vice versa. The dual symmetry $\CD$ is therefore anomaly-free and takes the form $\CD = \Rep(H)$ for some Hopf algebra $H$, which is group-theoretical whenever $\CD$ is.

The exact factorization condition is quite restrictive. Can it be relaxed to product factorization and allow $K_1 \cap K_2 = K_{12}$ for some nontrivial subgroup $K_{12}$? A naive attempt fails: the broken symmetry in $\CM_{K_1,\psi_1}$ is the coset $K_2 \setminus K_{12}$, but gauging all of $K_2$ also breaks $K_{12}$, so the phase after gauging is not trivially symmetric. The remedy is to compensate by stacking a nontrivial SPT on $K_{12}$, or equivalently, by performing a twisted gauging \cite{Lu:2024ytl}. This is because the suitable SPT is invariant under gauging, e.g. cluster state of $\IZ_2\times \IZ_2$ is invariant under gauging $\IZ_2\times \IZ_2$ \cite{Seifnashri24}. 

To illustrate, consider $G = \IZ_2^a \times \IZ_2^b \times \IZ_2^c \times \IZ_2^d$ with type-III anomaly $\omega = a \cup b \cup d$, and take $K_1 = \IZ_2^a \times \IZ_2^b \times \IZ_2^c$, $K_2 = \IZ_2^b \times \IZ_2^c \times \IZ_2^d$, so that $K_{12} = \IZ_2^b \times \IZ_2^c$. Gauging $K_2$ in the phase $\CM_{K_1, 1}$ yields $\CM_{\IZ_2^a \times \IZ_2^d, 1}$, which still has $K_{12}$ broken. Concretely, the two equivalent resolutions are as follows. One can stack an SPT with $\psi = b \cup c$ on the $K_{12}$ part of $\CM_{K_1, 1}$; gauging $K_2$ then maps $\CM_{K_1, \psi}$ to the SPT phase $\CM_{G, \psi}$. Alternatively, one can perform a twisted gauging of $K_2$ with twist $\psi = b \cup c$, which directly maps $\CM_{K_1, 1}$ to $\CM_{G, 1}$.

We now summarize the conditions under which a DQCP of $(G,\omega)$ can be mapped to an order-to-disorder transition of an anomaly-free (possibly non-invertible) symmetry via generalized gauging. Suppose $G$ has two incompatible subgroups $K_1, K_2$ such that every element $g \in G$ can be written as $g = k_1 k_2$ with $k_i \in K_i$. Then the duality to an order-to-disorder transition holds if one of the following conditions is satisfied:
\begin{enumerate}
    \item $K_1 \cap K_2 = \{e\}$ (exact factorization). In this case, gauging $K_2$ directly maps $\CM_{K_1,\psi_1}$ to a trivially symmetric phase of the dual symmetry.
    \item $K_1 \cap K_2 = K_{12}$ for some nontrivial subgroup $K_{12}$, and there exists a nontrivial 2-cocycle $\psi \in H^2(K_{12}, U(1))$ (more precisely non-degenerate 2-cocycle) that can be used to perform a twisted gauging of $K_2$, preventing the spontaneous breaking $K_{12}$.
\end{enumerate}
We note that the above conditions are precisely the conditions under which the group-theoretical fusion category $\CC(G,\omega;L,\psi)$ admits a fiber functor. The physical analysis above therefore provides an alternative, constructive proof of the mathematical classification of fiber functors for group-theoretical fusion categories \cite{ostrik2002module,Lu:2022ver}.

This condition generalizes to DQMCPs with multiple subgroups $K_i$, though determining when $\omega$ trivializes on the broken part of the symmetry becomes more involved. For instance, consider $G = \IZ_2^a \times \IZ_2^b \times \IZ_2^c$ with anomaly $\omega = a \cup b \cup c + a \cup b \cup b$. Although $\omega$ trivializes on each individual $K_i = \IZ_2$, the broken symmetry $\IZ_2^a \times \IZ_2^b$ of the gapped phase $\CM_{\IZ_2^c,1}$ is anomalous and therefore cannot be gauged.

Conversely, one can ask: when does an order-to-disorder transition of a non-invertible symmetry $\Rep(H)$ admit a dual description as a DQCP? The generalized gauging discussed above does not guarantee that the dual symmetry takes the form $\Rep(H)$ for a non-group Hopf algebra $H$. It is therefore useful to reverse-engineer DQCPs by starting from a theory with $\Rep(H)$ symmetry and ungauging an appropriate \textit{proper} subgroup to recover an anomalous group symmetry $(G,\omega)$. A broad class of Hopf algebras for which this reverse construction is guaranteed are those built from the twisted bicrossed product (reviewed in \appref{app:twbiprod}): for these, ungauging a suitable \textit{proper} subgroup always yields a group symmetry carrying a mixed anomaly \cite{natale2003grouphopf}.

It is known that all semisimple Hopf algebras of dimension less than $36$ are group-theoretical \cite{nikshych2008non,Gelaki:2009blp,natale2010hopf,cuadra2017orders,Lu2025nongrp}. Consequently, every anomaly-free non-invertible symmetry with Frobenius-Perron dimension less than $36$ is related to some invertible symmetry $(G,\omega)$ by gauging a proper subgroup $L$ \footnote{If a finite group $G$ has no proper nontrivial subgroup, then its order must be prime, therefore it is abelian, whose possible dual symmetry cannot be non-invertible.}, possibly after stacking an SPT $\psi \in H^2(L, U(1))$. Combined with the preceding analysis, this yields the following result.
\begin{physicstheorem}
\begin{enumerate}
    \item DQCP of $(G,\omega)$ whose incompatible subgroups $K_1, K_2$ satisfy $G = K_1 K_2$ (i.e., every $g \in G$ decomposes as $g = k_1 k_2$ with $k_i \in K_i$) is dual to an order-to-disorder transition of an anomaly-free, possibly non-invertible, group-theoretical fusion category symmetry, provided that either (a) $K_1 \cap K_2 = \{e\}$, i.e. exact factorization, or (b) $K_1 \cap K_2 = K_{12}$ for some nontrivial subgroup $K_{12}$ admitting a nontrivial 2-cocycle $\psi \in H^2(K_{12}, U(1))$ which defines a nondegenerate projective class.
    \item Conversely, any order-to-disorder continuous transition of an anomaly-free group-theoretical non-invertible symmetry is dual to a transition between two partially symmetry-breaking phases of an ordinary group symmetry $G$ with a possible 't Hooft anomaly $\omega$. In particular, this applies to all anomaly-free non-invertible symmetries with $\mathrm{dim} H=\mathrm{FPdim}(\CC)<36$, since these are necessarily group-theoretical. When the 3-cocycle $\omega \in H^3(G, U(1))$ is nontrivial, the dual transition is a DQCP.
\end{enumerate}
\end{physicstheorem}

\begin{proof}
    (1) Under the stated conditions, there exists a generalized gauging that maps the partially SSB phase $\CM_{K_1,\psi_1}$ to a trivially symmetric gapped phase of the dual symmetry and some $\CM_{K_2,\psi_2}$ to fully SSB phase, as shown in the preceding discussion.

    (2) The existence of a trivially symmetric gapped phase on one side of the order-to-disorder transition implies the non-invertible symmetry is anomaly-free and hence of the form $\Rep(H)$ for some semisimple Hopf algebra $H$. Since non-invertible, it excludes all the abelian cases, in particular $\IZ_p$s. Since $H$ is group-theoretical, we have $\Rep(H) \cong \CC(G,\omega;L,\psi)$ for some finite group $G$, 3-cocycle $\omega$, and subgroup $L < G$ . For $\CC(G,\omega;L,\psi)$ to be non-invertible with nontrivial $\omega$, the subgroup $L$ must be a proper nontrivial subgroup of $G$. Ungauging $L$ maps the trivially symmetric phase of $\CC(G,\omega;L,\psi)$ to a partially SSB phase of $(G,\omega)$ in which $L$ is broken. The anomaly $\omega$ must trivialize on the unbroken part, indicating a mixed anomaly between $L$ and the remaining unbroken subgroup. The fully SSB phase of $\CC(G,\omega;L,\psi)$ maps to a different partially SSB phase with an incompatible unbroken symmetry, completing the DQCP structure. When $\dim H=\mathrm{FPdim}(\CC) < 36$, the Hopf algebra is automatically group-theoretical, and the same argument applies.
\end{proof}

The \(\operatorname{Rep}(H_8)\) example realizes the exact-factorization case of Physics Theorem 1: in the dual \((D_8,\gamma)\) frame, take \(K_1=\langle s\rangle\) and \(K_2=\langle sr,r^2\rangle\).
These subgroups obey \(K_1\cap K_2=\{1\}\) and \(D_8=K_1K_2\), with \(\gamma|_{K_1}=\gamma|_{K_2}=1\), so both define admissible partially SSB phases.
Gauging \(L=K_1=\langle s\rangle\) gives the dual symmetry \(\CC(D_8,\gamma;\langle s\rangle,1)\simeq \operatorname{Rep}(H_8)\).
Under this gauging, \(\mathcal M_{\langle s\rangle,1}\) maps to the fully SSB phase of \(\operatorname{Rep}(H_8)\), while \(\mathcal M_{\langle sr,r^2\rangle,1}\) maps to the trivially symmetric phase.
Thus the \((D_8,\gamma)\) DQCP between \(\mathcal M_{\langle s\rangle,1}\) and \(\mathcal M_{\langle sr,r^2\rangle,1}\) is dual to the \(\operatorname{Rep}(H_8)\) order-to-disorder transition between \(H_{11}\) and \(H_{14}\).
Equivalently, the two descriptions share the same SymTFT bulk \(\CD(D_8)^\gamma\simeq \hcenter\), but use different symmetry boundaries.

From the SymTFT perspective, the duality in Physics Theorem~1 corresponds to changing the Lagrangian algebra on the symmetry boundary of the twisted quantum double $D(G)^{\omega}$: the original $\Rep(G)$ boundary realizes the group symmetry $(G,\omega)$, while a different Lagrangian algebra on the same bulk realizes the dual $\Rep(H)$ symmetry. The shared bulk $D(G)^{\omega}$ encodes the equivalence between the two descriptions.

\subsection{Examples}
In this subsection, we present representative examples of DQCPs that admit dual descriptions as order-to-disorder transitions of (possibly non-invertible) symmetries. Starting from an anomalous group symmetry $(G,\omega)$, we gauge suitable anomaly-free subgroups and determine the resulting dual fusion categories $\CC(G,\omega;L,\psi)$, together with the induced mapping between gapped phases. Gapped phases correspond to module categories, which for group-theoretical fusion categories admit an explicit description in terms of subgroup data and double-coset constructions \cite{ostrik2002module,natale2017equivalence,bartsch2024non,bartsch2024non2}; a brief review is given in \appref{app:grpthryH8}.

Our primary focus is on cases where the dual symmetry takes the form $\Rep(H)$ for a non-group Hopf algebra $H$. Most examples below fall into the exact factorization case $K_1 \cap K_2 = \{e\}$; an example realizing the more general condition (2) of Physics Theorem~1, with nontrivial $K_{12}$, was already given in the discussion of $G = \IZ_2^a \times \IZ_2^b \times \IZ_2^c \times \IZ_2^d$ above. The examples are organized as follows:
\begin{enumerate}
    \item Abelian symmetries with nontrivial $3$-cocycle twists;
    \item Tambara-Yamagami categories arising from self-duality under gauging;
    \item Non-abelian symmetries with mixed anomalies;
    \item More general constructions based on twisted bicrossed products.
\end{enumerate}

\subsubsection{Abelian group symmetries with nontrivial twists}
In this subsection, we focus on two classes of abelian examples: (1) $\mathbb{Z}_p^a \times \mathbb{Z}_p^b$ with a type-II twist $a \cup \beta  b$, where $\beta$ is the Bockstein, which reduces to $a\cup b \cup b$ when $p=2$, and (2) $\mathbb{Z}_p^a \times \mathbb{Z}_p^b \times \mathbb{Z}_p^c$ with a type-III twist $a \cup b \cup c$. Here $p$ is a prime number, and $a,b,c$ denote $\mathbb{Z}_p$-valued $1$-cochains appearing in the cup product. We illustrate that gauging the subgroup of an abelian symmetry with appropriate 't Hooft anomaly may lead to interesting non-invertible symmetry.

\paragraph{$\IZ_p^a\times \IZ_p^b$ with type-II twist}
In the context of DQCPs, this class of symmetries was studied in \cite{carolyn2023dqcp} for the case $p=2$. The twist implies that the symmetry defect of $\IZ_p^a$ will carry the fractionalized charge of $\IZ_p^b$. More generally, gauging the $\mathbb{Z}_p^b$ subgroup maps the theory to one with $\mathbb{Z}_{p^2}$ symmetry. The possible gauging procedures and the resulting mappings between gapped phases are summarized below.
\begin{equation}
\begin{array}{c|c}
\hline\hline
\CC=(\IZ^a_p\times \IZ^b_p,\omega_{II}) 
& \CC/\IZ_p^b\cong(\IZ_{p^2},1) 
\\\hline
\CM_{\IZ_p^a,1} 
& \CM_{\IZ_{p^2},1} 
\\ \hline
\CM_{\IZ_p^b,1} 
& \CM_{1,1}
\\ \hline
\CM_{1,1} 
& \CM_{\IZ_p,1}
\\ \hline\hline
\end{array}
\end{equation}
where $\CM_{K,\psi}$ denotes the gapped phase with unbroken symmetry $K$ with a possible 2-cocycle twist $\psi\in H^2(K,U(1))$. The symmetry of the original DQCP is $\CC=(\IZ^a_p\times \IZ^b_p,\omega_{II})$ and the dual symmetry under gauging $H=\IZ_p^b$ is $\CC/H\cong(\IZ_{p^2},1) $. From the first and third columns, the DQCP of $\IZ_p^b$ SSB to $\IZ_p^a$ SSB is mapped $\IZ_{p^2}$ disorder to order transition.

The above construction can be generalized to arbitrarily many $\IZ_{p_i}$ groups with type-II anomaly. For $(\IZ_p\times \IZ_p\times \IZ_p,\omega_{II})$, up to basis transformation, it is equivalent to the case of $(\IZ_p\times \IZ_p,\omega_{II})$, but with an additional $\IZ_p$ factor. 

\paragraph{$\IZ^a_p\times \IZ^b_p\times \IZ^c_p$ with $\omega_{III}$ twist}
This case is less studied in the DQCP context, one may construct the lattice model and study the deconfined multicritical point. From the symmetry perspective, all the proper subgroups of $\IZ^a_p\times \IZ^b_p\times \IZ^c_p$ used below are anomaly free, we can gauge any of them. Here we list some interesting ones:
\begin{equation}
    \begin{array}{c|c|c}
    \hline\hline
         (\IZ_p^a\times \IZ^b_p\times \IZ^c_p,\omega_{III})
         & \CC/\grp{a}\cong \mathrm{H}_3(\mathbb{Z}_p)  
         & \CC/\grp{b,c}\cong \Rep(\mathrm{H}_3(\mathbb{Z}_p))
         \\ \hline
         \CM_{1,1}
         & \CM_{\IZ_p,1}
         & \CM_{\IZ_p^2,1}
         \\ \hline
         \CM_{\grp{a},1}
         & \CM_{1,1}
         & \text{Sym}
         \\ \hline
         \CM_{\grp{b},1}
         & \CM_{\IZ_p^2,1}
         & \CM_{\IZ_p^{\hat{c}},1}
         \\ \hline
         \CM_{\grp{a b},1}
         & \CM_{\IZ_p^2,\psi}
         & \text{Sym}
         \\ \hline
         \CM_{\grp{a,b},1}
         & \CM_{\IZ_p,1}
         & \CM_{\IZ_p^{\hat{c}},1}
         \\ \hline
         \CM_{\grp{b,c},1}
         & \CM_{\IZ_p^3,\psi}
         & \text{Fully SSB}
         \\ \hline
         \CM_{\grp{a,c},1}
         & \CM_{\IZ_p,1}
         & \CM_{\IZ_p^{\hat{b}},1}
         \\ \hline
         \CM_{\grp{ab^2,c},1}
         & \CM_{\IZ_p^3,\psi}
         & \CM_{\IZ_p^{\hat{b}},1}
         \\ \hline\hline
    \end{array}
\end{equation}
where Sym and Fully SSB refer to the fully symmetric trivial state and the fully SSB state, $\mathrm{H}_3(\mathbb{Z}_p)$ is the Heisenberg group over the field $\IZ_p$\footnote{The group $\mathrm{H}_3(\IZ_p)$ has the relation that is given by $\mathrm{H}_3(\mathbb{Z}_p)
=
\langle\, x, y, z
\;\mid\;
x^p = y^p = z^p = 1,\;
xz = zx,\;
yz = zy,\;
xy = z\, yx
\rangle$. An example is $\mathrm{H}_3(\mathbb{Z}_2) \cong D_8$.} \cite{Lanzetta:2022lze}. The $\Rep(\mathrm{H}_3(\mathbb{Z}_p))$ has the fusion rule \cite{Lu:2024ytl,Lu:2024lzf},
\begin{equation}
    \CP_i \otimes \CP_j = \begin{cases}
        \bigoplus_{g\in\IZ_p\times \IZ_p}g&i+j=0\mod p\\ p \CP_{i+j}
    \end{cases},\quad \CP_i\otimes g = g\otimes \CP_i = \CP_i 
\end{equation}
where $\Rep(\mathrm{H}_3(\mathbb{Z}_p))$ has $\IZ_p$ grading, with the trivial sector $\CC_0 = \IZ_p \times \IZ_p$ and non-trivial sectors containing only one simple object $\CP_i$ of quantum dimension $p$. 
Hence, $\Rep(\mathrm{H}_3(\mathbb{Z}_p)) \cong \CC(\IZ_p^3,\omega_{III},\IZ_p\times \IZ_p,1)$. In particular, the $\Rep(\mathrm{H}_3(\mathbb{Z}_p))$ invariant SPTs correspond to the SSB phase with unbroken symmetry $\grp{a^ib^jc^k}$, $\CM_{\grp{a^ib^jc^k}}$ in the original frame, where $i\ne0$. This is because $\IZ_p^3$ admits the exact factorization of $\grp{b,c},\grp{a^ib^jc^k}$ with $i \ne 0$. In particular, from the first and last column, the DQCP between unbroken $\grp{a}$ SSB and unbroken $\grp{b,c}$ SSB is mapped to $\Rep(\mathrm{H}_3(\mathbb{Z}_p))$ disorder-to-order transition.

\subsubsection{Tambara-Yamagami fusion category}
The Tambara-Yamagami fusion category $\TY(A,\chi,\epsilon)$ provides a categorical description of self-duality symmetries arising from gauging an abelian group $A$. The simple objects are the invertible elements $g\in A$ and a non-invertible self-duality line $\CN$, with the fusion rule,
\begin{equation}
    g\otimes h = gh,\quad \CN\otimes g=g\otimes \CN = \CN, \quad \CN\otimes \CN = \sum_{g\in A} g
\end{equation}
Additional data $\chi,\epsilon$ classify the $F$-symbols. Here, $\chi: A \times A \to U(1)$ is a symmetric bicharacter specifying the identification between the original symmetry and its dual quantum symmetry, while $\epsilon=\pm1 \in H^3(\mathbb{Z}_2,U(1))$ is the Frobenius--Schur indicator associated with the duality defect line. See \cite{Thorngren2019Fusion,Thorngren:2021yso,Zhang:2023wlu,Antinucci:2023ezl,Lu:2024ytl} for more applications of Tambara-Yamagami fusion category in physics. For example, the Kramers-Wannier duality in the transverse field Ising model corresponds to $\TY(\IZ_2,\chi,+1)$, where $\chi(g,h)=(-1)^{gh}$ and it is equivalent to the Ising fusion category.

A Tambara-Yamagami fusion category is group-theoretical if it admits a Lagrangian subgroup $L \subset A$ \cite{Gelaki:2009blp,Sun2023TYanomaly}. Moreover, if a Tambara-Yamagami category admits a fiber functor, namely if the associated self-duality symmetry is anomaly free, then the category is necessarily group-theoretical \cite{natale2003grouphopf}. Recall that the necessary condition for a fusion category to be anomaly-free is that the quantum dimensions of all the line operators should be integers. The anomaly of Tambara-Yamagami category was studied in \cite{Zhang:2023wlu,Antinucci:2023ezl}. Then $|A|$ the order of $A$ should be a perfect square is a necessary condition. We list the group-theoretical construction of the Tambara-Yamagami category as follows,
\begin{equation}
	\begin{array}{c|c}
		\hline\hline
		\text{TY fusion category} & \text{Group-theoretical construction} \\ \hline
		\TY(\IZ_2\times\IZ_2,\chi_\od,+1)& \makecell{\CC(D_8,1,\IZ_2^s,1), \,\CC(D_8,1,D_8,1)\cong \Rep(D_8), \\ \CC(\IZ_2^3,\omega_{III},\IZ_2^2,1) \text{\cite{goff_gauge_2007,yuji2017gauging}}}  \\ \hline
		\TY(\IZ_2\times\IZ_2,\chi_\dd,+1)& \Rep(H_8) \cong \CC(D_8,\gamma,\IZ_2^s,1) \text{\cite{natale2017equivalence,etingof2021tensor,wang2024gauging,Lu:2025gpt}} \\\hline
		\TY(\IZ_2\times\IZ_2,\chi_\od,-1)& \makecell{\CC(Q_8,1,Q_8,1) \cong \Rep(Q_8), \\\cong\CC(\IZ_2^3,\omega_{III+II+II},\IZ_2^2,1) \text{\cite{goff_gauge_2007,Li:2024fhy}}} \\ \hline
		\TY(\IZ_N\times \IZ_N,\chi_\od,+1)   & \makecell{\CC((\IZ_N\times \IZ_N)\rtimes \IZ_2^s,1,\IZ_N,1) \\\CC((\IZ_N\times \IZ_N)\rtimes \IZ_2^s,1,\IZ_N\times \IZ_N,\psi)} \text{\cite{Lu2025nongrp}} \\ \hline\hline
	\end{array}
\end{equation}
where $(\IZ_N\times \IZ_N)\rtimes \IZ_2^s = \grp{a,b,s|a^N=b^N=s^2=1,ab=ba,a s = s b}$ and $\psi \in H^2(\IZ_N\times \IZ_N,U(1))$. Generally $\Rep(G) \cong \CC(G,1,G,1)$, we refer the interested readers to \cite{goff_gauge_2007} for more equivalent descriptions of order $8$ symmetries through their twisted quantum doubles. $\Rep(Q_8)$ is also equivalent via gauging to $(D_8, \omega)$ with anomaly $\omega = b$ or $b + c$, where $b$ and $c$ denote the generators of the two nontrivial $\IZ_2$ factors in $H^3(D_8, U(1))$, with $b$ corresponding to the mixed-type \cite{goff_gauge_2007,uribe2017classification}. We note that $ \TY(\IZ_N\times \IZ_N,\chi_\od,1) \cong \CC((\IZ_N\times \IZ_N)\rtimes \IZ_2^s,1,\IZ_N\times \IZ_N,\psi)$ is notable, as it corresponds to a twisted gauging of a normal subgroup, producing a non-invertible symmetry described by the representation category of a Hopf algebra. We will also discuss the group-theoretical construction of a family of Hopf algebra that was constructed by Masuoka \cite{masuoka1997calculations} and relates to the Tambara-Yamagami fusion category in later subsection.

\subsubsection{Some $\IZ_2$-extension of Tambara-Yamagami fusion category}
In this subsection, we list interesting examples involving non-abelian groups, and its dual fusion category which describes self-duality under gauging a non-invertible symmetry. 
The TY category of $\IZ_2\times \IZ_2$ can be isomorphic to $\Rep(H_8)$ or $\Rep(D_8)$ depending on the data $\chi,\epsilon$ as listed in the previous subsection. One can further consider gauging the above non-invertible symmetries, some gauging may lead to self-duality, which is described by $\IZ_2$-extension of $\Rep(H_8)$ or $\Rep(D_8)$. In particular, the self-duality fusion category $\CC$ under gauging the algebra object $\CA = 1+ab+\CN$ in $\Rep(H_8)$ or $\Rep(D_8)$ contains two sectors $\CC_0,\CC_1$,
\begin{equation}
    \CC_0 = \Rep(D_8)\text{ or } \Rep(H_8), \quad \CC_1 = \{\CD_1,\CD_2\}
\end{equation}
with the fusion rule,
\begin{equation}
    \CD_i \otimes\CD_i = 1\oplus ab\oplus\CN,\quad \CD_i\otimes \CN = \CD_1\oplus\CD_2,\quad g\otimes D_i  = \CD_i
\end{equation}
where $i=1,2$ and another version is $\CD_i \otimes \CD_{\bar{i}} = 1\oplus ab \oplus \CN$, where $\bar{1}=2,\bar{2}=1$. Such fusion categories are multiplicity free and classified in \cite{vercleyen2024low}. Some of them have the group-theoretical construction. To be specific, according to \cite{Lu:2025gpt} and following their notations,
\begin{equation}
\renewcommand{\arraystretch}{1.6}
\begin{array}{c|c}
\hline\hline
\textbf{$\IZ_2$-extension of }\mathrm{Rep}\,D_8
& \textbf{$\IZ_2$-extension of }\mathrm{Rep}\,H_8
\\ \hline\hline
\mathcal{E}^{(1,+)}_{\mathbb{Z}_2,I}\,\mathrm{Rep}\,D_8
= \mathcal{C}(D_{16}, \omega_{0,0}; \widetilde{\mathbb{Z}}_2^{\,s}, 1)
&
\mathcal{E}^{(1,+)}_{\mathbb{Z}_2,I}\,\mathrm{Rep}\,H_8
= \mathcal{C}(D_{16}, \omega_{2,0}; \widetilde{\mathbb{Z}}_2^{\,s}, 1)
\\
\mathcal{E}^{(2,+)}_{\mathbb{Z}_2,I}\,\mathrm{Rep}\,D_8
= \mathcal{C}(D_{16}, \omega_{4,0}; \widetilde{\mathbb{Z}}_2^{\,s}, 1)
&
\mathcal{E}^{(2,+)}_{\mathbb{Z}_2,I}\,\mathrm{Rep}\,H_8
= \mathcal{C}(D_{16}, \omega_{6,0}; \widetilde{\mathbb{Z}}_2^{\,s}, 1)
\\
\mathcal{E}^{(1,-)}_{\mathbb{Z}_2,I}\,\mathrm{Rep}\,D_8
= \mathcal{C}(D_{16}, \omega_{0,1}; \widetilde{\mathbb{Z}}_2^{\,s}, 1)
&
\mathcal{E}^{(1,-)}_{\mathbb{Z}_2,I}\,\mathrm{Rep}\,H_8
= \mathcal{C}(D_{16}, \omega_{2,1}; \widetilde{\mathbb{Z}}_2^{\,s}, 1)
\\
\mathcal{E}^{(2,-)}_{\mathbb{Z}_2,I}\,\mathrm{Rep}\,D_8
= \mathcal{C}(D_{16}, \omega_{4,1}; \widetilde{\mathbb{Z}}_2^{\,s}, 1)
&
\mathcal{E}^{(2,-)}_{\mathbb{Z}_2,I}\,\mathrm{Rep}\,H_8
= \mathcal{C}(D_{16}, \omega_{6,1}; \widetilde{\mathbb{Z}}_2^{\,s}, 1)
\\ \hline\hline
\mathcal{E}^{1}_{\mathbb{Z}_2,II}\,\mathrm{Rep}\,D_8
= \mathcal{C}(SD_{16}, \omega_{0}; \widetilde{\mathbb{Z}}_2^{\,s}, 1)
&
\mathcal{E}^{1}_{\mathbb{Z}_2,II}\,\mathrm{Rep}\,H_8
= \mathcal{C}(SD_{16}, \omega_{2}; \widetilde{\mathbb{Z}}_2^{\,s}, 1)
\\
\mathcal{E}^{2}_{\mathbb{Z}_2,II}\,\mathrm{Rep}\,D_8
= \mathcal{C}(SD_{16}, \omega_{4}; \widetilde{\mathbb{Z}}_2^{\,s}, 1)
&
\mathcal{E}^{2}_{\mathbb{Z}_2,II}\,\mathrm{Rep}\,H_8
= \mathcal{C}(SD_{16}, \omega_{6}; \widetilde{\mathbb{Z}}_2^{\,s}, 1)\\
\hline\hline
\end{array}
\end{equation}
where $D_{16},\, SD_{16}$ are parameterized as,
\begin{equation}
D_{16}=\langle r, s \mid r^{8} = s^{2} = 1,\;
s r s = r^{-1} \rangle ,\quad SD_{16}=\langle r, s \mid r^{8} = s^{2} = 1,\;
s r s = r^{3} \rangle , 
\end{equation}
and the anomaly of $D_{16}$ is given by $\omega_{m,n}\in H^3(D_{16},U(1))$,
\begin{equation}
\begin{aligned}
\omega_{m,n}(g_1,g_2,g_3)
=\exp\!\Big[
\frac{2\pi m\ii}{8^2}(-1)^{j_1}i_1\big(i_2+(-1)^{j_2}i_3-[i_2+(-1)^{j_2}i_3]_8\big)
\\+\frac{2\pi n \ii}{4}i_1\big(j_2+j_3-[j_2+j_3]_2\big)
\Big],
\quad m\in\mathbb Z_8,\ n\in\mathbb Z_2 ,
\end{aligned}
\end{equation}
where the group elements $g\in D_{16}$ (and also $SD_{16}$) is parameterized as $g=r^is^j$ and $[\cdots]_n$ denotes mod $n$. The anomaly of $SD_{16}$ is,
\begin{equation}
\begin{aligned}
\omega_{2n}(g_1,g_2,g_3)
=\exp\Bigg\{
&\frac{2\pi n \ii}{8^{2}}\, i_1
\Big([3^{j_1} i_2]_8 + [3^{j_1+j_2} i_3]_8- [3^{j_1} i_2 + 3^{j_1+j_2} i_3]_8\Big)\\
&\quad + \frac{2\pi n \ii}{8^{2}}\,[3 i_1]_8
\Big([3^{j_1+1} i_2]_8 + [3^{j_1+j_2+1} i_3]_8- [3^{j_1+1} i_2 + 3^{j_1+j_2+1} i_3]_8\Big)\Bigg\}.
\end{aligned}
\end{equation}
Moreover, some of the above fusion categories are anomaly-free. In particular, anomaly-free $\mathcal{E}^{(2,+)}_{\mathbb{Z}_2,I}\,\mathrm{Rep}\,D_8
= \mathcal{C}(D_{16}, \omega_{4,0}; \widetilde{\mathbb{Z}}_2^{\,s}, 1),\,\mathcal{E}^{(1,-)}_{\mathbb{Z}_2,I}\,\mathrm{Rep}\,D_8
= \mathcal{C}(D_{16}, \omega_{0,1}; \widetilde{\mathbb{Z}}_2^{\,s}, 1)$ have non-trivial 3-cocycles on $D_{16}$. Such fusion categories can be dual to DQCP with $D_{16}$ symmetry with corresponding anomaly.

\subsubsection{General construction from twisted bicrossed product}
According to \cite{natale2003grouphopf}, a broad class of group-theoretical Hopf algebras is obtained from \textit{twisted bicrossed products}. We refer to \cite{natale2003grouphopf} and Appendix~\ref{app:twbiprod} for the construction. In this framework,
\begin{equation}
    \Rep(H)=\CC(G,\omega(\tau,\sigma),K,1),
\end{equation}
where $G=KL$ is an exact factorization, i.e. every $g\in G$ decomposes uniquely as $g=k\ell$ with well-defined projection maps $G\xrightarrow{\pi}K$ and $G\xrightarrow{p}L$. $K \cap L = \{e\}$ and falls in our case (a) in Theorem 1. The $3$-cocycle $\omega\in H^3(G,U(1))$ is determined by compatible $2$-cocycles $\tau$ and $\sigma$ which determines the Hopf algebra extension (reviewed in Appendix~\ref{app:twbiprod}).

Physically, the exact factorization $G=KL$ implies that $(G,\omega(\tau,\sigma))$ describes a $1+1$d system with a mixed anomaly trivialized on both $K$ and $L$, allowing a DQCP between phases with unbroken $K$ and $L$. Gauging $K$ produces the non-invertible symmetry.

Such (quasi-)Hopf algebra extension have recently been used to understand anomaly resolution of non-invertible symmetries, namely how anomalous categorical symmetries can be embedded into anomaly-free ones \cite{Perez-Lona:2025ncg,Lu2025nongrp}.

The interesting examples are,
\begin{enumerate}
    \item A family of Hopf algebras of order $2N^2$. \begin{equation}
    \Rep(H_{2N^2}(\xi)) = \CC((\IZ_N\times \IZ_N)\rtimes\IZ_2,\omega_\xi,\IZ_2,1)
\end{equation}
where \cite{shimizu2012some}
\begin{equation}
\omega_\xi\bigl((i_1,j_1,a_1),(i_2,j_2,a_2),(i_3,j_3,a_3)\bigr)
=
\begin{cases}
1, & \text{if } a_3 = 0, \\[4pt]
\xi^{\,j_1 i_2}, & \text{if } a_3 \neq 0 \text{ and } a_2 = 0, \\[4pt]
\xi^{\,i_1 j_1 + i_1 i_2}, & \text{if } a_3 \neq 0 \text{ and } a_2 \neq 0,
\end{cases}
\end{equation}
where $(i_s,j_s) \in \IZ_N\times \IZ_N, a_s\in \IZ_2$, $s=1,2,3$. For example, when $N=2$, $H_{2\times 2^2}(+1) \cong \IC D_8$, while $\Rep(H_{2\times 2^2}(-1))\cong \Rep(H_8) \cong \CC((\IZ_2\times \IZ_2)\rtimes\IZ^s_2,\omega_{-1},\IZ^s_2,1)$, which reduce to the previous case. 
\item A family of Hopf algebras of dimension $N^3$ with odd $N$. \begin{equation}
    \Rep(H_{N^3}(\xi,\zeta)) = \CC((\IZ_N\times \IZ_N)\rtimes\IZ^a_N,\omega_{(\xi,\zeta)},\IZ^a_N,1)
\end{equation}
where \cite{shimizu2012some},
\begin{equation}
    \omega_{(\xi,\zeta)}\bigl((i_1,j_1,a_1),(i_2,j_2,a_2),(i_3,j_3,a_3)\bigr)
=
\left(\lambda_{j_1+j_2}\,\lambda_{j_1}^{-1}\,\lambda_{j_2}^{-1}\right)^{a_3}
\,\zeta^{\,a_3 j_1 i_2 + \binom{a_3}{2} j_1 j_2}
\end{equation}
where $\lambda$ is an $N$-th root of $\xi^{-1}$ and $\binom{m}{n}$ denotes the binomial coefficient. Therefore, one could construct the DQCP between SSB phases of unbroken $\IZ_N$ and $\IZ_N\times \IZ_N$ symmetry.

The dual of $H_{N^3}(\xi,\zeta)$ is constructed in \cite{castano2018cocycle}, which also admits group-theoretical construction. In particular, the representation category of $H_{N^3}(\xi,\zeta)^*$ is equivalent to $\CC(\mathrm{H}_3(\IZ_N), \widetilde\omega_{\xi,\zeta}, \IZ_N \times \IZ_N,1)$, where $\mathrm{H}_3(\IZ_N)$ is the Heisenberg group over finite field $\mathbb{Z}_N$, and the explicit 3-cocycle is given in \cite{castano2018cocycle}. 
\end{enumerate}

We note that such twisted bicrossed product is also called Knit product or matched pair when the twists are trivial. The semi-direct product is a special case of it. This is used to construct non-abelian group from abelian groups. When applying to the 2+1d topological order, such construction can result in non-abelian surface code from the abelian ones \cite{Tantivasadakarn:2022hgp,Manjunath:2026swm}.  

\section*{Acknowledgements}
We thank Arkya Chatterjee, Zhengdi Sun, Vibhu Ravindran, Luisa Eck, Zhenghan Wang and Bowen Yang for helpful discussions. 
S.L. acknowledges support from the Chinese Academy of Sciences (CAS) under Grant No. YSBR-150
and a startup fund from the Institute of Physics, CAS. D.C.L. is supported by the Simons Collaboration on Ultra-Quantum Matter, which is a grant from the Simons Foundation (grant No. 651440).
X.C. is supported by the Simons collaboration on `Ultra-Quantum Matter'' (grant number 651438), the Simons Investigator Award (award ID 828078), the Institute for Quantum Information and Matter, and the Walter Burke Institute for Theoretical Physics at Caltech.
\appendix

\section{Twisted Bicrossed Product Construction of Hopf Algebra}\label{app:twbiprod}
This appendix follows \cite{natale2003grouphopf}. Let's consider two finite groups $K,L$, with right action of $K$ on $L$, $\triangleleft:L\times K \rightarrow L$ and left action of $L$ on $K$, $\triangleright: L\times K \rightarrow K$,
\begin{align}
    \ell_i \triangleright (k_i k_j)
= (\ell_i \triangleright k_i)\bigl((\ell_i \triangleleft k_i) \triangleright k_j\bigr),\\
(\ell_i \ell_j) \triangleleft k_i
= \bigl(\ell_i \triangleleft (\ell_j \triangleright k_i)\bigr)
  (\ell_j \triangleleft k_i).
\end{align}
for all $k_i \in K,\ell_i \in L$. These groups and left/right actions give the a matched pair of groups. Moreover, $G = KL$ is an exact factorization of $G$, meaning that there are well-defined maps $G\xrightarrow{\pi} K,G\xrightarrow{p}L$, where,
\begin{equation}
    p(k\ell) = \ell,\, \pi(k\ell) = k,\quad  k\in K,\ell\in L
\end{equation}
and $\ell k = (\ell \triangleright k)(\ell \triangleleft k)$. Physically, we consider $K,L$ as two incompatible subgroups of $G$. 
Consider the left action of \(K\) on \(\kk^{L}\), where $\kk$ is a field, defined by
\[
k \cdot \phi(\ell) = \phi(\ell \triangleleft k),
\qquad \phi \in \kk^{L}.
\]
Let
\[
\sigma : K \times K \longrightarrow (\kk^L)^{\times}
\]
be a normalized \(2\)-cocycle. Writing
\[
\sigma = \sum_{\ell \in L} \sigma_{\ell}\,\delta_{\ell},
\]
we require that for all \(\ell \in L\) and \(k_i, k_j, k_m \in K\),
\begin{align}
\sigma_{\ell \triangleleft k_i}(k_j, k_m)\,
\sigma_{\ell}(k_i, k_j k_m)
&=
\sigma_{\ell}(k_i k_j, k_m)\,
\sigma_{\ell}(k_i, k_j), \\
\sigma_{\ell}(k_i, 1) &= 1 = \sigma_{\ell}(1, k_i). 
\end{align}

Dually, consider the right action of \(L\) on \(\kk^{K}\), given by
\[
(\psi\cdot \ell)(k) = \psi(\ell \triangleright k),
\qquad \psi \in \kk^{K}.
\]
Let
\[
\tau = \sum_{k \in K} \tau_k\,\delta_k
\]
be a normalized \(2\)-cocycle
\[
\tau : L \times L \longrightarrow (\kk^K)^{\times}.
\]
Equivalently, for all \(k \in K\) and \(\ell_i, \ell_j, \ell_m \in L\),
\begin{align}
\tau_{k}(\ell_i \ell_j, \ell_m)\,
\tau_{k \triangleright \ell_m}(\ell_i, \ell_j)
&=
\tau_{k}(\ell_j, \ell_m)\,
\tau_{k}(\ell_i, \ell_j \ell_m), \\
\tau_{k}(\ell_i, 1) &= 1 = \tau_{k}(1, \ell_i).
\end{align}

We further assume that \(\sigma\) and \(\tau\) satisfy the following compatibility conditions: for all \(k_i, k_j \in K\) and
\(\ell_i, \ell_j \in L\),
\begin{align}
\sigma_{\ell_i \ell_j}(k_i, k_j)\,
\tau_{k_i k_j}(\ell_i, \ell_j)
&=
\tau_{k_i}(\ell_i, \ell_j)\,
\tau_{k_j}\bigl(
\ell_i \triangleleft (\ell_j \triangleright k_i),
\ell_j \triangleleft k_i
\bigr) \notag \\
&\times \sigma_{\ell_i}\bigl(
\ell_j \triangleright k_i,
(\ell_j \triangleleft k_i) \triangleright k_j
\bigr)\,
\sigma_{\ell_j}(k_i, k_j), \\
\sigma_{1}(k_i, k_j) &= 1,
\qquad
\tau_{1}(\ell_i, \ell_j) = 1. 
\end{align}
Therefore the vector space
\[
A = \kk^{L} \otimes \kk K
\]
acquires the structure of a (semisimple) Hopf algebra, with algebra structure given by a crossed product and coalgebra structure given by a crossed coproduct. We shall denote this Hopf algebra by
\[
A = \kk^{L} \,{}^{\tau}\!\#_{\sigma}\, \kk K,
\]
and write \(\delta_{\ell} k\) for the elementary tensor \(\delta_{\ell} \otimes k \in A\).
The multiplication and comultiplication in \(A\) are then determined by
\begin{align}
(\delta_{\ell_i }k_i)(\delta_{\ell_j }k_j)
&=
\delta_{\ell_i \triangleleft k_i,\, \ell_j}\,
\sigma_{\ell_i}(k_i, k_j)\,
\delta_{\ell_i}k_i k_j, 
\\[6pt]
\Delta(\delta_{\ell_i}k_i)
&=
\sum_{\ell_m \ell_n = \ell_i}
\tau_{k_i}(\ell_m, \ell_n)\,
\delta_{\ell_m}(\ell_n \triangleright k_i)
\otimes
\delta_{\ell_n}k_i. 
\end{align}
for all \(\ell_i \in L\) and \(k_i \in K\).

There is an exact sequence of Hopf algebras
\[
1 \longrightarrow \kk^{L} \longrightarrow A \longrightarrow \kk K \longrightarrow 1,
\]
and conversely, every Hopf algebra \(A\) fitting into an exact sequence of this form
is isomorphic to \(\kk^{L} \,{}^{\tau}\!\#_{\sigma}\, \kk K\) for suitable actions
\(\triangleright, \triangleleft\) and cocycles \(\sigma\) and \(\tau\).

Finally, the Hopf algebra extension is classified by $\operatorname{Opext}(\kk^G,\kk F)$, whose elements $[\tau,\sigma]$ is mapped to the 3-cocycle of $G$,
\begin{equation}
    \omega(\tau,\sigma)(g_i, g_j, g_k)
=
\tau_{\pi(g_k)}\!\bigl(
p(g_i) \triangleleft \pi(g_j),\,
p(g_j)
\bigr)\,
\sigma_{p(g_i)}\!\bigl(
\pi(g_j),\,
p(g_j) \triangleright \pi(g_k)
\bigr),
\quad
g_i, g_j, g_k \in G.
\end{equation}
The 3-cocycle $\omega$ is trivial when restricting to the subgroups $K,L$. Recall that the Hopf algebra from the twisted bicrossed product is denoted as $A= \kk^{L} \,{}^{\tau}\!\#_{\sigma}\, \kk K$, and $G=KL$, then by \cite{natale2003grouphopf},
\begin{equation}
    \Rep(A) = \Rep(\kk^{L} \,{}^{\tau}\!\#_{\sigma}\, \kk K)  = \CC(G,\omega(\tau,\sigma),K,1)
\end{equation}
For instance, $G=L\rtimes K$ is the case where $\tau,\sigma$ are trivial, and $\triangleright: L\times K\rightarrow K$ is trivial, but $\triangleleft: L\times K\rightarrow L$ is non-trivial. An explicit example is $K=\IZ_2, L=\IZ_n$ and $\Rep(D_{2n}) = \CC(\IZ_{n}\rtimes \IZ_2,1,\IZ_2,1)$. There are more non-trivial examples of twisted bicrossed product in the main text.

We provide detailed construction of the examples in the main text as follows. 

\paragraph{A family of Hopf algebras of dimension $2N^2$}
According to \cite{masuoka1997calculations}, a family of Hopf algebras of dimension \(2N^{2}\) with $N>1$ can be constructed by choosing \(K=\mathbb{Z}_{2}\) and 
\(L=\mathbb{Z}_{N}\times \mathbb{Z}_{N}\). The nontrivial structure arises from an action of \(K\) on \(L\) given by
\[
(i,j)\triangleleft 0 = (i,j), \qquad (i,j)\triangleleft 1 = (j,i),
\]
where \((i,j)\in \mathbb{Z}_{N}\times \mathbb{Z}_{N}\), and \(0,1\in \mathbb{Z}_{2}\) denote the identity and 
nontrivial elements of \(K\), respectively. \cite{masuoka1997calculations} shows the extension is classified by $\IZ_N$. We denote $\xi$ as the $N$-th root of unit, and denote the Hopf algebra as $H_{2N^2}(\xi)$. The order of $\xi$ distinguish different Hopf algebras. Moreover,
\begin{equation}
    \Rep(H_{2N^2}(\xi)) = \CC((\IZ_N\times \IZ_N)\rtimes\IZ_2,\omega_\xi,\IZ_2,1)
\end{equation}
where \cite{shimizu2012some}
\begin{equation}
\omega_\xi\bigl((i_1,j_1,a_1),(i_2,j_2,a_2),(i_3,j_3,a_3)\bigr)
=
\begin{cases}
1, & \text{if } a_3 = 0, \\[4pt]
\xi^{\,j_1 i_2}, & \text{if } a_3 \neq 0 \text{ and } a_2 = 0, \\[4pt]
\xi^{\,i_1 j_1 + i_1 i_2}, & \text{if } a_3 \neq 0 \text{ and } a_2 \neq 0,
\end{cases}
\end{equation}
where $(i_s,j_s) \in \IZ_N\times \IZ_N, a_s\in \IZ_2$, $s=1,2,3$. For example, when $N=2$, $H_{2\times 2^2}(+1) \cong \IC D_8$, while $\Rep(H_{2\times 2^2}(-1))\cong \Rep(H_8) \cong \CC((\IZ_2\times \IZ_2)\rtimes\IZ^s_2,\omega_{-1},\IZ^s_2,1)$, which reduce to the previous case.

Generally, the DQCP characterized by the 3-cocycle $\omega_\xi$ describes a transition between symmetry-broken phases with unbroken $\IZ_2$ and unbroken $\IZ_N \times \IZ_N$, which is dual to the order-to-disorder transition of $\Rep(H_{2N^2}(\xi))$.

\paragraph{A family of Hopf algebras of dimension $N^3$} Consider odd integer $N$, \cite{masuoka1997calculations} constructed a family of Hopf algebras with dimension $N^3$. The right action of $K=\IZ_N$ on $L=\IZ_N\times \IZ_N$ is given by,
\begin{equation}
(i,j) \triangleleft a = (i + a j,\, j),
\qquad a,i,j \in \mathbb{Z}_N .
\end{equation}
Such extension is classified by $\IZ_N\times \IZ_N$, labeled by $(\xi,\zeta)$ which are $N$-th root of unit, then,
\begin{equation}
    \Rep(H_{N^3}(\xi,\zeta)) = \CC((\IZ_N\times \IZ_N)\rtimes\IZ^a_N,\omega_{(\xi,\zeta)},\IZ^a_N,1)
\end{equation}
where \cite{shimizu2012some},
\begin{equation}
    \omega_{(\xi,\zeta)}\bigl((i_1,j_1,a_1),(i_2,j_2,a_2),(i_3,j_3,a_3)\bigr)
=
\left(\lambda_{j_1+j_2}\,\lambda_{j_1}^{-1}\,\lambda_{j_2}^{-1}\right)^{a_3}
\,\zeta^{\,a_3 j_1 i_2 + \binom{a_3}{2} j_1 j_2}
\end{equation}
where $\lambda$ is an $N$-th root of $\xi^{-1}$ and $\binom{m}{n}$ denotes the binomial coefficient. Therefore, one could construct the DQCP between SSB phases of unbroken $\IZ_N$ and $\IZ_N\times \IZ_N$ symmetry.

The dual of $H_{N^3}(\xi,\zeta)$ is constructed in \cite{castano2018cocycle}, which also admits group-theoretical construction. In particular, the representation category of the Hopf algebra $H_{N^3}(\xi,\zeta)^*$ is equivalent to $\CC(\mathrm{H}_3(\IZ_N), \widetilde \omega_{\xi,\zeta}, \IZ_N \times \IZ_N,1)$, where $\mathrm{H}_3(\IZ_N)$ is the Heisenberg group over finite field $\mathbb{Z}_N$, and the explicit 3-cocycle is given in \cite{castano2018cocycle}.

\paragraph{List of Hopf algebras of small orders}
In the main text, we comment on when a deconfined quantum critical point with an ordinary invertible symmetry admits a dual description as an order-to-disorder transition of a non-invertible symmetry. Such a duality arises when the non-invertible symmetry is described by the representation category of a group-theoretical Hopf algebra, which is typically constructed via a twisted bicrossed product. In this case, the dual symmetry is an ordinary group endowed with a nontrivial $3$-cocycle ’t~Hooft anomaly. It is known that all semisimple Hopf algebras of dimension $< 36$ are group-theoretical. Among them, the Hopf algebras that are neither group algebras nor dual group algebras are:
\begin{itemize}
    \item \textbf{Order $8$.}  
    Up to isomorphism, the only semisimple Hopf algebra of dimension $8$ that is neither
    commutative nor cocommutative is the Kac--Paljutkin Hopf algebra $H_8$.
    Its representation category $\Rep(H_8)$ is group-theoretical, and its group-theoretical
    construction has been discussed in the previous sections \cite{wakui2018braided}.

    \item \textbf{Order $12$.}  
    The semisimple Hopf algebras of dimension $12$ that are neither commutative nor
    cocommutative are isomorphic to the Hopf algebras $A_\pm$ \cite{fukuda1997semisimple}.
    The group of grouplike elements satisfies $G(A_+) \cong \IZ_2 \times \IZ_2$, while
    $G(A_-) \cong \IZ_4$. The representation categories $\Rep(A_\pm)$ each contain four
    simple objects of quantum dimension $1$ and two simple objects of quantum dimension $2$.

    \item \textbf{Order $16$.}  
    Among semisimple Hopf algebras of dimension $16$ that are neither commutative nor
    cocommutative, their representation categories contain simple objects of one of the
    following two types \cite{kashina2000hopf16}:
    \begin{enumerate}
        \item eight simple objects of quantum dimension $1$ and two of quantum dimension $2$;
        \item four simple objects of quantum dimension $1$ and three of quantum dimension $2$.
    \end{enumerate}

    \item \textbf{Orders $18$ and $20$.}  
    Semisimple Hopf algebras of dimensions $18$ and $20$ are of the form $pq^2$ with distinct
    primes $p$ and $q$ and are group-theoretical \cite{natale1999semisimple,natale2004semisimple}.

    \item \textbf{Order $24$.}  
    Every semisimple Hopf algebra of dimension $24$ is group-theoretical and fits into an
    abelian exact sequence, possibly after a cocycle twist of the multiplication or the
    comultiplication \cite{natale2010hopf}.

    \item \textbf{Order $27$.}  
    Semisimple Hopf algebras of dimension $27 = p^3$ are group-theoretical and arise from
    cocycle deformations of group algebras
    \cite{masuoka1997calculations,castano2018cocycle}.
\end{itemize}
At dimension $36$, there exists a non-group-theoretical semisimple Hopf algebra constructed in \cite{nikshych2008non,Gelaki:2009blp,natale2010hopf,cuadra2017orders}. This example has recently been generalized to dimension $p^{2}q^{2}$, where $p$ and $q$ are prime numbers \cite{galindo2024integral}. From a physical perspective, such non-group-theoretical Hopf algebras give rise to anomaly-free non-invertible symmetries that are not dual to any ordinary (invertible) group symmetry. A corresponding lattice realization of symmetric gapped phases with such a symmetry has recently been constructed in \cite{Lu2025nongrp}.

Beyond dimension $36$, there exist infinite families of semisimple Hopf algebras that are group-theoretical, some of which have been discussed in the previous sections, e.g. $H_{2N^2}(\xi),H_{N^3}(\xi,\zeta)$. On the other hand, for certain dimensions there are no semisimple Hopf algebras that are simultaneously noncommutative and non-cocommutative. In particular:
\begin{itemize}
    \item \textbf{Prime order $p$.}  
    Every semisimple Hopf algebra of dimension $p$ is isomorphic to the group algebra
    $\IC[\IZ_p]$ \cite{zhu1994hopf}.

    \item \textbf{Order $p^2$.}  
    Every semisimple Hopf algebra of dimension $p^2$ is isomorphic to either
    $\IC[\IZ_{p^2}]$ or $\IC[\IZ_p \times \IZ_p]$ \cite{masuoka1996hopfpn}.

    \item \textbf{Order $pq$.}  
    Every semisimple Hopf algebra of dimension $pq$, where $p$ and $q$ are distinct primes,
    is isomorphic to either the group algebra $\IC[G]$ or the dual group algebra $\IC^G$
    for a group $G$ of order $pq$ \cite{etingof1998semisimple}.
    For example, when $pq=6$, the possible Hopf algebras are
    $\IC[S_3]$, $\IC[\IZ_6]$, and their dual Hopf algebras
    $\IC^{S_3}$ and $\IC^{\IZ_6} \cong \IC[\IZ_6]$.
\end{itemize}

\section{$\Rep(H_8)$ categorical data}\label{app:grpthryH8}
$\Rep (H_8)$ is equivalent to the Tambara-Yamagami fusion category of $\IZ_2\times \IZ_2$, which describes the self-duality under gauging the $\IZ_2\times \IZ_2$ symmetry. The simple lines are the invertible $\IZ_2\times\IZ_2$ lines $1,\eta^o,\eta^e,\eta^o\eta^e$ and non-invertible $\CN$, with the following fusion rules,
\begin{equation}
    g\times h = gh,\  g\times \CN=\CN\times g = \CN,\ \CN\times \CN = 1+\eta^o+\eta^e+\eta^o\eta^e
\end{equation}
The $F$-symbols which are not 1 are,
\begin{equation}
    F^{g \CN h}_\CN = F^{\CN g \CN}_h = \chi(g,h),\ F^{\CN\CN\CN}_{\CN,(g,h)}=\frac{1}{2}\chi(g,h)
\end{equation}
where $\chi(g,h)=(-1)^{g_1h_1+g_2 h_2}$ is the diagonal bicharacter. The diagonal bicharacter corresponds to identify the dual symmetry with the original symmetry without any automorphism actions \cite{Thorngren2019Fusion}. In particular it differs from $\Rep(D_8)$, where $\Rep(D_8)$ has the off-diagonal bicharacter which corresponds to identify the dual symmetry with the original symmetry after swapping, and $\chi_{\text{off-diagonal}}(g,h)=(-1)^{g_1h_2+g_2 h_1}$.

$\Rep(H_8)$ is a group theoretical fusion category, which is Morita equivalent to $(D_8,\gamma)$, where $\gamma \in H^3(D_8,U(1))$. In particular, $\Rep(H_8)\cong \CC(D_8,\gamma,\IZ_2^s,1)$. This means one can start with a theory with $D_8$ symmetry together with the anomaly $\gamma$ and gauge the $\IZ_2^s$ anomaly-free subgroup to get the $\Rep(H_8)$ fusion category. Here, we parameterize the $D_8$ group via,
\begin{equation}
    D_8=\langle r,s| r^4=s^2=1,srs=r^{-1} \rangle
\end{equation}
The anomaly $\gamma$ is given by,
\begin{equation}
    \gamma(r^{i_1}s^{j_1},r^{i_2}s^{j_2},r^{i_3}s^{j_3}) = \exp\left(\frac{4\pi i}{4^2}(-1)^{j_1}i_1(i_2+(-1)^{j_2}i_3-[i_2+(-1)^{j_2}i_3]_4)\right)
\end{equation}
where $[\cdots]_4 =\cdots \mod{4}$. Recall that $H^3(D_8,U(1))=\IZ_4 \times \IZ_2^2$, $\gamma$ corresponds to the order 2 one in the $\IZ_4$. Another useful presentation of $D_8$ is,
\begin{equation}\label{eq:d8abc}
    D_8=\langle a,b,c| a^2=b^2=c^2=1,ab=ba,ac=cb \rangle
\end{equation}
and the isomorphism between the two presentation is,
\begin{equation}
    a=s,\ b=r^2s,\ c=r s
\end{equation}

\subsection{Gapped phases}\label{app:gappedphase}
The gapped phases of a fusion category symmetry are in one-to-one correspondence with its indecomposable module categories, and for group-theoretical fusion categories these are classified in \cite{ostrik2002module,natale2017equivalence}. For $\Rep(H_8) \cong \CC(D_8, \gamma, \IZ_2^s, 1)$, each indecomposable module category $\CM_{K,\psi_K}$ is labeled by an anomaly-free subgroup $K \subseteq D_8$ together with a 2-cocycle $\psi_K \in H^2(K, U(1))$. Its simple objects, which correspond to the ground states of the associated gapped phase, are pairs $([g], \pi_g)$ consisting of
\begin{enumerate}
    \item a double coset $[g] \in \IZ_2^s \backslash D_8 / K$ with representative $g$;
    \item an irreducible representation\footnote{In general $\pi_g$ can be projective, with the projective phase given in Lemma 3.3 of \cite{natale2017equivalence}.} $\pi_g$ of the stabilizer $\IZ_2^s \cap {}^g K \le \IZ_2^s$.
\end{enumerate}
Two module categories are equivalent, and therefore describe the same gapped phase, when they are related by conjugation in $D_8$ \cite{natale2017equivalence}. A module category with a single simple object corresponds to a $\Rep(H_8)$-symmetric phase with a unique ground state, i.e. a $\Rep(H_8)$ symmetry-protected trivial phase.

The $D_8,\gamma$ group contains the following anomaly-free subgroups, up to conjugation, ($H_1 \sim g H_2g^{-1},\exists g\in D_8$),
\begin{align}
    \langle 1 \rangle,\ \langle s \rangle,\ \langle sr \rangle,\ \langle r^2 \rangle,\ \langle s,r^2 \rangle,\ \langle sr,r^2 \rangle,  
\end{align}
The number of gapped phases of gauging related fusion categories are the same but the ground state degeneracies may be different. Note that for the last two cases, although the subgroups $K=\langle s,r^2\rangle, \langle sr,r^2\rangle$ are $\IZ_2\times \IZ_2$ and admit non-trivial 2-cocycle, due to the mixed anomaly $\gamma$, adding the non-trivial 2-cocycle doesn't lead to physically distinct phases \cite{wang2024gauging}. Due to the non-trivial 3-cocycle $\gamma$, this is also consistent with the fact that $\Rep(H_8)\cong \CC(D_8,\gamma,\IZ_2^s,1)$ has less symmetric gapped phases (6 gapped phases) compared with $\Rep(D_8) \cong \CC(D_8,1,\IZ_2^s,1)$ (11 gapped phases).

The 3-cocycle $\gamma$ gives the mixed anomaly between the subgroup $\langle s\rangle$ and $\langle sr,r^2\rangle$, (or $\langle sr\rangle$ and $\langle s,r^2\rangle$). Therefore, the DQCP between unbroken $\langle s\rangle$ and $\langle sr,r^2\rangle$ (or $\langle sr\rangle$ and $\langle s,r^2\rangle$ resp.) is described by this mixed anomaly. Under the gauging of $\IZ_2^s$, the DQCP is dual to the SSB transition between $\Rep(H_8)$ SPT and $\Rep(H_8)$ SSB phase (or $H_{13}$ to $H_{12}$ transition resp.).
\begin{equation}
    \begin{array}{c|c|c|c}
    \hline\hline
         (K,\psi_K) & \Rep(H_8) \text{ Module Categories}  & (D_8,\gamma) & \text{Ham} \\\hline
         (\langle1\rangle,1)   & 4, \{1+\eta^o\eta^e,\eta^o+\eta^e,\CN,\CN\}&  8,\ \langle1\rangle\text{ unbroken} & H_{15} \\\hline
         (\langle s\rangle,1)   & 5, \{\Rep(H_8)\}&  4,\ \langle s\rangle\text{ unbroken} & H_{11} \\\hline
         (\langle sr\rangle,1)   & 2,\{ 1+\eta^o\eta^e+\CN,\eta^o+\eta^e+\CN\}&  4,\ \langle sr\rangle\text{ unbroken} & H_{13}\\\hline
         (\langle r^2\rangle,1)   & 2, \{(1+\eta^o+\eta^e+\eta^o\eta^e)^*,2\CN\} &  4,\ \langle r^2\rangle\text{ unbroken}& H_{16} \\\hline
         (\langle s,r^2\rangle,1)   & 4,\{1+\eta^o,\eta^e+\eta^o\eta^e,\CN,\CN\} &  2,\ \langle s,r^2\rangle\text{ unbroken} & H_{12} \\\hline
         (\langle sr,r^2\rangle,1)   &1, \{1+\eta^o+\eta^e+\eta^o\eta^e+2\CN\}&  2,\ \langle sr,r^2\rangle\text{ unbroken} & H_{14} \\
         \hline\hline
    \end{array}
\end{equation}
The first 2 columns are from \cite{wang2024gauging}. We match the gapped phases with the Hamiltonians in the last column. These are read from the action of lines on the objects in the module categories, namely the symmetry action on the ground states of the symmetric gapped phase. For example, $\{ 1+\eta^o\eta^e+\CN,\eta^o+\eta^e+\CN\}$ are two ground states, say $\ket{v_1},\ket{v_2}$, and the symmetry actions are given by,
\begin{equation}
 \CN \ket{v_i} = \ket{v_1} + \ket{v_2},\quad \eta^o \eta^e \ket{v_{i}} = \ket{v_{i}},\quad \eta^o /\eta^e \ket{v_{i}} = \ket{v_{\bar{i}}}
\end{equation}
the action is matched with \eqref{eq:H13act}.

\subsection{Gauging the subgroup}\label{app:gaugesub}
The $\Rep(H_8)$ fusion category is dual to $(D_8,\gamma)$ via gauging the $\IZ_2^s$ subgroup, to be specific, 
\begin{equation}
    \Rep(H_8)
\overset{\xleftarrow{\text{gauge }\IZ_2^s\hphantom{\,1+\eta^o\eta^e}}}
{\xrightarrow[\text{gauge }1+\eta^o\eta^e\hphantom{\,\IZ_2^s}]{}}
(D_8,\gamma)
\end{equation}
When using the presentation \eqref{eq:d8abc}, gauging $\IZ_2^a$ gets back to $\Rep(H_8)$.

In general, one can start with the original symmetry fusion category $\CC$ and gauge the algebraic object $A$, the dual symmetry fusion category is given by the ${}_A\CC_A$-bimodule category. For $\CC=\VEC_G^\omega$, one can gauge the anomaly free subgroup $L$ with a possible discrete torsion $\psi\in H^2(L,U(1))$, which corresponds to $A=\sum_{\ell\in L}\ell$ with the multiplication $m: \ell_1\times \ell_2 \rightarrow \psi(\ell_1,\ell_2) \ell_1\ell_2$. Following \cite{ostrik2002module,natale2017equivalence,bartsch2024non,bartsch2024non2}, the objects in the bimodule category are labeled by $([g],\pi_g)$, where
\begin{enumerate}
    \item A double coset $[g]\in L\setminus G/L$ with representative $g$.
    \item An irreducible (projective) representation $\pi_g$ of $L_g$ with 2-cocycle, 
    \begin{equation}\label{eq:bimodule_phase}
        c_g(\ell_1,\ell_2)=\frac{\psi(\ell_1^g,\ell_2^g)}{\psi(\ell_1,\ell_2)}\frac{\omega(\ell_1,\ell_2,g)\omega(g,\ell_1^g,\ell_2^g)}{\omega(\ell_1,g,\ell_2^g)}
    \end{equation}
\end{enumerate}
where the double coset is defined by $LgL=\{\ell_1 g \ell_2|\ell_1,\ell_2\in L\}$, and $L_g\equiv L\cap gLg^{-1}$, $\ell^g=g\ell g^{-1}$. $([g],\pi_g)$ labels the symmetry line operators in the dual symmetry category. The quantum dimension of the simple line operator is given by,
\begin{equation}\label{eq:qdimg}
    \dim ([g],\pi_g) = |HgH| / |H| \dim(\pi_g).
\end{equation}

\paragraph{$\Rep(H_8)$ from gauging $\IZ_2^s$ in $\VEC_{D_8}^\gamma$}
In our case, $G=D_8$ with 3-cocycle $\gamma$, and we gauge the $\IZ_2^s$ subgroup to get $\Rep(H_8)$. The double cosets are,
\begin{equation}
    [1]\equiv \{1,s\},\quad [r^2]\equiv \{r^2s,r^2\},\quad [r]\equiv \{r,r^3,rs,r^3s\}.
\end{equation}
whose $L_g$ are $\IZ_2^s,\IZ_2^s,\IZ_1$. Since the $\gamma$ is the mixed anomaly between $\langle s \rangle$ and $\langle sr,r^2\rangle$, the \eqref{eq:bimodule_phase} is trivial, so the $L_g$s do not have projective representation, the simple objects in the bimodule are,
\begin{equation}
    1=([1],+),\quad \eta^o\eta^e= ([1],-),\quad \eta^o= ([r^2],+),\quad \eta^e= ([r^2],-),\quad \CN=([r],1)
\end{equation}
where $\pm$ is the trivial/sign representation of $\IZ_2^s$. It is easy to check $\eta^i$ are invertible and $\CN$ has quantum dimension 2 according to \eqref{eq:qdimg}. The multiplication of cosets is given by forming the ring $\IZ[L\setminus G /L]$ and $[g]\rightarrow x_g=\sum_{g_i\in[g]} g_i$, then the multiplication is, $x_{g_1}\cdot x_{g_2}$. In particular,
\begin{equation}
    x_r\cdot x_r = 4(x_1+x_{r^2})
\end{equation}
which means the product of double cosets $[r^2]$ with itself decomposes into the direct sum of double cosets $[1]$ and $[r^2]$. Next, we need to determine the irreps of these outcoming lines. As the outcoming lines are grouped by $[1]$ and $[r^2]$, we first look at $[1]$ with the representative $1$. The vector space is 2-dimensional, since $v_1\equiv (r,r^3)$ and $v_2\equiv (r^3,r)$ will yield $1$, but the $\IZ_2^s$ action permutes them, i.e. 
\begin{equation}
    v_1=(r,r^3)\xrightarrow{\text{diagonal $s$ action}} (s r,r^3 s) = (r^3 s,sr)\sim (r^3,r) = v_2.
\end{equation}
So the $\IZ_2^s$ action is represented by $\left(\begin{smallmatrix}
    0&1 \\ 1&0
\end{smallmatrix}\right)$ and can be diagonalized as in the space spanned by $v_1\pm v_2$ with eigenvalues $\pm1$. Therefore, the vector space decomposed as $([1],1)\oplus ([1],\text{sign})$, similar for $[r^2]$. Hence, the fusion rule of $\CN=([r],1)$ is given by,
\begin{equation}
    ([r],1)\times ([r],1) = ([1],1)\oplus ([1],\text{sign})\oplus ([r^2],1)\oplus ([r^2],\text{sign})
\end{equation}
which recovers the fusion rule of the non-invertible line $\CN$ in $\Rep(H_8)$.

Note that the mixed anomaly doesn't contribute to a non-trivial projective representation, the simple objects and the fusion rules of the bimodule category $\Rep(H_8)$ are the same as the $\Rep(D_8)$ which is resulting from gauging $\IZ_2^s$ in $D_8$ without anomaly. The mixed anomaly here only changes the $F$-symbols of the bimodule category. This is different from the case of gauging $\IZ_2^a \times \IZ_2^b$ in $\IZ_2^a\times \IZ_2^b \times \IZ_2^c$ with $a\cup b\cup c$ anomaly, where the anomaly leads to non-trivial projective representation, therefore changes the fusion rule from $\IZ^3_2$ to $\Rep(D_8)$.

\paragraph{$\Rep(H_8)$ from gauging $\IZ_2^a\times \IZ_2^b$ in $\VEC_{D_8}^\gamma$}
Using the presentation \eqref{eq:d8abc},
\begin{equation}
    D_8=\langle a,b,c| a^2=b^2=c^2=1,ab=ba,ac=cb \rangle
\end{equation}
The $\gamma$ is the mixed anomaly between the $\IZ_2=\langle abc\rangle$ subgroup and the $\IZ_2^a\times \IZ^b_2=\langle a,b\rangle$ subgroup. One can gauge the $\IZ_2^a\times \IZ^b_2$ to get the $\Rep(H_8)$ fusion category. In general, one can add a discrete torsion $ H^2(\IZ_2^a\times \IZ^b_2,U(1))$ before gauging, but with or without discrete torsion lead to the identical result due to the absorption of the discrete torsion by the mixed anomaly $\gamma$.

To be specific, the double cosets are 
\begin{equation}
    [1,a,b,ab] , \quad [abc,c,ac,bc]
\end{equation}
The mixed anomaly $\gamma$ will activate a non-trivial projective representation for the second double coset, while trivial for the first one. Therefore, the simple lines are,
\begin{equation}
    1=([1],++),\ \eta^o=([1],-+),\ \eta^e=([1],+-),\ \eta^o\eta^e=([1],--),\ \CN=([abc],\pi)
\end{equation}
where $\pm\pm$ are linear representations of $\IZ_2^a\times \IZ^b_2$, while $\pi$ is the projective representation, $\pi(a)=\sigma^x, \pi(b)=\sigma^z$, $\pi\otimes \pi = ++ \oplus -+\oplus +-\oplus --$.

\section{Data of Drinfeld center of $\Rep(H_8)$}
\label{app:Drinfeld_center}
The identification of anyons in different descriptions up to anyon permutation symmetry is shown in \tabref{tab:anyon_id}. Where $\CZ(\TY)$ refers to $\CZ(\TY(\IZ_2^2,\chi_\text{diag},+1))$ and the anyons are obtained using the method from the study of center of graded fusion category \cite{Gelaki:2009blp}. $D(D_8)^\gamma$ refers to the twisted quantum double of $D_8$ with 3-cocycle twist $\gamma$. 
\begin{table}[h!]
    \centering
$\begin{array}{|c|c|c|c|c|}
\hline
\CZ\text{(TY)} 
& \CZ(\text{Ising}^2)/(1+\psi\bar\psi\psi\bar\psi) 
& D(D_8)^\gamma
& \text{spin} 
& \text{quantum dim.} \\
\hline
X_{1,1} & 1 & [1,\, \chi_{1}^{D_8}] & 1 & 1 \\ \hline
X_{2,1} & \bar{\psi}_1 & [r^2,\, \chi_{3}^{D_8}] & -1 & 1 \\ \hline
X_{3,1} & \bar{\psi}_2 & [r^2,\, \chi_{4}^{D_8}] & -1 & 1 \\ \hline
X_{4,1} & \psi_1 \bar{\psi}_2 & [1,\, \chi_{4}^{D_8}] & 1 & 1 \\ \hline
X_{1,-1} & \psi_1 \bar{\psi}_1 & [1,\, \chi_{3}^{D_8}] & 1 & 1 \\ \hline
X_{2,-1} & \psi_1 & [r^2,\, \chi_{1}^{D_8}] & -1 & 1 \\ \hline
X_{3,-1} & \psi_2 & [r^2,\, \chi_{2}^{D_8}] & -1 & 1 \\ \hline
X_{4,-1} & \psi_1 \psi_2 & [1,\, \chi_{2}^{D_8}] & 1 & 1 \\ \hline
Y_{1,2} & \sigma_1 \bar{\sigma}_1 & [s,\, \chi_{\{0,1\}}^{\mathbb{Z}_2\times\mathbb{Z}_2}] & 1 & 2 \\ \hline
Y_{1,3} & \sigma_2 \bar{\sigma}_2 & [s,\, \chi_{\{0,0\}}^{\mathbb{Z}_2\times\mathbb{Z}_2}] & 1 & 2 \\ \hline
Y_{1,4} & A & [1,\, \chi_{5}^{D_8}] & 1 & 2 \\ \hline
Y_{2,3} & B & [r^2,\, \chi_{5}^{D_8}] & 1 & 2 \\ \hline
Y_{2,4} & \sigma_2 \bar{\sigma}_2 \bar{\psi}_1 & [s,\, \chi_{\{1,0\}}^{\mathbb{Z}_2\times\mathbb{Z}_2}] & -1 & 2 \\ \hline
Y_{3,4} & \sigma_1 \bar{\sigma}_1 \bar{\psi}_2 & [s,\, \chi_{\{1,1\}}^{\mathbb{Z}_2\times\mathbb{Z}_2}] & -1 & 2 \\ \hline
Z_{1,1} & \bar{\sigma}_1 \bar{\sigma}_2 & [r,\, \chi_{1}^{\mathbb{Z}_4}] & e^{-i\pi /4} & 2 \\ \hline
Z_{2,1} & \sigma_1 \bar{\sigma}_2 & [rs,\, \chi_{\{1,0\}}^{\mathbb{Z}_2\times\mathbb{Z}_2}] & 1 & 2 \\ \hline
Z_{3,1} & \bar{\sigma}_1 \sigma_2 & [rs,\, \chi_{\{1,1\}}^{\mathbb{Z}_2\times\mathbb{Z}_2}] & 1 & 2 \\ \hline
Z_{4,1} & \sigma_1 \sigma_2 & [r,\, \chi_{2}^{\mathbb{Z}_4}] & e^{i \pi/4} & 2 \\ \hline
Z_{1,-1} & \bar{\sigma}_1 \bar{\sigma}_2 \psi_1 & [r,\, \chi_{3}^{\mathbb{Z}_4}] & e^{i 3\pi/4} & 2 \\ \hline
Z_{2,-1} & \sigma_1 \bar{\sigma}_2 \bar{\psi}_1 & [rs,\, \chi_{\{0,0\}}^{\mathbb{Z}_2\times\mathbb{Z}_2}] & -1 & 2 \\ \hline
Z_{3,-1} & \bar{\sigma}_1 \sigma_2 \psi_1 & [rs,\, \chi_{\{0,1\}}^{\mathbb{Z}_2\times\mathbb{Z}_2}] & -1 & 2 \\ \hline
Z_{4,-1} & \sigma_1 \sigma_2 \bar{\psi}_1 & [r,\, \chi_{0}^{\mathbb{Z}_4}] & e^{-i 3\pi/4} & 2 \\ 
\hline
\end{array} $
    \caption{Equivalence between the different notations labeling anyons of $\CZ(\Rep(H_8))$.}
    \label{tab:anyon_id}
\end{table}

\section{Lagrangian Algebra of $\mathcal{A}_1$ to $\mathcal{A}_6$}
\label{app:LagrangianAlgebra}

In this section, we give a `physical' derivation of the data describing the Lagrangian Algebra of $\mathcal{A}_1$ to $\mathcal{A}_6$. In particular, we are going to derive the $M$-$3j$ symbol defined graphically as
\begin{figure}[th]
\begin{center}
\includegraphics[width= 0.4\textwidth]{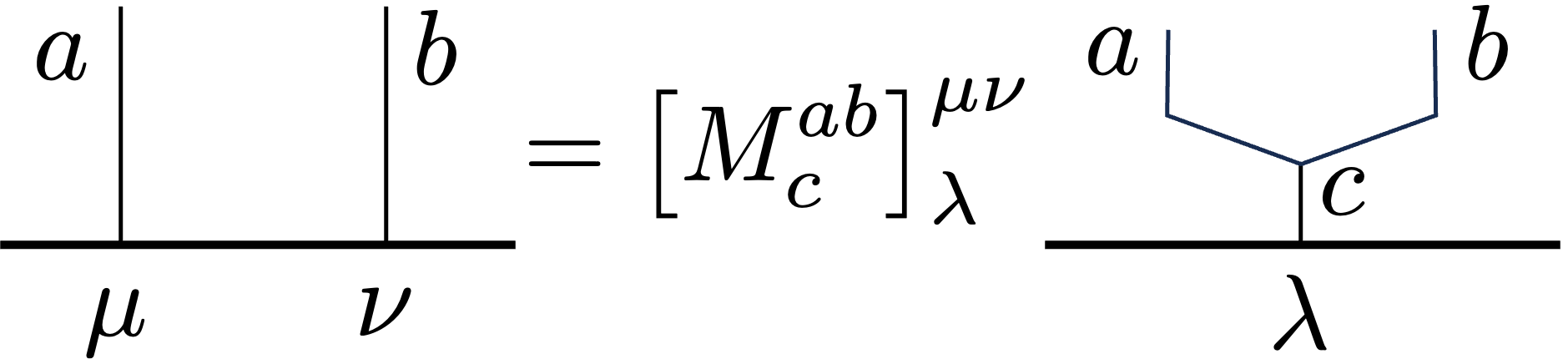}
\label{fig:Lagrange}
\end{center}
\end{figure}

where the thick bottom line indicates the boundary described by the Lagrangian Algebra $\mathcal{A} = \bigoplus_a n_a a$ and the vertical lines are labeled by objects in the Lagrangian algebra $a \in \mathcal{A}$. $\mu$, $\nu$ and $\lambda$ labels the internal dimensions of $a$, $b$ and $c$. $\mu = 0,1,...,n_a-1$; $\nu = 0,1,...,n_b-1$; $\lambda = 0,1,...,n_c-1$.

As discussed in section~\ref{sec:boundary_sym}, $\mathcal{A}_5$ and $\mathcal{A}_6$ leads to a Vec($D_8$) symmetry with anomaly. That means, the condensation in $\mathcal{A}_5$ and $\mathcal{A}_6$ are of the symmetry charges of the group $D_8$ (including four one-dimensional irrep and one two-dimensional irrep). The $M$-$3j$ symbol of $\mathcal{A}_5$ and $\mathcal{A}_6$ are hence the $3j$ symbol (the Clebsch-Gordan coefficient) of the irreps of $D_8$. Since this is standard data, we will not derive it here. We list part of the data for reference. 
$$
\begin{array}{l}
M^{1,A}_A = \begin{pmatrix} 1 & 0 \\ 0 & 1 \end{pmatrix}, \ M^{\psi_1\bar{\psi}_1,A}_A = \begin{pmatrix} 1 & 0 \\ 0 & -1 \end{pmatrix}, \ M^{\psi_1\bar{\psi}_2,A}_A = \begin{pmatrix} 0 & 1 \\ 1 & 0 \end{pmatrix}, \ M^{\psi_1\psi_2,A}_A = \begin{pmatrix} 0 & -1 \\ 1 & 0 \end{pmatrix},\\
M^{A,A}_1 = \frac{1}{\sqrt{2}} \begin{pmatrix} 1 & 0 \\ 0 & 1 \end{pmatrix}, \ M^{A,A}_{\psi_1\bar{\psi}_1} = \frac{1}{\sqrt{2}}\begin{pmatrix} 1 & 0 \\ 0 & -1 \end{pmatrix}, \ M^{A,A}_{\psi_1\bar{\psi}_2} = \frac{1}{\sqrt{2}}\begin{pmatrix} 0 & 1 \\ 1 & 0 \end{pmatrix}, \ M^{A,A}_{\psi_1\psi_2} = \frac{1}{\sqrt{2}}\begin{pmatrix} 0 & -1 \\ 1 & 0 \end{pmatrix}
\end{array}
$$
The row and column indices $\mu, \nu =0,1$ are the internal dimensions of $A$.

$\mathcal{A}_1$ to $\mathcal{A}_4$, on the other hand, all lead to the Rep($H_8$) symmetry. Therefore, their $M$-$3j$ symbols do not directly correspond to the $3j$ symbol of a group. Since the four Lagrangian Algebras all have the same structure, we will illustrate the derivation more carefully for $\mathcal{A}_1$ and comment on the difference of $\mathcal{A}_2$, $\mathcal{A}_3$ and $\mathcal{A}_4$. Recall that
$$
\mathcal{A}_1 = 1 \oplus \psi_1\bar{\psi}_1 \oplus \sigma_1\bar{\sigma}_1 \oplus \sigma_2\bar{\sigma}_2 \oplus A
$$

All condensed anyons in $\mathcal{A}_1$ condense with multiplicity $1$, indicating that only one dimension within their internal space is condensable. $\psi_1\bar{\psi}_1$ is an abelian anyon, therefore its condensation is straight-forward. Its fusion rule follows directly from the bulk
$$
\psi_1\bar{\psi}_1 \otimes \psi_1\bar{\psi}_1 = 1, \psi_1\bar{\psi}_1 \otimes \sigma_i\bar{\sigma}_i = \sigma_i\bar{\sigma}_i, \  \psi_1\bar{\psi}_1 \otimes A = A
$$
$\sigma_1\bar{\sigma}_1$, $\sigma_2\bar{\sigma}_2$, and $A$ are all two-dimensional anyons. Their condensable dimensions can be understood as follows. If only  $\psi_1\bar{\psi}_1$ is condensed, the bulk topological order  decouples into two copies of Toric Code (i.e. $Z_2$ gauge theory). $\sigma_1\bar{\sigma}_1$ splits into $e_1$ and $m_1$, the two bosonic anyons in the first Toric Code. $\sigma_2\bar{\sigma}_2$ splits into $e_2$ and $m_2$, the two bosonic anyons in the second Toric Code.
Similarly, $A$ splits into $e_1e_2$ and $m_1m_2$. To further condense $\sigma_1\bar{\sigma}_1$,$\sigma_2\bar{\sigma}_2$ and $A$ in $\mathcal Z(\Rep(H_8))$, only one anyon from each of the splittings may condense. WLOG, we can set the condensable dimension of these three anyons to be $|e_1\rangle$, $|e_2\rangle$, and $|e_1e_2\rangle$. From the condensable dimensions, we can derive their multiplication within $\mathcal{A}_1$ (which we denote with $\times$ in contrast to the fusion $\otimes$ in the bulk)
$$
\sigma_1\bar{\sigma}_1 \times \sigma_2\bar{\sigma}_2 = A, \ \sigma_1\bar{\sigma}_1 \times A =  \sigma_2\bar{\sigma}_2, \ \sigma_2\bar{\sigma}_2 \times A =  \sigma_1\bar{\sigma}_1
$$
The multiplication of a pair of $\sigma_1\bar{\sigma}_1$ ($\sigma_2\bar{\sigma}_2$, $A$), on the other hand, requires a bit more work. While within the two copies of Toric Code, $e_1$, $e_2$, $e_1e_2$ are all order two anyons, this is not true in the original topological order before the condensation of $\psi_1\bar{\psi}_1$. In fact, since $e_1$ is mapped to $m_1$ under $\sigma_1\sigma_2$, the correct multiplication rule can be found from the following calculation.
$$
|e_1\rangle |e_1\rangle = \frac{1}{\sqrt{2}} \left(\frac{1}{\sqrt{2}}|e_1\rangle |e_1\rangle + \frac{1}{\sqrt{2}}|m_1\rangle |m_1\rangle\right) + \frac{1}{\sqrt{2}}\left(\frac{1}{\sqrt{2}}|e_1\rangle |e_1\rangle - \frac{1}{\sqrt{2}}|m_1\rangle |m_1\rangle\right)
$$
The first term $\left(|e_1\rangle |e_1\rangle + |m_1\rangle |m_1\rangle\right)$ represents the true identity anyon, while the second term $\left(|e_1\rangle |e_1\rangle - |m_1\rangle |m_1\rangle\right)$ represents the $\psi_1\bar{\psi}_1$ anyon because it carries a $-1$ change under $\sigma_1\sigma_2$. Therefore, the fusion of two copies of the (condensable dimension) of $\sigma_1\bar{\sigma}_1$ (and similarly $\sigma_2\bar{\sigma}_2$ and $A$) within $\mathcal{A}_1$ is given by
$$
\sigma_1\bar{\sigma}_1 \times \sigma_1\bar{\sigma}_1 = \sigma_2\bar{\sigma}_2 \times \sigma_2\bar{\sigma}_2 = A \times A = \frac{1}{\sqrt{2}} \left(1 \oplus  \psi_1\bar{\psi}_1 \right)
$$

Written as $M$-$3j$ symbols, we have
$$
\begin{array}{l}
M^{1,x}_x = 1, \ M^{\psi_1\bar{\psi}_1, \psi_1\bar{\psi}_1}_1 = 1, \ M^{\psi_1\bar{\psi}_1, \sigma_i\bar{\sigma}_i}_{\sigma_i\bar{\sigma}_i} = 1, \ M^{\psi_1\bar{\psi}_1, A}_{A} = 1, \ M^{\sigma_i\bar{\sigma}_i, \sigma_{\bar{i}}\bar{\sigma}_{\bar{i}}}_A = 1, \ M^{\sigma_i\bar{\sigma}_i, A}_{\sigma_{\bar{i}}\bar{\sigma}_{\bar{i}}} = 1, \\
M^{\sigma_i\bar{\sigma}_i,\sigma_i\bar{\sigma}_i}_1 = M^{\sigma_i\bar{\sigma}_i,\sigma_i\bar{\sigma}_i}_{\psi_1\bar{\psi}_1} = M^{A,A}_1 = M^{A,A}_{\psi_1\bar{\psi}_1} = \frac{1}{\sqrt{2}}
\end{array}
$$

For $\mathcal{A}_2$, the $M$-$3j$ symbols are similar, with some important differences. $\mathcal{A}_2$ contains the same set of anyons as $\mathcal{A}_1$ except that $A$ is replace by $B$. $B$ differs from $A$ by a single $\psi$. WLOG, we can assume that the condensed dimension in $\sigma_1\bar{\sigma}_1$ changes from $e_1$ to $m_1$ while the condensed dimension in $\sigma_2\bar{\sigma}_2$ remains to be $e_2$. The condensed dimension in $B$ is then $m_1e_2$. The product among the three still takes the form
$$
\sigma_1\bar{\sigma}_1 \times \sigma_2\bar{\sigma}_2 = B, \ \sigma_1\bar{\sigma}_1 \times B =  \sigma_2\bar{\sigma}_2, \ \sigma_2\bar{\sigma}_2 \times B =  \sigma_1\bar{\sigma}_1
$$
The product of a pair of $\sigma_1\bar{\sigma}_1$ changes because
$$
|m_1\rangle|m_1\rangle = \frac{1}{\sqrt{2}} \left(\frac{1}{\sqrt{2}}|e_1\rangle|e_1\rangle + \frac{1}{\sqrt{2}} |m_1\rangle|m_1\rangle\right) - \frac{1}{\sqrt{2}}\left(\frac{1}{\sqrt{2}}|e_1\rangle|e_1\rangle - \frac{1}{\sqrt{2}}|m_1\rangle|m_1\rangle\right)
$$
Therefore,
$$
\sigma_1\bar{\sigma}_1 \times \sigma_1\bar{\sigma}_1 = \frac{1}{\sqrt{2}} \left( 1 - \psi_1\bar{\psi}_1 \right)
$$
while the fusion of a pair of $\sigma_2\bar{\sigma}_2$ remains the same
$$
\sigma_2\bar{\sigma}_2 \times \sigma_2\bar{\sigma}_2  = \frac{1}{\sqrt{2}} \left( 1 + \psi_1\bar{\psi}_1 \right)
$$
Moreover, since 
$$
\left(\psi_1\bar{\psi}_1 \times \sigma_1\bar{\sigma}_1 \right) \times \sigma_1\bar{\sigma}_1 = \psi_1\bar{\psi}_1 \times \left(\sigma_1\bar{\sigma}_1  \times \sigma_1\bar{\sigma}_1\right)  = \frac{1}{\sqrt{2}} \left( \psi_1\bar{\psi}_1-1 \right)
$$
we have
$$
\psi_1\bar{\psi}_1 \times \sigma_1\bar{\sigma}_1 = - \sigma_1\bar{\sigma}_1
$$
while the fusion of $\psi_1\bar{\psi}_1$ and $\sigma_2\bar{\sigma}_2$ remains the same
$$
\psi_1\bar{\psi}_1 \times \sigma_2\bar{\sigma}_2 = \sigma_2\bar{\sigma}_2
$$
The $M$-$3j$ symbols differ from those for $\CA_1$ in the following terms
$$
M^{\psi_1\bar{\psi}_1, \sigma_1\bar{\sigma}_1}_{\sigma_1\bar{\sigma}_1} = -1, \ M^{\sigma_1\bar{\sigma}_1,\sigma_1\bar{\sigma}_1}_{\psi_1\bar{\psi}_1} = -\frac{1}{\sqrt{2}}
$$

\section{Rigorous Derivation of the Dual Lattice Model Symmetry}\label{app:LatticeDualSymmetry}
In this section, we will derive the symmetry of the dual lattice model obtained from gauging $\mathbb{Z}_2^{oe}$ in a more rigorous way. 

The Kramers-Wannier transformation operator $\Noe$ that implements the gauging can be defined by the following two conditions: 
\begin{enumerate}
    \item $\Noe(X_n,Z_nZ_{n+1})=(Z_nZ_{n+1},X_{n+1})\Noe$, where periodic boundary condition is assumed on both sides of $\Noe$, i.e. $(X,Z)_{2N+n}\equiv (X,Z)_n$.  
    \item $\Noe[\ket{00\cdots 0}+\ket{11\cdots 1}]/2=\ket{++\cdots +}$. 
\end{enumerate}
In particular, the first condition implies $\Noe \prod_n X_n=\Noe$ and thus $\Noe$ annihilates all $\mathbb{Z}_2^{oe}$ odd states. 

Let $H$ be a $\mathbb{Z}_2^{oe}$ symmetric Hamiltonian and suppose $H$ is a sum of local operators. We claim there is a unique local Hamiltonian $\tilde H$ satisfying 
\begin{align}
    \tilde H\Noe =\Noe H,  
\end{align}
and this $\tilde H$ is the Hamiltonian after gauging. The existence of such an $\tilde H$ is relatively easy to see: Without loss of generality, we can assume that each local term in $H$ is $\IZ_2^{oe}$ symmetric (if not, symmetrize it). Such a local term can be locally generated by $X$ and $ZZ$ operators, and thus can be mapped to a local term of $\tilde H$ by the operator mapping rule of $\Noe$. To prove the uniqueness of $\tilde H$, suppose $\tilde H_1$ and $\tilde H_2$ both satisfy the above equation and are both sums of local terms. 
Let $\eta_X:=\prod_nX_n$. Since the image of $\Noe$ contains all $\eta_X$ even states, $\Delta \tilde H:= \tilde H_1-\tilde H_2$ must act trivially on all $\eta_X$ even states. Therefore, $\Delta \tilde H(1+\eta_X)/2=0$, or equivalently, $\Delta \tilde H=-\Delta \tilde H\eta_X$. This implies $\Delta \tilde H=0$ by the following lemma. 
\begin{lemma}\label{thm:LocalityLemma}
    \textbf{(Pauli String is not Compatible with Locality)} Let $\{ O_n \}$ and $\{ O'_n \}$ be sets of local operators such that $O_n$ and $O'_n$ are supported near the site $n$. If 
    \begin{align}
        \eta_X\sum_nO_n=\sum_nO'_n, 
        \label{eq:PauliLocalityLemmaCondition}
    \end{align}
    then $\sum_nO_n=\sum_nO'_n=0$. 

\end{lemma}
\begin{proof}
    Consider partial tracing Eq.\,\ref{eq:PauliLocalityLemmaCondition} over two sites $j$ and $j'$ that are sufficiently far away. The left-hand side of the equation will vanish. This is because for each $n$, at least one of the two sites $j$ and $j'$ is away from the support of $O_n$, and tracing over that site kills $\eta_XO_n$. 
    On the right-hand side, however, all terms away from sites $j$ and $j'$ are unaffected. 
    This suffices to prove that $\sum_n  O_n'$ commutes with any local operator and hence must be a number. Moreover, this number must be zero since its partial trace vanishes. 
\end{proof}
\noindent Our proof for the vanishing of $\Delta \tilde H$ can be summarized as the following corollary. 
\begin{corollary}
    Let $K$ be a sum of local operators. $K\Noe=0$ implies $K=0$.
\end{corollary}
\begin{proof}
    Since the image of $\Noe$ contains all $\eta_X$ even states, $K\Noe=0$ implies $K(1+\eta_X)/2=0$. Equivalently, $K=-\eta_XK$ and it follows from the above lemma that $K=0$. 
\end{proof}

Now let us consider the symmetry of $\tilde H$. From $\eta_X\Noe=\Noe$, we have 
\begin{align}
    (\eta_X\tilde H\eta_X-\tilde H)\Noe=0. 
\end{align}
Since $\eta_X\tilde H\eta_X-\tilde H$ is a sum of local operators, it must vanish by the above corollary. Therefore, $\tilde H$ has a $\mathbb{Z}_2^X$ symmetry generated by $\eta_X$. 

If $H$ has the lattice ${\rm Rep}(H_8)$ symmetry introduced in the main text, $\tilde H$ will also have some additional symmetries which we now determine. Let $\eta_Z:=\prod_n Z_n$. We have $\eta_Z\Noe=\Noe \eta_o$. It follows that $(\eta_Z \tilde H\eta_Z-\tilde H)\Noe=0$. By the above corollary, $\eta_Z \tilde H\eta_Z-\tilde H=0$ and hence $\tilde H$ has a $\mathbb{Z}_2^Z$ symmetry generated by $\eta_Z$. Considering $\eta_e$ will lead to the same result. 

Finally, let us consider the fate of the noninvertible symmetry generator $\mc N$. We define $V:=T\prod_n\mc H_n$ where $T$ is the translation operator mapping $(X_n,Z_n)$ to $(X_{n+1},Z_{n+1})$ and $\mc H_n$ is the Hadamard gate acting on the site $n$. One can verify
\begin{align}
    (1+\eta_X)V\Noe=\Noe\mc N
\end{align}
by checking their actions on all $X_n$ and on the state $\ket{00\cdots 0}$. Now apply both sides of the above equation to $H$. The left-hand side can be computed as 
\begin{align}
    (1+\eta_X)V\Noe H=V(1+\eta_Z)\Noe H=V\tilde H (1+\eta_Z)\Noe,  
\end{align}
while the right-hand side is 
\begin{align}
    \Noe\mc N H=\tilde H\Noe \mc N=\tilde H (1+\eta_X)V\Noe=\tilde HV(1+\eta_Z)\Noe. 
\end{align}
These imply 
\begin{align}
    (V^{-1}\tilde H V-\tilde H)(1+\eta_Z)(1+\eta_X)=0. 
\end{align}
By a proof similar to that for Lemma \ref{thm:LocalityLemma}, we must have $V^{-1}\tilde H V-\tilde H=0$. 

To summarize, we have proved that the dual Hamiltonian $\tilde H$ commutes with $\eta_X$, $\eta_Z$, and $V:=T\prod_n \mc H_n$.

\bibliography{Refs.bib}

\end{document}